\documentclass[twoside,11pt]{article}

\usepackage{blindtext}
\usepackage{comment}
\usepackage{bbm}
\usepackage{jmlr2e}
\usepackage{xcolor}
\usepackage{multirow}

\newcommand\ind{\protect\mathpalette{\protect\independenT}{\perp}}
\def\independenT#1#2{\mathrel{\rlap{$#1#2$}\mkern2mu{#1#2}}}

\usepackage{mathtools}
\DeclarePairedDelimiter\ceil{\lceil}{\rceil}
\DeclarePairedDelimiter\floor{\lfloor}{\rfloor}

\newtheorem{assumption}{Assumption}

\usepackage{lastpage}
\jmlrheading{23}{2022}{1-\pageref{LastPage}}{1/21; Revised 5/22}{9/22}{21-0000}{Author One and Author Two}

\ShortHeadings{}{}
\firstpageno{1}

\begin{document}
	
	\title{A causal fused lasso for interpretable heterogeneous treatment effects estimation}
	
	\author{\name Oscar Hernan Madrid Padilla \email oscar.madrid@stat.ucla.edu\\
		\addr Department of Statistics and Data Sciences\\
		Univeristy of California, Los Angeles\\
		Los Angeles, CA 90095.
		\AND
		\name Yanzhen Chen\thanks{These authors contributed equally} \email imyanzhen@ust.hk \\
		\addr Department of ISOM\\
		Hong Kong University of Science and Technology\\
		Clear Water Bay, Hong Kong.\\
		\AND
		\name Carlos Misael Madrid Padilla \thanks{These authors contributed equally}  \email carlosmisael@wustl.edu \\
		\addr Department of Statistics and Data Science\\
		Washington University in St Louis\\
		St Louis, MO 63130. 
		\AND
		\name Gabriel Ruiz\email ruizg@ucla.edu \\
		\addr Department of Statistics\\
		Univeristy of California, Los Angeles\\
		Los Angeles, CA 90095.}
	
	\editor{}
	
	\maketitle
	
	\begin{abstract}
		We propose a novel method for estimating heterogeneous treatment effects based on the fused lasso. By first ordering samples based on the propensity or prognostic score, we match units from the treatment and control groups. We then run the fused lasso to obtain piecewise constant treatment effects with respect to the ordering defined by the score. 
		Similar to the existing methods based on discretizing the score, our methods yield interpretable subgroup effects. However, existing methods fixed the subgroup a priori, but our causal fused lasso forms data-adaptive subgroups. 
		We show that the estimator consistently estimates the treatment effects conditional on the score under very general conditions on the covariates and treatment.  We demonstrate the performance of our procedure using extensive experiments that show that it can be interpretable and competitive with state-of-the-art methods. 
	\end{abstract}
	
	\begin{keywords}
		Causality, propensity score, prognostic score, total variation 
	\end{keywords}

	\section{Introduction}
	
	Causal inference focuses on the causal relationships between covariates and their outcomes, yet is deeply rooted in and advances the way we understand the world. Applications of causal inference include   medical tests and personalized medicine \citep{zhao2017selective,imbens2015causal}, and economic and public policy evaluations \citep{angrist2008mostly, ding2016randomization,shalit2017estimating}. Our paper proposes a powerful tool for estimating heterogeneous treatment effects under general  assumptions.

	
	We adopt the potential outcomes framework \citep{neyman1923applications,rubin1974estimating}, where each subject has an observed outcome variable caused by a treatment indicator and a set of other predictors. The main challenge in estimating heterogeneous treatment effects stems from the fact that  for any given subject we  only observe the outcome under treatment or control but not both. This is also summarized as a ``missing data" problem \citep{ding2018causal}. To address this problem, some pioneering works relied on matching via pre-specified groups  \citep{assmann2000subgroup,pocock2002subgroup,cook2004subgroup}. However, these approaches are sensitive to subject grouping which are largely selected using domain expertise. Our approach, in contrast, does not rely on pre-specified subgroups.  Rather, it simultaneously identifies the subgroups and their associated treatment effects.  While our approach is similar in spirit to \cite{abadie2018endogenous}, it is more flexible allowing for discontinuous treatment effect functions.


	Our estimator integrates the merits of similarity scores with the fused lasso method using a simple two-step approach. First, we construct a statistic for each unit and sort observations according to it. The intuition of this step is to summarize the similarities among units using the statistics constructed. In the second step, we perform  matching of units of  the treatment and control groups based on the statistics generated in the first step. The differences in observed outcomes between the matched pairs  guide the fused lasso method, which is a one-dimensional nonparametric regression method as introduced in \cite{mammen1997locally} and \cite{tibshirani2005sparsity}, to estimate the treatment effects for different units.  A key difference between our causal fused lasso approach and the usual fused lasso is that in the latter, there is a given input signal $y$ and an ordering associated to it. In contrast, in causal inference there are no measurements available associated with the individual effects, which is the reason behind our two-step approach.
	
	To be more specific, for the first step of our proposed method, we capture similarities among units using widely adopted statistics such as the propensity score  (\cite{rosenbaum1983central,rosenbaum1984reducing}),  and the prognostic score (see e.g  \cite{hansen2008prognostic,abadie2018endogenous}). For the former,  we   fit a parametric  model such as  logistic regression. For the prognostic score method,  we  regress the outcome on the covariates using the control group data only. 
	
	Despite being a simple, our method enjoys the following properties:
	
	\begin{enumerate}
		
		\item  From a theoretical perspective, we establish that our method consistently estimates treatment effects under minimal assumptions. These assumptions include a general random design for covariates and bounded variation of the conditional mean of the outcome, conditional on both the subgroup and treatment assignment.
		
		\item  Our estimator is computationally efficient, with overall complexity on the order of 
		$O(nd + n^2)$, where $n$ is the number of units and  $d$ is the number of covariates. The $nd$
		term corresponds to fitting a linear regression model and can be considered linear in 
		$n$ when 
		$d$ is small. The $n^2$ term, which may be dominant, arises from a matching step that requires computing pairwise scores between all units. This is comparable to the complexity of a 
		$K$-nearest neighbors algorithm. However, this computational cost can be significantly reduced through parallel computing.
		
		\item Unlike many nonparametric methods, our estimator offers interpretability. Moreover, experimental results demonstrate that it either outperforms or matches state-of-the-art approaches in estimating heterogeneous treatment effects, as measured by mean squared error for estimating the conditional treatment effects.
	\end{enumerate}

	\subsection{Previous work}
	
	
	A substantial body of statistical research has focused on estimating heterogeneous treatment effects, with many studies building on the seminal Bayesian additive regression tree (BART) framework introduced by \cite{chipman2010bart}. The core idea behind BART-based methods is to impose the prior commonly used in BART  on both the regression function of the control group and that of the treatment group, as exemplified in \cite{hill2011bayesian, green2012modeling, hill2013assessing}. More recently, \cite{hahn2020bayesian} introduced a BART-based approach designed to handle small effect sizes and confounding by observables. Beyond BART, other Bayesian methods include the linear model prior proposed by \cite{heckman2014treatment} and the Bayesian nonparametric framework developed in \cite{taddy2016nonparametric}.
	

	In a separate line of research, tree-based regression methods have been developed to estimate heterogeneous treatment effects, with regression trees emerging as a key approach. This direction was initiated by \cite{su2009subgroup}, who proposed an estimator based on the widely used CART method \cite{breimanclassification}. More recently, \cite{athey2016recursive} and  \cite{wager2018estimation} extended this work by employing the random forest framework from \cite{breiman2001random} to construct estimators. A significant contribution of \cite{wager2018estimation} was the introduction of an inferential framework for treatment effect estimates using the infinitesimal jackknife. Building on this foundation, \cite{athey2019generalized} developed a generalized random forest approach, improving robustness in practical applications compared to the estimator from \cite{wager2018estimation}.

Beyond regression tree-based methods, various machine learning approaches have been explored for estimating heterogeneous treatment effects. \cite{crump2008nonparametric} introduced nonparametric tests to detect treatment effect heterogeneity. \cite{imai2013estimating} proposed a method that integrates hinge loss \citep{wahba2002soft} with lasso regularization \citep{tibshirani2005sparsity} to enhance estimation accuracy. \cite{tian2014simple} developed an approach capable of handling a high-dimensional feature space while capturing interactions between treatment and covariates.

Other contributions include \cite{weisberg2015post}, who designed a method based on variable selection, and \cite{taddy2016nonparametric}, who developed Bayesian nonparametric techniques applicable to both linear regression and tree-based models. \cite{syrgkanis2019machine} introduced a flexible framework that can incorporate any machine learning method and, under valid instrumental variables, account for unobserved confounders. From a theoretical perspective, \cite{gao2020minimax} analyzed the fundamental limits of estimating heterogeneous treatment effects under H\"{o}lder smoothness conditions.

A different line of work focuses on meta-learning strategies in causal inference. \cite{kunzel2019metalearners} proposed a meta-learner framework that can leverage various estimators and is particularly effective when treatment group sizes are highly imbalanced. We note that the term meta-learning is sometimes used differently in the broader machine learning literature. For example, it can refer to transfer learning settings where metadata across tasks is used to inform model training \cite{vanschoren2019meta}. Our use of meta-learning follows the convention in causal inference, where meta-algorithms aggregate supervised learners to estimate treatment effects.

For comprehensive reviews of testing procedures and estimation methods for individual treatment effects, see \cite{willke2012concepts} and \cite{caron2022estimating}.

In this paper, we propose a two-step method for estimating heterogeneous treatment effects. First, we construct a one-dimensional score that serves as the basis for ordering observations. We then apply the one-dimensional fused lasso, following \cite{tibshirani2005sparsity}, to identify subgroups. By design, our approach directly estimates subgroups of individuals with the same treatment effect without requiring heuristic arguments or pre-specifying subgroups, as in \cite{abadie2018endogenous}. Moreover, our method is highly interpretable—arguably even simpler than CART—since it relies solely on a single learned covariate.

Finally, we highlight related  work regarding the fused lasso, the nonparametric  tool that we use in this paper.  Also known as total variation denoising, the fused lasso  first appeared in the machine learning literature \citep{rudin1992nonlinear}, and then in the statistical literature \citep{mammen1997locally}. A discretized  version of total variation regularization was introduced by \cite{tibshirani2005sparsity}.  Since then, multiple  authors  have used  the fused lasso for nonparametric regression in different  frameworks. \cite{tibshirani2014adaptive}  proved  that the fused lasso can attain minimax rates  for  estimation of a one-dimensional function that has bounded variation. 
\citet{guntuboyina2020adaptive}  provided minimax results  for  the fused lasso when estimating  piecewise constant functions. \cite{hutter2016optimal,chatterjee2019new}  studied  the convergence rates of the fused lasso  for denoising  of grid  graphs. \cite{wang2016trend,padilla2017dfs}   considered    extensions  of the fused lasso  to general graphs  structures.  \cite{padilla2020adaptive}  proposed  the fused lasso  for multivariate nonparametric regression and showed adaptivity results for different levels of the regression function. \cite{ortelli2019synthesis} studied further connections between the lasso and fused lasso. 

\subsection{Notation}

For two  random variables  $X$ and $Y$, we use the notation  $X \ind Y$ to indicate that they are independent.  We  write $a_n \,=\, O(b_n)$ if  there exist  constants  $N$ and $C$ such that  $n\geq N $ implies that  $a_n \leq  C b_n$, for sequences  $\{a_n\}, \{b_n\} \subset    \mathbb{R} $. In addition, when $a_n  = O(b_n)$ and  $b_n = O(a_n)$  we use the notation $a_n \asymp b_n$  and sometimes the notation $a_n = \Theta(b_n)$. For a sequence of random variables $\{X_n\}$, we denote  $X_n \,=\,O_{\mathbb{P}}(a_n)$ if for every $\epsilon>0$ there exist positive $N$ and  $C$ such that $ \mathbb{P}( \vert  X_n\vert > C a_n  ) < \epsilon$ for all $n\geq N$.

Finally,   for a random vector  $X \in   \mathbb{R}^d$  we say that $X$ is sub-Gaussian($C$)  for  $ C>0$    if 
\[
\|X\|_{ \psi_2 } \,:=\, \underset{v \in   \mathbb{R}^d \,:\,\|v\|= 1}{\sup}  \|v^{\top} X\|_{ \psi_2 }   \,<\,  C, 
\]
where for a random  variable  $u$ we have
\[
\|u\|_{\psi_2}	:=  \underset{ k\geq 1 }{\sup}\,\,  k^{-1/2} \left\{ E\left(  |u|^k  \right)^{1/k} \right\}.
\]

\subsection{Outline}

The rest of the paper  proceeds as follows. Section \ref{sec:methodology}   describes the mathematical setup of the paper  and  presents  the proposed  class of  estimators. Section \ref{sec:theory}  then   develops  theory  for the corresponding estimators  based on propensity  and prognostic scores. Section \ref{sec:experiments}  provides  extensive comparisons  with  state-of-the-art   methods in the literature on  heterogeneous treatment effects estimation. All  the proofs of the theoretical results  can be found in the Supplementary Material.

\section{Methodology}
\label{sec:methodology}

In this section, we introduce the main methods of the paper. We begin in Section 2.1 with a predecessor estimator that serves as both a foundation and motivation for our approach. In Section \ref{sec:prognostic}, we develop a prognostic score-based estimator suited for completely randomized experiments, followed by a propensity score-based method in Section \ref{sec:propensity}, designed for observational studies.


\subsection{A predecessor estimator}
\label{sec::predecesor}

Consider independent draws $\{ Z_i, X_i, Y_i(1), Y_i(0)   \}_{i=1}^n $, where $Z_i \in \{0,1\}$ is a binary treatment indicator, $X_i \in \{1,\ldots, K\}$ is a discrete and ordinal covariate, and $Y_i \in \mathbb{R}$ is an outcome. Under the Stable Unit Treatment Value Assumption (SUTVA) (see, e.g., \citet{imbens2015causal}), the observed outcome can be expressed as:
\begin{equation}
	\label{eqn:outcome}
	Y_i = Z_iY_i(1) + (1-Z_i) Y_i(0).
\end{equation} 
Under the unconfoundedness assumption 
$Z_i \ind \{  Y_i(1), Y_i(0)\} \mid X_i $ and overlap (see \citet{imbens2004nonparametric}), we can write the subgroup causal effect as
\begin{eqnarray*}
	\tau_{[k]} &=& E\{  Y_i(1) - Y_i(0) \mid X_i = k \}  \\
	&=& E\{  Y_i(1) \mid Z_i =1, X_i = k \}  - E\{  Y_i(0) \mid Z_i = 0, X_i = k \} \\
	&=& E( Y_i \mid Z_i =1, X_i = k )  - E(   Y_i \mid Z_i = 0, X_i = k ) , 
\end{eqnarray*}
which can be identified by the joint distribution of the observed data $\{(  Z_i,X_i,Y_i ) \}_{i=1}^n .$

We can estimate $\tau_{[k]} $ by the sample moments:
\begin{equation}
	\label{eqn:first_estimator}
	\hat{\tau}_{[k]} \,=\,   \frac{1}{n_{[k]1}}\sum_{Z_i=1,X_i=k} Y_i \,-\, \frac{1}{n_{[k]0}}\sum_{Z_i=0,X_i=k} Y_i = \bar{Y}_{[k]1} - \bar{Y}_{[k]0},
\end{equation}
where $n_{[k]z}$ is the sample size of units with covariate value $k$ under the treatment arm $z$.

However, the estimates  $\hat{\tau}_{[k]}$  are well-behaved only when the sample sizes  $n_{[k]z}$   are sufficiently large. When sample sizes are small, particularly as 
$K$ increases, many of the estimates 
$\hat{\tau}_{[k]}$  can become highly variable, resulting in noisy approximations of the true parameters.

When we expect that many subgroup causal effects are small or even zero,  it is reasonable to shrink some $\hat{\tau}_{[k]}$'s to zero, similar to the idea of the lasso estimator \citep{tibshirani1996regression}. Moreover, we may also expect that some subgroup causal effects are close or even identical, then it is reasonable to shrink some $\hat{\tau}_{[k]}$'s to the same value, similar to the idea of the fused lasso \citep{tibshirani2005sparsity}. 
Motivated by these considerations,  define
\begin{equation}
	\label{eqn:predecesor}
	\displaystyle     \tilde{\tau} \,=\,   \underset{ b \in  \mathbb{R}^K }{\arg \min}   \,\,\left\{\frac{1}{2}\sum_{ k=1  }^{K}   (b_k     -   \hat{\tau}_{[k]} )^2   +  \lambda\sum_{ k=1  }^{K-1} \vert  b_k - b_{k+1}\vert   \right\},
\end{equation}
for some  tuning parameter   $\lambda>0$. 
As in \cite{tibshirani2005sparsity} and \cite{tibshirani2014adaptive}, the second term in (\ref{eqn:predecesor}) is the fused lasso penalty  to enforce a piecewise constant structure of the estimates. We use $\ell_1$ regularization instead of $\sum_{ k=1  }^{K-1} \vert  b_k - b_{k+1}\vert^2 $   because the latter would result in linear estimator that is not locally adaptive, see \cite{donoho1998minimax}.

Notice that the term $\sum_{k=1}^K \vert b_k -b_{k+1}\vert$   applies the same penalty to each difference  $ b_k- b_{k+1}$.  At first glance, this might suggest that the categories of 
$X$ must be equally spaced. Specifically, if categories 1 and 2 are closer than categories 2 and 3 in a given application, one might wonder whether $\vert b_1-b_2\vert$ 
should be penalized more than $\vert b_2 -b_3\vert $. The answer is no—the fused lasso penalty is adaptive in the sense that it can adjust to the true signal’s structure, even when the locations of the jumps are unknown or unevenly spaced. For further details, see \cite{tibshirani2014adaptive} and \cite{guntuboyina2020adaptive}.

While $\tilde{\tau}$ seems appealing, it requires   that the covariates  $X_i$ are univariate and categorical. Also, it requires  that there is an order of the covariates under which treatment effects are piecewise constant. Both of these assumptions are usually  not met in practice, as typically  $X_i$  is a vector that can have continuous random variables. The next section proposes a general  class of estimators to handle  such situations.

\subsection{Prognostic-based  estimator}
\label{sec:prognostic}

In this subsection, we focus on completely randomized experiments where $Z_i \ind \{ Y_i(1), Y_i(0), X_i \}$, with $X_i \in  \mathbb{R}^d$.   representing the covariate vector. We define 
$ E\{Y(0)|X\}$ as the prognostic score \citep{hansen2008prognostic}. Furthermore, let  $g(X) = X^{\top }\theta^* $ 
be the best linear approximation of the prognostic score, where
\begin{equation}
	\label{prognostic_score2} \displaystyle  \theta^*  \,=\,    \underset{ \theta \in   \mathbb{R}^d }{\arg \min }\,  \, E\left\{  \left(Y - X^\top \theta\right) ^2   |Z=0 \right\}. 
\end{equation}
\cite{abadie2018endogenous}  assumed a piecewise structure for treatment effects when seen as a function of the approximate prognostic score. However, their approach relied on a fixed number (typically five) of subgroups, determined by discretizing an estimator of the approximate prognostic score. Inspired by \cite{abadie2018endogenous}, we  assume that the treatment effects, as a function of the prognostic score, are piecewise constant. However, rather than pre-specifying the number of subgroups, we assume that $\tau(s) =  E\{ Y(1) - Y(0)\mid g(X) = s\}$  
is either piecewise constant with an unknown number of segments or has bounded variation. See Section \ref{sec:theory} for the precise definition. Let  $\sigma   \,:\,  \{ 1,\ldots,n\}\rightarrow    \{ 1,\ldots,n\}$ be the permutation such that

\begin{equation*}
	\label{eqn:original_sigma}g\{X_{  \sigma(1)}\}   \leq  \cdots\leq \cdots \leq  g\{X_{  \sigma(n) }\}.
\end{equation*}
If $g(s)$ is known and the individual causal effects $Y_i(1) - Y_i(0)$ are known, a natural  estimator  for $\tau_i^* = \tau\{g(X_i)\}$,
for  $i = 1,\ldots,n$,  would be the solution to 
\begin{equation}
	\label{eqn:fused_lasso1}
	\displaystyle  \underset{b \in   \mathbb{R}^n}{\mathrm{minimize} }   \left[    \frac{1}{2}\sum_{i=1}^{n}   \left\{  Y_i(1) - Y_i(0)    -   b_i\right\}^2     +    \lambda \sum_{ i=1  }^{n-1}   \left\vert  b_{\sigma(i)  }  -  b_{\sigma(i+1)  } \right\vert   \right],	
\end{equation}
for  a tuning parameter  $\lambda >0$. This is similar to the standard fused lasso for  one-dimensional nonparametric regression \citep{tibshirani2005sparsity}. The first term   in (\ref{eqn:fused_lasso1})  provides a  measure of fit to the data, and the second term penalizes the total variation to enforce  a piecewise constant structure of the estimated treatment effects. The goal is to adaptively estimate subgroups where the treatment effect, conditional on the prognostic score, remains constant. This approach is conceptually similar to \cite{morucci2023matched}, where the authors constructed a neighborhood for each unit based on a score derived from covariates and then estimated both 
$E\{ Y(1)\mid g(X)\}$ and $E\{ Y(0)\mid g(X)\}$  using a nearest-neighbor-type estimator. However, a key distinction is that the fused lasso \cite{tibshirani2014adaptive} is well known for its ability to adapt to discontinuities in the regression function, whereas the method in \cite{morucci2023matched} relies on a Lipschitz continuity condition.  

	In practice, neither $g(s)$ or the individual treatment effects $Y_i(1)-Y_i(0)$ are known. We first estimate the   prognostic  scores  as  $\hat{g}(X_i) =  X_i^{\top} \hat{\theta}   $  where   
	\begin{equation}
		\label{eqn:theta_hat}
		\displaystyle  \hat{\theta}   \,=\,    \underset{ \theta \in   \mathbb{R}^d }{\arg \min }\,  \,  \left\{\frac{1}{m} \sum_{i=1}^m  \left( Y_i^{\prime} -  X_i^{\prime \top} \theta  \right)^2 \right\}, 
	\end{equation}
	where  $\{ (X_i^{\prime},Y_i^{\prime}) \}_{i=1}^{m}$ are independent copies of 
	$(X,Y)$ conditional on $Z=0$  Additionally, to establish our theoretical results, we require that $\{ (X_i^{\prime},Y_i^{\prime} ) \}_{i=1}^{m}$
	be independent of $\{ (Z_i, X_i, Y_i ) \}_{i=1}^n $. This condition can be satisfied through sample splitting, where the data is evenly divided into two subsets of equal size.  Hence, $\hat{\theta}$  is the estimated vector of coefficients when regressing the outcome  variable on the  covariates  conditioning on the  treatment assignment  being the control group.  Based on the estimated prognostic score, we find  the permutation $\hat{\sigma}$   satisfying 
	\begin{equation}
		\label{eqn:permutation_hat}
		\hat{g}\{X_{  \hat{\sigma}(1) }\}   \leq \cdots \leq \hat{g}\{X_{  \hat{\sigma}(n) }\}.
	\end{equation}

	We then match the units to impute the missing potential outcomes. Define 
	\begin{equation}
		\label{eqn:matching}
		\widetilde{Y}_i  \,=\,   Y_{N(i)},  \,\,\,\,\,\text{with}\,\,\,\,  N(i) \,=\,  \underset{   j \,:\,   Z_j \neq Z_i }{\arg \min}\,\vert \, \hat{g}(X_{i }) - \hat{g}(X_{j })\vert.
	\end{equation}
	So if $Z_i = 1$, then $\widetilde{Y}_i $ is the imputed $Y_i(0)$ and the imputed individual effect is $Y_i - \widetilde{Y}_i $; if $Z_i = 0$, then $\widetilde{Y}_i $ is the imputed  $Y_i(1)$ and the imputed individual effect is $\widetilde{Y}_i  - Y_i$. With these ingredients, we define  the 
	estimator
	\begin{equation}
		\label{eqn:fused_lasso3}
		\hat{\tau} =  \displaystyle  \underset{b \in   \mathbb{R}^n}{\arg \min}   \left[    \frac{1}{2}\sum_{i=1}^{n}   \left\{    Y_i - \widetilde{Y}_i        +  (-1)^{Z_i}  b_i\right\}^2     +    \lambda \sum_{ i=1  }^{n-1}   \left\vert  b_{ \hat{ \sigma}(i)  }  -  b_{\hat{\sigma}(i+1)  } \right\vert   \right].
	\end{equation}
	The optimization problem in (\ref{eqn:fused_lasso3}) can be  solved in $O(n)$ operations by employing the algorithm from  \cite{johnson2013dynamic}  or   that of \cite{barbero2014modular}. Therefore,  $\hat{\tau}$  is our final  estimator  of the vector of subgroup treatment effects  $\tau^*$, where $\tau_i^* = \tau\{g(X_i)\}$,
	for  $i = 1,\ldots,n$, with    $g(x) =  x^{\top} \theta^*$. Similarly, for  $x \notin \{X_1,\ldots,X_n\}$ we can estimate $\tau\{g(x)\}$
	with $\hat{\tau}_i$  where  $X_i$  is the closest  to $x$ among $X_1,\ldots,X_n$. A related  prediction rule was used in a different context by \cite{padilla2020adaptive}.
	
	Because  $\hat{\tau}$ is piecewise constant, we can think  of  its different pieces as data-driven subgroups of units. These so-called  subgroups are estimated adaptively and do not need  to be prespecified.  See Figure \ref{fig1}  for  a visual example of $\hat{\tau}$.

	
	Notice that  we have used  a different data set  for  estimating the prognostic  score $\hat{\theta}$  than the  one for  which we estimate the treatment effects.  This can be achieved in practice  by  sample splitting. The reason why we  proceed in this way is to prevent  $\hat{\sigma}$ from being correlated to  $\{(Z_i,X_i,Y_i)\}_{i=1}^n$.
	
	Regarding  the choice of  $\lambda$,   we proceed as in \cite{tibshirani2012degrees}. Thus,  for each value of $\lambda$ from a list of choices, we  compute the estimator  in (\ref{eqn:fused_lasso3}) and its corresponding degrees of freedom as in \cite{tibshirani2012degrees}. Then we select  the  value of $\lambda$  with the smaller    Bayesian information criterion (BIC).

	
	\begin{figure}[t!]
		\begin{center}
			\includegraphics[width = 2in,height= 1.7in]{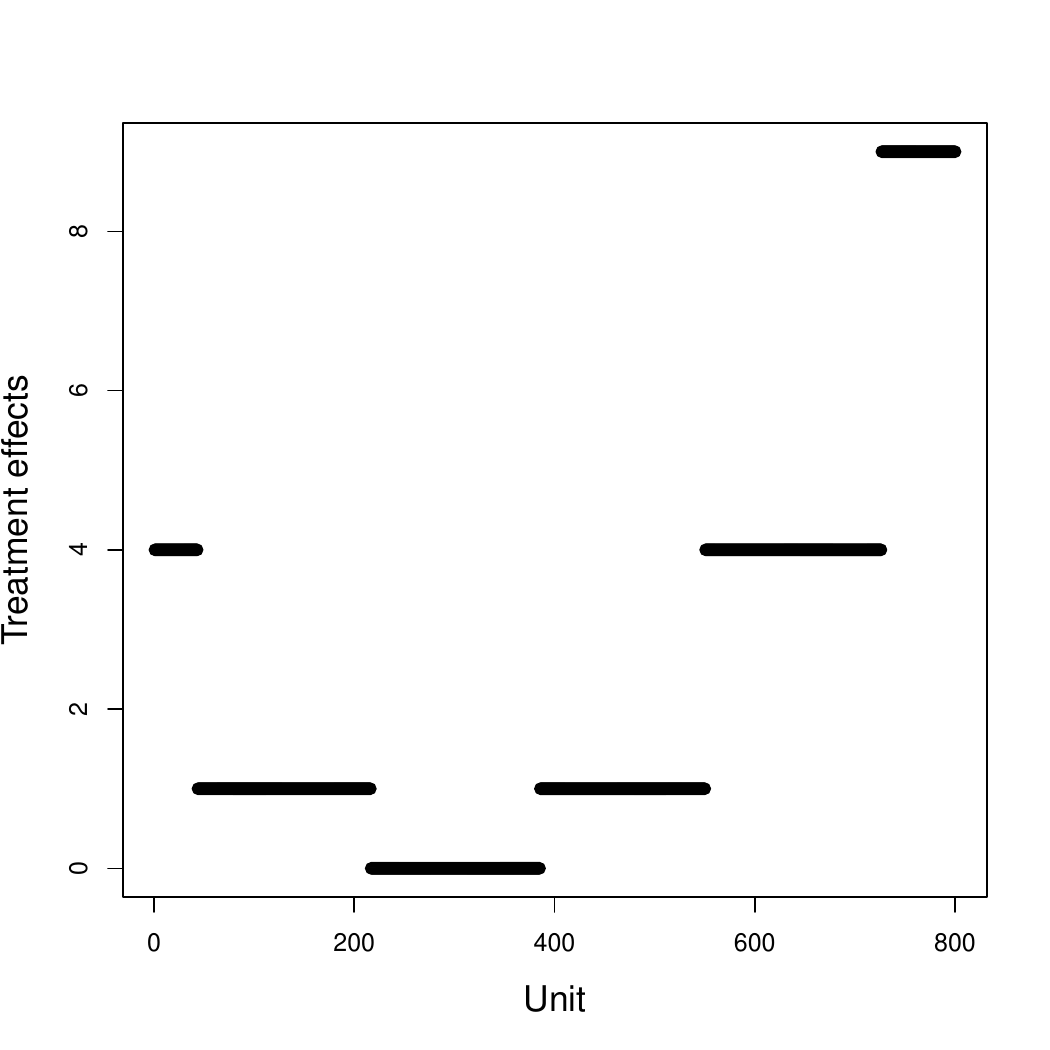} 
			\includegraphics[width=2in,height= 1.7in]{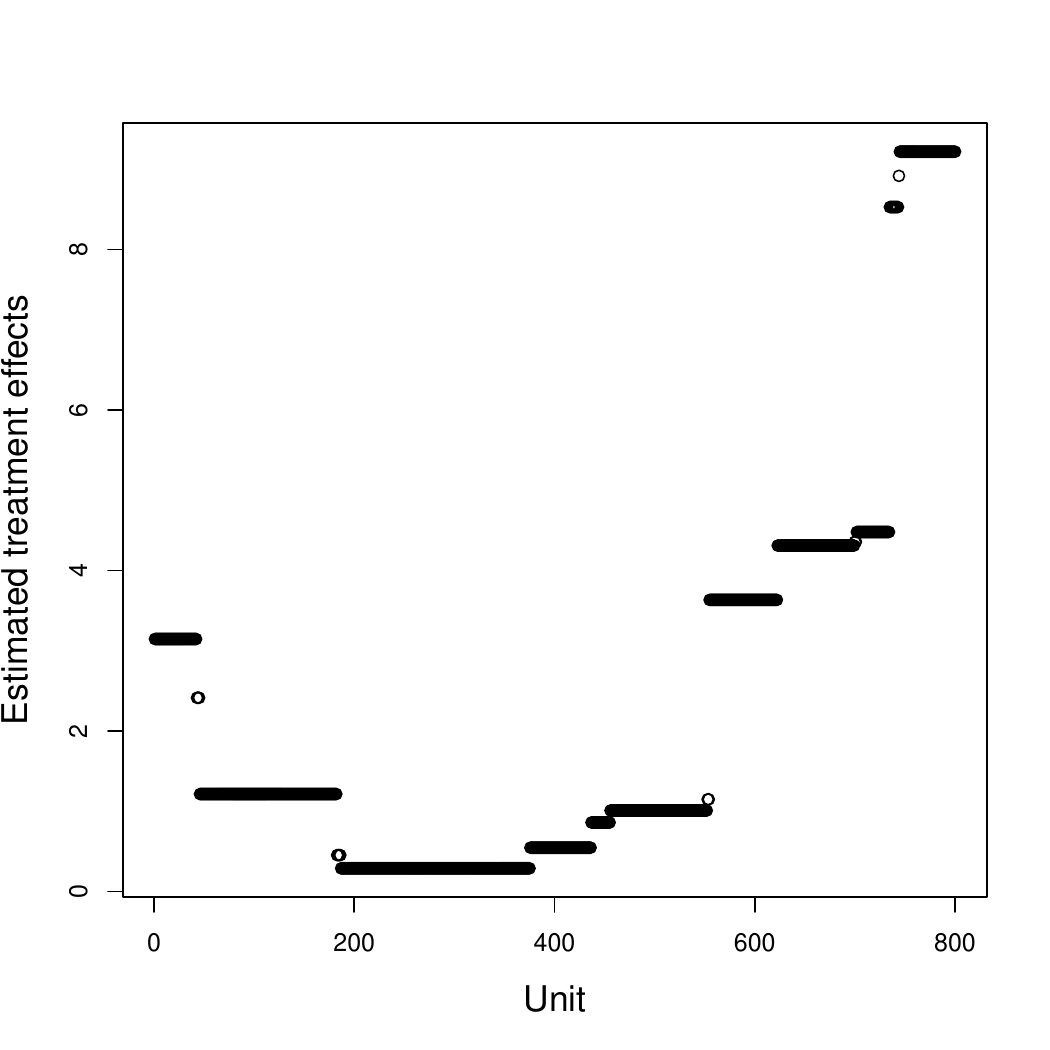}\\
			\caption{\label{fig1} The left panel shows a plot of  $ \tau_{\sigma}^* = ( \tau_{ \sigma(1)  }^*,   \ldots,  \tau_{ \sigma(n)  }^*   )^{\top}  $, where   $\sigma$  is a permutation satisfying  $   X_{ \sigma(1) }^{\top} \boldsymbol{e}_1  <   X_{ \sigma(2) }^{\top} \boldsymbol{e}_1   <  \cdots  <  X_{ \sigma(n) }^{\top} \boldsymbol{e}_1$.  The right panel then shows a plot of   $( \hat{\tau}_{ \hat{ \sigma}(1)    } , \ldots, \hat{\tau}_{ \hat{ \sigma}(n)    }   )^{\top}$, where $\hat{ \sigma}$ is the order  based on the estimated prognostic score as defined in  (\ref{eqn:permutation_hat}).}			
	\end{center}
\end{figure}

\begin{example}
	To illustrate  the  behavior of  $\hat{\tau}$  defined in  (\ref{eqn:fused_lasso3}) we consider  a simple  example. We generate  $\{(Z_i, X_i,Y_i)\}_{i=1}^n$   with  $n =  800$,   and  $d = 10$, from the model
	\begin{equation}
		\label{eqn:scenario1}
		\begin{aligned}
			Y_i(0) &= f_0(X_i) \,+ \,\epsilon_i,\\
			Y_i(1) &= f_0(X_i) \,+\, \tau(X_i)\,+ \,\epsilon_i,\\
			\mathbb{P}(Z_i=1|X_i) & =0.5,\\
			f_0(x)  &=\ \mathrm{sin}\{2(4\pi x^{\top}   \boldsymbol{e}_1   -2)\} + 2.5(4\pi x^{\top}   \boldsymbol{e}_1 -2)+1  ,\,\,\,\, \forall x \in [0,1]^d, \\   
		\end{aligned}
	\end{equation}
	\begin{equation}
		\label{eqn:scenario1.2}
		\begin{aligned}
			\tau(x) &=\  \left\lfloor    \frac{10}{1 + \exp\left( \frac{f_0(x)}{15} -\frac{1}{30} \right)}- 5 \right\rfloor^2,\,\,\,\,  \forall x \in [0,1]^d, \\   
			X_i & \overset{\mathrm{ind} }  {\sim} \  U[0,1]^d,\\
			\epsilon_i & \overset{\mathrm{ind} }  {\sim} \  \mathcal{N}(0,1),
		\end{aligned}
	\end{equation}
	where $\boldsymbol{e}_1  =   (1,0,\ldots,0)^{\top} \in  \mathbb{R}^d$.
	
Notice that, by construction,  $\tau$ can be interpreted as a piecewise constant function of $x^{\top} \boldsymbol{e}_1$. Figure \ref{fig1} illustrates   a plot of the vector  of treatment effects  $\tau^* =(\tau(X_1),\ldots, \tau(X_n)   )^{\top}$. In addition,  Figure  \ref{fig1} shows that   our estimator  $\hat{\tau} $  can reasonably estimate  $\tau^*$ when choosing the tuning parameter via BIC. This is surprising to an extent  since  the ordering that makes $\tau^*$ piecewise constant  is unknown, as both  the propensity  score  and the  function  $f_0$  are unknown.  
\end{example}

\subsection{Propensity score based  estimator}
\label{sec:propensity}

While the prognostic score-based approach can, in principle, be applied to observational studies, it is primarily suited for randomized experiments. As stated in Theorem \ref{thm:prognostic_upper_bound}, the estimate  $\hat{\tau}$ obtained from Equation (\ref{eqn:fused_lasso3}) targets the quantity
\[
\rho_i^* =  E\{Y (1)|g(X) =  g(X_i),  Z = 1\} -  E\{ Y (0)|g(X) =  g(X_i),  Z = 0\},
\]
for $i=1,\ldots,n$. However, to have $\rho_i^*$ to equals the conditional treatment effect 
$E\{Y(1) - Y(0) |g(X)=g(X_i)\}$, we must assume $ Y(0), Y(1) \,\ind\,   Z   \,|\,g(X)$.  This ignorability condition may not hold in observational settings. Therefore, the prognostic score-based method described in Section \ref{sec:prognostic} is intended for use in randomized experiments, consistent with the estimator proposed in \cite{abadie2018endogenous}.

To adapt our approach for observational studies, we follow \cite{rosenbaum1983central} and shift focus from the prognostic score to the propensity score. Accordingly, we define the subgroup treatment effect as:
\begin{equation}
\label{eqn:subgroups2}
\tau(s) =   E\left\{   Y(1) - Y(0) \, |\, e(X )=s\right\},
\end{equation}
where $e(X) = \mathbb{P}(Z=1 \mid  X)$ is the propensity score \citep{rosenbaum1983central}. We can define an analogous estimator for $\tau = (\tau\{e(X_i)\} )_{i=1}^n$ based on the estimated propensity scores. However, our estimator based on the propensity score is specifically designed for observational studies, as the true propensity score is constant in completely randomized experiments. In such settings, conditioning on the propensity score offers no meaningful variation, and thus it does not make sense to estimate heterogeneous treatment effects as a function of the propensity score. 


To be specific, we first estimate the propensity score based on logistic regression to obtain $\hat{e}(X_i) =   F(X_i^{\top }  \hat{\theta})$ where $F(x)=\exp(x)/\{1+\exp(x)\}$ and 
\begin{equation}
\label{eqn:logistic_theta}
\displaystyle  \hat{\theta}   \,=\,    \underset{ \theta \in  \mathbb{R}^d }{\arg \max }\, \left(    \sum_{i=1}^{m}   \left[ Z_i^{\prime} \log  F(  X_i^{\prime \top} \theta )  +  (1-Z_i^{\prime} )\log\{1-   F( X_i^{\prime \top} \theta )\}    \right]     \right).
\end{equation}
Again the sequence $\{ (X_i^{\prime},Z_i^{\prime}) \}_{i=1}^{m}$ consists of independent copies of $(X,Z)$  independent of $\{ (X_i,Z_i,Y_i) \}_{i=1}^{n}$. Obtain the permutation $\hat\sigma$ that satisfies  
\begin{equation}
\label{eqn:permutation_hat2}
\hat{e}\{X_{  \hat{\sigma}(1) }\}   \leq\cdots \leq \hat{e}\{X_{  \hat{\sigma}(n) }\},
\end{equation}
with $\hat{P}$ being the associated permutation matrix. Use matching to impute the missing potential outcomes based on
\begin{equation}
\label{eqn:matching2}
\widetilde{Y}_i  \,=\,   Y_{N(i)},  \,\,\,\,\,\text{with}\,\,\,\,  N(i) \,=\,  \underset{   j \,:\,   Z_j \neq Z_i }{\arg \min}\,\vert \, \hat{e}(X_{i }) - \hat{e}(X_{j })\vert.
\end{equation}
The final estimator for $\tau$ becomes 
\begin{equation}
\label{eqn:fused_lasso5}
\displaystyle\hat{\tau} =   \hat{P}^T\left( \underset{b \in   \mathbb{R}^n}{\arg \min}   \left[   \frac{1}{2} \sum_{i=1}^{n}   \left\{   (-1)^{Z_{\hat{\sigma}(i)     } +1} \left(Y_{  \hat{\sigma}(i)   } - \widetilde{Y}_{  \hat{\sigma}(i)  }  \right)     -  b_{i}\right\}^2     +    \lambda \sum_{ i=1  }^{n-1}   \left\vert  b_{ i  }  -  b_{i+1  } \right\vert   \right]\right),	
\end{equation}
for a tuning parameter  $\lambda
>0$.

One should be cautious in interpreting the propensity score based estimator defined in (\ref{eqn:fused_lasso5}).  Specifically,    (\ref{eqn:fused_lasso5})  estimates  $\rho\{ e(X_i) \}$ for $i=1,\ldots,n$  where
\[
\rho(s)\,:=\,f_1(s) \,-\,f_0(s),
\]   
with
\[
f_0(s) : =  E\{Y |  Z=0,\, e(X)   =  s  \}, \,\,\,\,\,\text{and}\,\,\,\,\,    f_1(s)  :=  E\{Y |  Z=1,\, e(X)   =  s  \}.
\]
However, in general  $\rho(s) \neq \tau(s)$ with $\tau$ as in (\ref{eqn:subgroups2}). Although    $\rho(s) = \tau(s)$  for all $s$   provided that   $Y(0),Y(1) \,\ind\, Z   \,|\, e(X)$. See Section \ref{sec:theory2} for a more detailed discussion.  

\section{Theory}
\label{sec:theory}
\subsection{Main result for prognostic score based estimator}
\label{sec:theory1}

We start  by  studying the  statistical behavior of the prognostic  score based  estimator  defined in (\ref{eqn:fused_lasso3}).  Towards  that   end,  we first  introduce some assumptions  used in our proofs to arrive at our  first  result.



\begin{assumption}[Overlap]
\label{as1}
The propensity score $e(X)$ has  support $ [e_{\min},e_{ \max }] \subset (0,1)$.
\end{assumption}


\begin{assumption}[Surrogate prognostic score]
\label{as7}
We define
\[
\theta^* \,=\,  \underset{ \theta \in  \mathbb{R}^d }{\arg \min} \,\,L(\theta),
\]
with $L(\theta) =  E\left\{\|  Y- X^{\top}\theta\|^2  |Z=0\right\} $, and suppose that  $L$  has a unique minimizer $\theta^* \neq 0$. In addition, we  assume that  the probability  density function of  $\tilde{g}(X) := X^{\top} \theta^*$  is bounded from below and above. The support of  $\tilde{g}(X)$ is denoted as $ [g_{\min},g_{\max}]$.

\end{assumption}

We note that Assumption \ref{as7}  states that the linear surrogate population prognostic score  is well behaved and uniquely defined. This   allows us to understand  the statistical properties of $\hat{\theta}$ defined in (\ref{eqn:theta_hat}), which is potentially  a misspecified
maximum likelihood estimator  \citep{white1982maximum}.



\begin{assumption}[Sub-Gaussian errors]
\label{as8}
Define   $V(z, x) = E\lbrace Y | Z=z, g(X)=g(x)\rbrace$ and $\epsilon_i = Y_i - V(Z_i, X_i)$ for  $i=1,\ldots,n$, with   $g(x) =  E\left\{Y(0) \,|\,  X =x\right\}$. Then  the vector  $ (\epsilon_1,\ldots,\epsilon_n)^{\top}$  has independent coordinates that are mean zero sub-Gaussian$(v)$ for some constant $v>0$.     
\end{assumption}

The previous assumption  requires  that the errors  are  independent, and mean zero  sub-Gaussian. This condition is standard in the analysis of total variation denoising; see for instance \cite{padilla2017dfs}. The resulting  condition 
allows for general models such as normal, bounded distributions, etc. In addition, Assumption \ref{as8} allows for the possibility of heteroscedastic errors.

Our next assumption has to do with behavior  of the mean functions  of the outcome variable, when conditioning on treatment  assignment and prognostic score. We start by recalling the definition of  bounded variation. For  a  function  $f\,:\, [l,u] \,\rightarrow \, \mathbb{R}$,   we  define  its  total  variation as 
\begin{equation}
\label{eqn:one_d_bv}
\displaystyle  \text{TV}(f)  \,=\, 
\underset{  r \geq 1  }{\sup} \,  \text{TV}(f,r) , 
\end{equation}
where
\[
\text{TV}(f,r)  =   \underset{  l \,\leq\, a_1  \,\leq \,\ldots\,\leq  a_r  \leq  u, \,  }{\sup}\, \sum_{l=1}^{r-1} \vert  f(a_l) \,-\,f(a_{l+1}) \vert.
\]
We say that $f$ has bounded variation if $\text{TV}(f) < \infty$. For  a fixed  $C$, the collection
\[
\mathcal{F}_C =  \{   f  \,:\,[0,1]  \rightarrow    \mathbb{R} \,\,:\,\,  \text{TV}(f)\leq C    \},
\]
is a rich class of functions that contains, among others,    Lipschitz continuous functions, piecewise constant and piecewise Lipschitz functions. We refer the reader to \cite{ mammen1997locally,tibshirani2014adaptive} for comprehensive studies of nonparametric regression on the class $\mathcal{F}_C$. 

\begin{assumption}
\label{as9} The functions $f_1(s)   =  E\{Y \,|\,  Z=1,\, g(X)   =  s  \}$, and  $f_0(s)   =  E\{Y \,|\,  Z=0,\, g(X)   =  s  \}$ for  $s \in [g_{\min},g_{\max}]$ are bounded and have  bounded  variation. The latter means that  $t_l    =   \mathrm{TV}(f_l,n)  $,  for  $l \in \{0,1\}$,  satisfy  $\max\{t_0,t_1\} = O(1)$.
\end{assumption}

Importantly, Assumption \ref{as9} allows for the possibility  that  functions  $f_0$ and $f_1$  can have discontinuities.   Our next condition imposes  a relationship between the prognostic score $g$  defined in   Assumption  \ref{as8}
and its  surrogate  $\tilde{g}$ defined in Assumption \ref{as7}.

\begin{assumption}
\label{as6.3}
Let $\tilde{\sigma}$ and  $\sigma$  be random permutations such that 
\[
\tilde{ g}\{X_{ \tilde{\sigma}(1) }\} \leq  \cdots \leq \tilde{ g}\{X_{ \tilde{\sigma}(n) }\},
\] 
and 
\[
g\{X_{ \sigma(1) }\} \leq  \cdots \leq g\{X_{\sigma(n) }\},
\]
respectively, with $\tilde{g}$ as defined  in Assumption \ref{as7}. Then we write 
\[
\mathcal{K}_j \,=\,\left\{  i\,:\,   \left( \sigma^{-1}(i) <\sigma^{-1}(j)  \, \,\text{and}\,\,   \tilde{\sigma}^{-1}(j) <\tilde{\sigma}^{-1}(i)\right) \,\text{or}\,\,  \left( \sigma^{-1}(j) <\sigma^{-1}(i)   \,\,\text{and}\, \,  \tilde{\sigma}^{-1}(i) <\tilde{\sigma}^{-1}(j)\right)  \,\,      \right\},
\]
set  $\kappa_j  :=  \vert  \mathcal{K}_j\vert $	for  $j = \{1,\ldots,n\}$, and 
\[
\kappa_{\max}\,:=\,\underset{j =1 ,\ldots,n}{\max}  \,\kappa_j ,
\]
and require that  $\kappa_{\max}	\,=\, O_{  \mathbb{P} }(\overline{\kappa}_n)$, where $\overline{\kappa}_n$ is a deterministic sequence. 
\end{assumption}

Assumption \ref{as6.3} allows to quantify the discrepancy between the order statistics of the prognostic  score   at the samples and the corresponding order statistics based on the surrogate  prognostic score.  The following remark further clarifies this.



\begin{remark}
Notice that $\kappa_j$   can be thought  as  the number of units   that have different relative orderings  in the  rankings  induced by  $\sigma$ and  $\tilde{\sigma}$. In addition, notice that   the parameter  $\overline{\kappa}_n$   gives an upper  bound on the  entries of the vector $(\kappa_1,\ldots,\kappa_n)$. In fact,  $\overline{\kappa}_n$    can be thought as an  $\ell_{\infty}$  version of the Kendall-Tau  distance  between the permutations   $\sigma^{-1}$ and  $\tilde{ \sigma}^{-1}$. Such  Kendall-Tau  distance   is given as $  \sum_{j=1}^{n}\kappa_j$	(see  \cite{kumar2010generalized} for an overview). Readers can also consider cases in which $\tilde{g}$  and  $g$ induce the same ordering. In such cases,  $\overline{\kappa}_n $ can be taken as zero. 
\end{remark} 

Our next assumption is a condition on the covariates.

\begin{assumption}
\label{as10}  The  random   vectors  $X_1,\ldots,X_n$  are independent  copies of $X$   which has   support $ [a,b]\subset   \mathbb{R}^d$, for some fixed points  $a,b \in  \mathbb{R}^d$. In addition, the following holds: 
\begin{itemize}
\item  	The probability  density  function of $X$, $p_X$,  is bounded.  This amounts to 
\[
p_{\min}  \leq     \underset{ x \in  [a,b]}{\inf}\,\,p_X(x) \,\,\leq  \underset{ x \in  [a,b]}{\sup}\,\,p_X(x)    \leq \,\,  p_{max},
\]  
for  some positive  constants  $p_{\min} $   and  $p_{\max}$.
\item  There exist  $D_{\max}, C_{\min}>0$   such that     	
\begin{equation}
	\label{eqn:eigenvalues}
	C_{\min}    < \Lambda_{\min}\{E( X X^{\top}  ) \}   \leq  \Lambda_{\max}\{E( X X^{\top} ) \}   <   D_{\max},
\end{equation}
where  $\Lambda_{\min}(\cdot)$ and  $\Lambda_{\max}(\cdot)$  are the minimum and maximum eigenvalue functions.
\end{itemize}
\end{assumption}

We emphasize that the first  condition in Assumption \ref{as10} is standard in nonparametric regression, see \cite{padilla2020adaptive}. It is slightly more general than  assuming that the covariates are uniformly drawn in $[0,1]^d$ as in the  nonparametric  regression models in \cite{gyorfi2006distribution}, and the  heterogenous treatment effect setting from \cite{wager2018estimation}. 


With these assumptions,  we are now ready to state our first result  regarding   the  estimation of   $\tau_i^* = \tau\{g(X_i)\}$,  with  $g$ the prognostic score defined in Assumption \ref{as8}.

\begin{theorem}
\label{thm:prognostic_upper_bound}
Suppose that Assumptions \ref{as1}--\ref{as10}  hold,  $m\asymp n$, and that  
\[
d(  \log^{1/2} n     +  d^{1/2}   \|\theta^*\|_1  )      \,\leq\, \frac{c_1 n  }{\log n},
\]
for some large enough constant $c_1>0$.  Then for  a value $t$  satisfying
\[
t  \asymp    \left(\overline{\kappa}_n+1\right)d \left\{ n (   \log^{1/2} n  +     d^{1/2}\|\theta^*\|_1 )\log n \right\}^{1/2}  ,
\]
and for a  choice of $\lambda$  with
\[
\lambda  \,=\,  \Theta\left\{ n^{1/3}   (\log n)^{2}   (\log \log n) t^{-1/3}       \right\},
\]
we have that the estimator defined in (\ref{eqn:fused_lasso3}) satisfies
\begin{equation}
\label{eqn:up1}
\displaystyle \frac{1}{n}  \sum_{i=1}^{n} ( \rho^*_i  - \hat{\tau}_i )^2 \,=\,  O_{ \mathbb{P} }\left[(\log \log n)(\log n)^2  (\overline{\kappa}_n+1)^{2/3}  d^{2/3}\left\{\frac{(   \log^{1/2} n   +     d^{1/2}\|\theta^*\|_1 )\log n }{n} \right\}^{1/3}   \right],
\end{equation}
where  $\rho_i^* =  E\{Y (1)|g(X) =  g(X_i),  Z = 1\} -  E\{ Y (0)|g(X) =  g(X_i),  Z = 0\}$  for  $i=1,\ldots,n$. If in addition   $Y(0), Y(1) \,\ind\,   Z   \,|\,g(X)$, then  (\ref{eqn:up1})  holds replacing $\rho_i^*$ with 
$\tau_i^* =   E\{Y(1) - Y(0) |g(X)=g(X_i)\}$  for  $i=1,\ldots,n$. 
\end{theorem}

We note that $Y(0), Y(1) \,\ind\,   Z   \,|\,g(X)$  is a strong assumption that might not hold in observational studies, however it might be reasonable  in completely randomized trials. 

On a related note, Theorem \ref{thm:prognostic_upper_bound} shows that, up to logarithmic factors, our proposed prognostic score based estimator achieves the convergence rate for the mean squared error given by:
\[
(\overline{\kappa}_n+1)^{2/3} d^{2/3} \left\{\frac{(  1 +     d^{1/2}\|\theta^*\|_1 )}{n} \right\}^{1/3}  .
\]
Here, the dimension $d$ is allowed to grow with the sample size  $n$, and accordingly, 
$\|\theta^*\|_1 $ may also grow with $n$.  However, in the special case where $d= O(1)$ and 
$\|\theta^*\|_1 $  = O(1), the rate simplifies to 
$  (\overline{\kappa}_n+1)^{2/3} /n^{1/3} $. Notably, when the prognostic score is exactly linear—so that  $g= \tilde{g}$ and hence 
$\sigma=\tilde{\sigma} $ —we have 
$\overline{\kappa}_n =0$, and the rate further simplifies to $n^{-1/3}$.

It is also possible for $\sigma=\tilde{\sigma} $   to hold even if 
$g\neq \tilde{g}$ , for example, when 
$g(x) = h(\tilde{g}(x))$  for some strictly increasing function $h$. Nevertheless, in the worst-case scenario, $\overline{\kappa}_n$  could be large enough to dominate the rate, though we do not expect such behavior in typical applications.

The rate $n^{-1/3}$  is slower than the typical rate $n^{-2/3}$  achieved in one-dimensional nonparametric regression using the fused lasso under a bounded variation assumption. However, a key distinction in our setting is that the design points are not directly observed but instead we rely on their estimates. This introduces an additional layer of complexity and variability that affects the convergence rate. We do not claim that our procedure is minimax optimal; in fact, we conjecture that it is not. Nonetheless, our method remains computationally efficient and provably consistent, making it a practical and scalable choice, as we demonstrate in Section \ref{sec:experiments}.

We conclude this section  with a remark  that can be thought as a straightforward generalization of Theorem \ref{thm:prognostic_upper_bound}.
\begin{remark}
\label{remark1}
Notice that the rate  $n^{-1/3}$ does not depend on the dimension $d$ of the covariates   as it is the case of other nonparametric estimators, see \cite{gao2020minimax}. In fact our results  are not directly comparable with \cite{gao2020minimax} as the authors  there consider different classes of functions. A main driver behind the rate $n^{-1/3}$  is Assumption  \ref{as9}. If instead $t^* = \max\{t_0,t_1\}$ is allowed to grow,  then  the upper  bound  in Theorem \ref{thm:prognostic_upper_bound} should be inflated by a factor $(t^*)^{2/3} $. Hence, similar to the discussion above this would lead to the rate   $(t^*)^{2/3} n^{-1/3}$.
\end{remark}

\subsection{Main result for propensity score based estimator}
\label{sec:theory2}

We now  study  the statistical  properties of the  estimator  defined in Section \ref{sec:propensity}.  Since the assumptions required to arrive at our main result here are similar to those  in Section \ref{sec:theory1},  here we only present  the conclusion of our result and the assumptions are given in Section \ref{sec:propensity2}.


\begin{theorem}
\label{thm:tv}
Under Assumptions  \ref{as1}, and \ref{as2}--\ref{as6}, $   dn^{1/2} \geq C_{\min}$, $n \asymp m$,   there exists  $t>0$  such that
\begin{equation}
\label{eqn:scaling}
t\,\asymp\, \max\left\{  \frac{dn^{1/2}  \cdot  \log^{1/2} n\,\log^{1/2} (nd) }{C_{\min} }   , \log n\right\}
\end{equation}
and choice of  $\lambda$ satisfying 
\[
\lambda  \,\asymp\,  n^{1/3}   (\log n)^{2}   (\log \log n) t^{-1/3},
\]
such that the estimator  $\hat{\tau}$ defined in (\ref{eqn:fused_lasso5}) satisfies
\begin{equation}
\label{eqn:up2}
\displaystyle \frac{1}{n}  \sum_{i=1}^{n} ( \rho^*_i  - \hat{\tau}_i )^2 \,=\,  O_{ \mathbb{P} }\left\{    \frac{ d^{2/3}  (\log n)^3 (\log \log n)}{C_{\min}^{2/3  }  n^{1/3 }}  \right\},
\end{equation}
where  $\rho_i^* =  E\{ Y (1)\,|  \,e(X) =  e(X_i),  Z = 1\} - E\{Y (0)\,|\,e(X) =  e(X_i),  Z = 0\}$  for  $i=1,\ldots,n$. If in addition   $Y(0), Y(1) \,\ind\,   Z   \,|\,e(X)$, then  (\ref{eqn:up2})  holds replacing $\rho_i^*$ with 
$\tau_i^* =  E\{ Y(1) - Y(0) \,|\,e(X)=e(X_i)\}$  for  $i=1,\ldots,n$.  
\end{theorem}

Importantly, Theorem \ref{thm:tv} implies that the estimator $\hat{\tau}$ defined in (\ref{eqn:fused_lasso5})  can consistently estimate the subgroup treatment effects  $\tau^*$ under general conditions.  One of such conditions is that  $Y(0), Y(1) \,\ind\,   Z   \,|\,e(X)$, which in the language of  \cite{rosenbaum1983central} means that treatment is strongly ignorable given $e(\cdot)$. As Theorem 3 in \cite{rosenbaum1983central} showed,  $Y(0), Y(1) \,\ind\,   Z   \,|\,e(X)$ holds under overlapping (Assumption \ref{as1}) and unconfoundedness   which can be writen as $Y(0), Y(1) \,\ind\,   Z   \,|\,e(X)$.  When these conditions are violated, Theorem \ref{thm:tv} shows that $\hat{\tau}$ can still approximate $\rho^*$ under Assumptions  \ref{as1}, and \ref{as2}--\ref{as6}.

\section{Experiments}

\label{sec:experiments}

We will now validate  with experiments the proposed methods in this paper. Throughout this section, we refer to the procedure in Section~\ref{sec:prognostic} as \textit{Causal Fused Lasso 1 (CFL1)}, and the procedure in Section~\ref{sec:propensity} as \textit{Causal Fused Lasso 2 (CFL2)}. For both estimators, we select the tuning parameter~$\lambda$ by minimizing the Bayesian Information Criterion (BIC) over a grid of candidate values, as described in Section~\ref{sec:prognostic}. For each $\lambda$, we compute the fused lasso estimator and evaluate BIC using the residual sum of squares and the estimated degrees of freedom following the approach of \cite{tibshirani2012degrees}.

We benchmark our methods against several widely used baselines. These include causal random forests Procedure 1 (WA1) and Procedure 2 (WA2) from \cite{wager2018estimation}, the robust generalized random forest (GRF) from Section~6.2 of \cite{athey2019generalized}, and the estimator of \cite{abadie2018endogenous} (ACW). We also include two flexible, nonparametric methods: Bayesian Additive Regression Trees (BART) from \cite{10.1214/09-AOAS285}, which models the outcome as a sum of regression trees and estimates individual treatment effects as the difference in posterior mean outcomes under treatment and control; and the Augmented Inverse Probability Weighting (AIPW) estimator from \cite{glynn2010introduction}, which combines outcome regression and propensity score weighting and enjoys double robustness.

To further enhance our evaluation and directly address recent developments in interpretable causal inference, we include a family of matching-based estimators motivated by the ``almost exact matching'' framework. These include: Genetic Matching (GM) from \cite{diamond2013genetic}, MALTS from \cite{parikh2022malts},   Lasso Coefficient Matching (LCM) from \cite{lanners2023variable}, ADD-MALTS from \cite{katta2024interpretable}, and Adaptive Hyperbox Matching (AHB) from \cite{morucci2020adaptive}. Each of these approaches defines a strategy to identify comparable units, using learned distance metrics or adaptive rules, and estimates the treatment effect for each unit by imputing one or both missing potential outcomes from matched units.

In addition, we consider two interpretable subset-based matching estimators: FLAME from \cite{wang2021flame} and DAME from \cite{liu2018interpretable}.  These algorithms construct matched groups by sequentially selecting subsets of covariates that optimize a trade-off between covariate balance and predictive accuracy of the outcome, resulting in interpretable, rule-based matching schemes. Unlike instance-level matching approaches such as MALTS or AHB that rely on learned distance metrics, FLAME and DAME perform combinatorial matching on covariate subsets to identify groups where units match exactly on a carefully chosen set of features. Because these methods are designed for categorical covariates, we adapt them to our continuous covariate setting by discretizing each feature into quantile-based bins prior to matching. While this preprocessing step introduces approximation error, it enables a meaningful comparison with these almost-exact matching approaches in our synthetic scenarios.


We note that CFL1 and CFL2 induce interpretable, data-driven subgroups via total variation regularization applied to estimated prognostic or propensity scores. This leads to piecewise-constant treatment effect estimates across individuals, derived not from pairwise similarity or nearest-neighbor heuristics, but from optimization principles that robustly segment units based on heterogeneity. In doing so, CFL1 and CFL2 attain the interpretability of subgroup-based methods like FLAME and DAME, while avoiding some of the limitations associated with distance-based matching, such as sensitivity to poor overlap or dependence on learned distance metrics.

\subsection{Simulated and Semi-Synthetic Experiments}

\begin{table}[t]
\centering
\caption{\label{tab1} Performance evaluations (median $\pm$ standard error) over 50 Monte Carlo simulations for synthetic scenarios with varying $(n,d)$. \textbf{Bold} indicates the best method, and \textit{italic} indicates the second-best.}

\medskip
{\fontsize{10}{14}\selectfont
\setlength{\tabcolsep}{6pt}
\begin{tabular}{|l|l|l|l|l|l|}
\hline
Method & $(n,d)$ & Scenario 1 & Scenario 2 & Scenario 3 & Scenario 4 \\
\hline
\multirow{4}{*}{CFL1} 
& (800, 2)  & \textbf{0.004 $\pm$ 0.0012} & 0.195 $\pm$ 0.055 & 0.181 $\pm$ 0.037 & \textbf{0.301 $\pm$ 0.065} \\
& (1600, 2) & \textbf{0.003 $\pm$ 0.0009} & 0.108 $\pm$ 0.032 & 0.136 $\pm$ 0.028 & \textbf{0.183 $\pm$ 0.044} \\
& (800, 10) & \textbf{0.005 $\pm$ 0.0013} & \textbf{0.503 $\pm$ 0.083} & 0.412 $\pm$ 0.082 & \textbf{0.450 $\pm$ 0.086} \\
& (1600, 10)& \textbf{0.003 $\pm$ 0.0010} & \textbf{0.319 $\pm$ 0.069} & 0.293 $\pm$ 0.063 & \textbf{0.277 $\pm$ 0.068} \\
\hline
\multirow{4}{*}{CFL2} 
& (800, 2)  & 0.011 $\pm$ 0.0023 & * & \textbf{0.074 $\pm$ 0.017} & * \\
& (1600, 2) & 0.004 $\pm$ 0.0011 & * & 0.051 $\pm$ 0.014 & * \\
& (800, 10) & 0.016 $\pm$ 0.0038 & * & \textbf{0.146 $\pm$ 0.033} & * \\
& (1600, 10)& 0.005 $\pm$ 0.0016 & * & \textbf{0.109 $\pm$ 0.027} & * \\
\hline
\multirow{4}{*}{GRF} 
& (800, 2)  & 0.013 $\pm$ 0.0031 & \textbf{0.152 $\pm$ 0.048} & 0.143 $\pm$ 0.031 & 1.788 $\pm$ 0.318 \\
& (1600, 2) & 0.011 $\pm$ 0.0024 & \textbf{0.063 $\pm$ 0.018} & 0.106 $\pm$ 0.024 & 0.771 $\pm$ 0.177 \\
& (800, 10) & 0.010 $\pm$ 0.0032 & 0.565 $\pm$ 0.092 & 0.408 $\pm$ 0.078 & 3.261 $\pm$ 0.557 \\
& (1600, 10)& 0.006 $\pm$ 0.0021 & 0.350 $\pm$ 0.077 & 0.359 $\pm$ 0.070 & 1.291 $\pm$ 0.308 \\
\hline
\multirow{4}{*}{AIPW} 
& (800, 2)  & 0.015 $\pm$ 0.0040 & 0.235 $\pm$ 0.058 & 0.109 $\pm$ 0.027 & 1.918 $\pm$ 0.362 \\
& (1600, 2) & 0.010 $\pm$ 0.0025 & 0.132 $\pm$ 0.036 & \textbf{0.049 $\pm$ 0.015} & 1.089 $\pm$ 0.291 \\
& (800, 10) & 0.023 $\pm$ 0.0052 & 0.672 $\pm$ 0.112 & 0.342 $\pm$ 0.059 & 3.988 $\pm$ 0.611 \\
& (1600, 10)& 0.019 $\pm$ 0.0041 & 0.498 $\pm$ 0.096 & 0.286 $\pm$ 0.051 & 3.134 $\pm$ 0.552 \\
\hline
\multirow{4}{*}{MALTS} 
& (800, 2)  & 0.018 $\pm$ 0.0042 & 0.248 $\pm$ 0.057 & 0.098 $\pm$ 0.023 & 2.156 $\pm$ 0.391 \\
& (1600, 2) & 0.011 $\pm$ 0.0031 & 0.112 $\pm$ 0.028 & \textit{\textbf{0.050 $\pm$ 0.014}} & 1.141 $\pm$ 0.283 \\
& (800, 10) & 0.030 $\pm$ 0.0061 & 0.618 $\pm$ 0.102 & 0.238 $\pm$ 0.051 & 3.274 $\pm$ 0.553 \\
& (1600, 10)& 0.024 $\pm$ 0.0054 & 0.401 $\pm$ 0.087 & 0.212 $\pm$ 0.046 & 2.148 $\pm$ 0.472 \\
\hline
\multirow{4}{*}{AHB} 
& (800, 2)  & 0.010 $\pm$ 0.0028 & 0.231 $\pm$ 0.051 & 0.117 $\pm$ 0.025 & 1.503 $\pm$ 0.297 \\
& (1600, 2) & \textbf{0.003 $\pm$ 0.0009} & 0.136 $\pm$ 0.031 & 0.081 $\pm$ 0.017 & 0.908 $\pm$ 0.239 \\
& (800, 10) & 0.013 $\pm$ 0.0034 & 0.492 $\pm$ 0.083 & 0.234 $\pm$ 0.048 & 2.182 $\pm$ 0.459 \\
& (1600, 10)& 0.005 $\pm$ 0.0016 & 0.331 $\pm$ 0.071 & 0.191 $\pm$ 0.041 & 1.523 $\pm$ 0.382 \\
\hline
\end{tabular}
}
\end{table}
\vspace{-4mm}

We assess the performance of our proposed methods across eight distinct scenarios that encompass both completely synthetic and semi-synthetic designs. Scenarios 1–6 are fully synthetic, with both covariates and outcomes generated from known models. Scenarios 7 and 8 are semi-synthetic, using covariates from real datasets (the National JTPA and Project STAR studies, respectively) and simulated outcomes as in \cite{abadie2018endogenous}. For Scenarios 1--4, we consider varying values of the sample size $n \in \{800, 1600\}$ and the covariate dimension $d \in \{2, 10\}$. For Scenarios 5 and 6, we set $n=4000$ and $d=10.$ For each combination, we generate a dataset $\{(Z_i, X_i, Y_i)\}_{i=1}^n$ according to the corresponding generative model. In Scenarios 7 and 8, which include semi-synthetic simulations, the values of $n$ and $d$ are determined by the underlying real datasets or specific design choices. In all cases, we evaluate performance using the mean squared error (MSE),
\[
\frac{1}{n} \sum_{i=1}^n \left( \tau_i^* - \hat{\tau}_i \right)^2,
\]
where $\tau_i^* = \mathbb{E}[Y_i \mid X_i, Z_i = 1] - \mathbb{E}[Y_i \mid X_i, Z_i = 0]$, and $\hat{\tau}_i$ is the estimate from a given method. MSEs are averaged over 50 Monte Carlo replications, and we report associated standard errors.
Descriptions of each scenario follow.

Scenarios 1–4, which are completely synthetic, are described in Section~\ref{sec:ex1.2}. The first  two scenarios come from \cite{wager2018estimation}  and both consist of  $\tau^*_i =0 $  for  all $i  \in \{1,\ldots,n\}$. 
In Scenario~3, we define the treatment effect as $\tau_i^* = \boldsymbol{1}_{\{e(X_i) > 0.6\}}$, where $e(x) = \Phi(\beta^\top x)$ is the propensity score and $\Phi$ denotes the standard normal CDF. The vector $\beta \in \mathbb{R}^d$ is fixed with $\beta_j = 1$ for $j \leq \lfloor d/2 \rfloor$ and $\beta_j = -1$ otherwise. Furthermore,  Scenario 4 is the model  described in (\ref{eqn:scenario1}).

\textit{Scenario 5.} This fully synthetic scenario comes from \cite{abadie2018endogenous}. Setting $d=10$  and  $n =4000$ the data is generated as: $	Y_i =\  1+    \beta^{\top}X_i  +\epsilon_i $,  $	X_i \overset{\mathrm{ind}}{\sim}  \mathcal{N}\left(0,\mathbf{I}_{d\times d}\right)$
and  $\epsilon_i \overset{\mathrm{ind}}{\sim}  \mathcal{N}\left(0,100-d\right)$, where  $\beta =(1,\ldots,1)^{\top} \in  \mathbb{R}^{d} $.  Moreover, the treatment indicators \( Z_i \in \{0,1\} \) are assigned such that \( \sum_i Z_i = \lceil n/2 \rceil \).
 Clearly, the vector  of treatment  effects satisfies  $\tau^* =0$.

\textit{Scenario 6.} For our final fully synthetic model  we  set  $d=10$, $n =4000$, and generate data as
\[
\begin{array}{lll}
Y_i    &= &     (1-Z_i )Y_i(0) +     Z_i Y_i(1), \\
Z_i & \sim  & \mathrm{Binom}(1,   \frac{1}{2}      ),\\
Y_i(l) & \sim &  \mathcal{N}( f_l(X_i) , 1   ), \,\,\,\forall l  \in \{0,1\}, \\
f_0(x)& =&  x^{\top} \beta,   \,\,\,\forall x  \in [0,1]^d, \\
f_1(x) &=& f_0(x) +   \boldsymbol{1}_{ \{ x^{\top} \beta>1 \}  }   +   \boldsymbol{1}_{ \{ x^{\top} \beta<0.2 \}  }    ,   \,\,\,\forall x  \in [0,1]^d, \\
X_i  &\overset{\mathrm{ind}}{\sim} &   U[0,1]^d,\,\,\,\forall i \{1,\ldots,n\},
\end{array}
\]
where  $\beta    \in   \mathbb{R}^{p}$  with $\beta_j =  1$  if $j  \in \{1,\ldots, \floor{p/2} \}$, and   $\beta_j =  -1$  otherwise. Notice that in this  case  the treatment effect for unit $i$ is   $\tau^{*}_i =  \boldsymbol{1}_{ \{ X_i^{\top} \beta>1 \}  }   +   \boldsymbol{1}_{ \{ X_i^{\top} \beta<0.2 \}  }   $.

The final two scenarios are semi-synthetic designs, where covariates are taken from real datasets, while the outcomes are simulated according to a known data-generating process. This allows us to evaluate the estimators in realistic covariate spaces while preserving ground-truth treatment effects.

\textit{Scenario 7.}  We use the  setting   of the National JTPA Study  used in \cite{abadie2018endogenous}. This consists of a  National JTPA Study  evaluating  an employment and training program commissioned by the U.S. Department of Labor
in the late 1980s. Other authors  that have also analyzed this  data include  \cite{orr1996does} and \cite{bloom1997benefits}. In the JTPA study, based on a randomized assignment, subjects were a assigned into one of two groups. In the treatment group the subjects had access to JTPA services that included one of three possibilities: on-the-job training/job search assistance, classroom training, and other services. In contrast, subjects in the control group were not given access to the JTPA services. The raw data consists of 2530 units $\left(n_{o b s}=2530\right), 1681$ of which are treated observations and 849 are untreated observations. With these measurements, we generate simulated data as in \cite{abadie2018endogenous} where the treatment effect is zero across all units. The details are given in Appendix \ref{sec:ex2.2}.

\textit{Scenario 8.} We also consider an example used in \cite{abadie2018endogenous}. Specifically, we use the  Project STAR
class-size study, see for instance  \cite{krueger1999experimental}. In this data, 3,764 students who entered the study in kindergarten were assigned to small classes or to regular-size classes
(without a teacher’s aide). The outcome variable is standardized end-of-the-year kindergarten math test scores. As for covariates, some of these include race, eligibility for the free lunch program, and school attended.  With the original data we proceed as in \cite{abadie2018endogenous} and  simulate  data in a setting where the treatment effects are all zero. The details are given in Appendix \ref{sec:ex2.3}.

The results of our experiments in Scenarios 1--4 for the top six competitors are reported in Table~\ref{tab1}. The full comparison with all methods is deferred to the appendix (see Table~\ref{tab11}). There, we  can see  that  for Scenario 1 the best  methods are our proposed estimators, which is reasonable since  the treatment effect is zero  across units. AHB matches our performance at the configuration $(n,d) = (1600, 2)$ in this scenario. In Scenario 2, we do not compare CFL2 since such method is not suitable for experimental designs where the propensity score takes on a constant value. GRF achieves the best performance in low dimension ($d=2$), closely followed by CFL1. However, in higher dimension ($d=10$), CFL1 outperforms GRF, highlighting its robustness in more complex covariate settings.

In Scenario 3, we see that  CFL2 generally outperforms the competitors. At the configuration $(n,d) = (1600, 2)$, however, AIPW achieves the best performance, closely followed by MALTS, both of which slightly outperform CFL2. Notice that in Scenario 3 the treatment effect  is a function of the propensity score, where  the propensity score belongs to the family of probit models. This does not seem to be a problem for  our estimator  which provides  accurate  estimation  despite relying on logistic  regression in the first stage. Moreover, in Scenario 4, we see that  CFL1 outperforms the competitors. Again, since  the propensity score is constant, we do not benchmark CFL2.

Table~\ref{tab2} summarizes results for Scenarios 5--8. To facilitate comparison across scenarios, we standardized the outcome variable 
\(Y\) in each case by subtracting the mean and dividing by the standard deviation before applying each estimator. In Scenarios 5 and 6, which are fully synthetic and involve a larger sample size of 4000 observations, CFL1 achieves the best performance, outperforming all benchmark methods by a substantial margin. These results underscore the flexibility and robustness of CFL1, particularly in large-sample settings and under varied treatment effect structures, including the presence of high noise (Scenario 5) and sharp discontinuities (Scenario 6).

In the semi-synthetic Scenarios 7 and 8, which use covariates from real datasets, CFL1 remains highly competitive. In Scenario 7 (JTPA), ACW attains the lowest MSE, followed closely by CFL1 and WA2. While ACW performs best in Scenario 7, this is likely due to the structure of the real-world covariates aligning well with its stratification procedure, rather than the constancy of the treatment effect alone. In Scenarios 1–4, despite having constant or piecewise-constant treatment effects, ACW is consistently outperformed by our methods, particularly CFL1 (see Table~\ref{tab11} in Appendix~\ref{add-simu-results}). This suggests that flexible subgroup discovery and modeling heterogeneity, as in CFL1 and CFL2, are advantageous even when effects are simple.

In Scenario 8 (Project STAR), which involves a high-dimensional covariate space and a null treatment effect, CFL1 again outperforms all competitors, including ACW. These findings demonstrate that CFL1 performs reliably not only with realistic covariate distributions but also under more complex or high-dimensional conditions where accurate estimation requires stronger regularization and global modeling capabilities. Overall, our results highlight the strong empirical performance of both CFL1 and CFL2 across diverse experimental setups, with one of the two—CFL1 or CFL2—consistently ranking among the top two methods in every scenario.

\begin{table}[t!]
\centering
\caption{\label{tab2} Performance evaluations (median $\pm$ standard error) over 50 Monte Carlo simulations for synthetic (Scenarios 5–6) and semi-synthetic (Scenarios 7–8) setups. \textbf{Bold} indicates the best method, and \textit{italic} indicates the second-best. For Scenarios 5–6, $(n,d) = (4000, 10)$; for Scenario 7, $(n,d) = (3764, 79)$; and for Scenario 8, $(n,d) = (2530, 18)$.}
\medskip
{\fontsize{11}{20}\selectfont
\setlength{\tabcolsep}{5pt}
\begin{tabular}{|l|c|c|c|c|}
\hline
Method & Scenario 5 & Scenario 6 & Scenario 7 & Scenario 8 \\
\hline
CFL1 & \textbf{0.078 $\pm$ 0.017} & \textbf{0.074 $\pm$ 0.015} & \textit{\textbf{0.042 $\pm$ 0.008}} & \textbf{0.308 $\pm$ 0.071} \\
\hline
WA1  & 0.489 $\pm$ 0.077 & 0.423 $\pm$ 0.072 & 0.103 $\pm$ 0.020 & 0.592 $\pm$ 0.089 \\
\hline
WA2  & 0.352 $\pm$ 0.066 & 0.295 $\pm$ 0.057 & 0.044 $\pm$ 0.010 & 0.435 $\pm$ 0.075 \\
\hline
GRF  & 0.319 $\pm$ 0.058 & 0.204 $\pm$ 0.041 & 0.045 $\pm$ 0.012 & 0.382 $\pm$ 0.068 \\
\hline
ACW  & 0.462 $\pm$ 0.070 & \textit{\textbf{0.131 $\pm$ 0.030}} & \textbf{0.015 $\pm$ 0.004} & 0.501 $\pm$ 0.086 \\
\hline
BART & \textit{\textbf{0.248 $\pm$ 0.044}} & 0.166 $\pm$ 0.037 & 0.057 $\pm$ 0.011 & \textit{\textbf{0.329 $\pm$ 0.069}} \\
\hline
AIPW & 0.278 $\pm$ 0.051 & 0.149 $\pm$ 0.035 & 0.044 $\pm$ 0.010 & 0.399 $\pm$ 0.073 \\
\hline
GM    & 0.376 $\pm$ 0.061 & 0.231 $\pm$ 0.048 & 0.061 $\pm$ 0.013 & 0.449 $\pm$ 0.078 \\
\hline
MALTS & 0.341 $\pm$ 0.059 & 0.188 $\pm$ 0.039 & 0.065 $\pm$ 0.013 & 0.371 $\pm$ 0.072 \\
\hline
LCM   & 0.302 $\pm$ 0.056 & 0.226 $\pm$ 0.043 & 0.071 $\pm$ 0.014 & 0.401 $\pm$ 0.079 \\
\hline
ADD-MALTS & 0.265 $\pm$ 0.050 & 0.132 $\pm$ 0.031 & 0.068 $\pm$ 0.013 & 0.343 $\pm$ 0.071 \\
\hline
AHB   & 0.312 $\pm$ 0.055 & 0.161 $\pm$ 0.034 & 0.066 $\pm$ 0.012 & 0.361 $\pm$ 0.073 \\
\hline
FLAME & 0.407 $\pm$ 0.066 & 0.273 $\pm$ 0.049 & 0.073 $\pm$ 0.014 & 0.371 $\pm$ 0.075 \\
\hline
DAME  & 0.331 $\pm$ 0.057 & 0.211 $\pm$ 0.045 & 0.060 $\pm$ 0.012 & 0.355 $\pm$ 0.072 \\
\hline
\end{tabular}
}
\end{table}

\begin{figure}[ht!]
\begin{center}
\includegraphics[width=2.58in,height=2.3in]{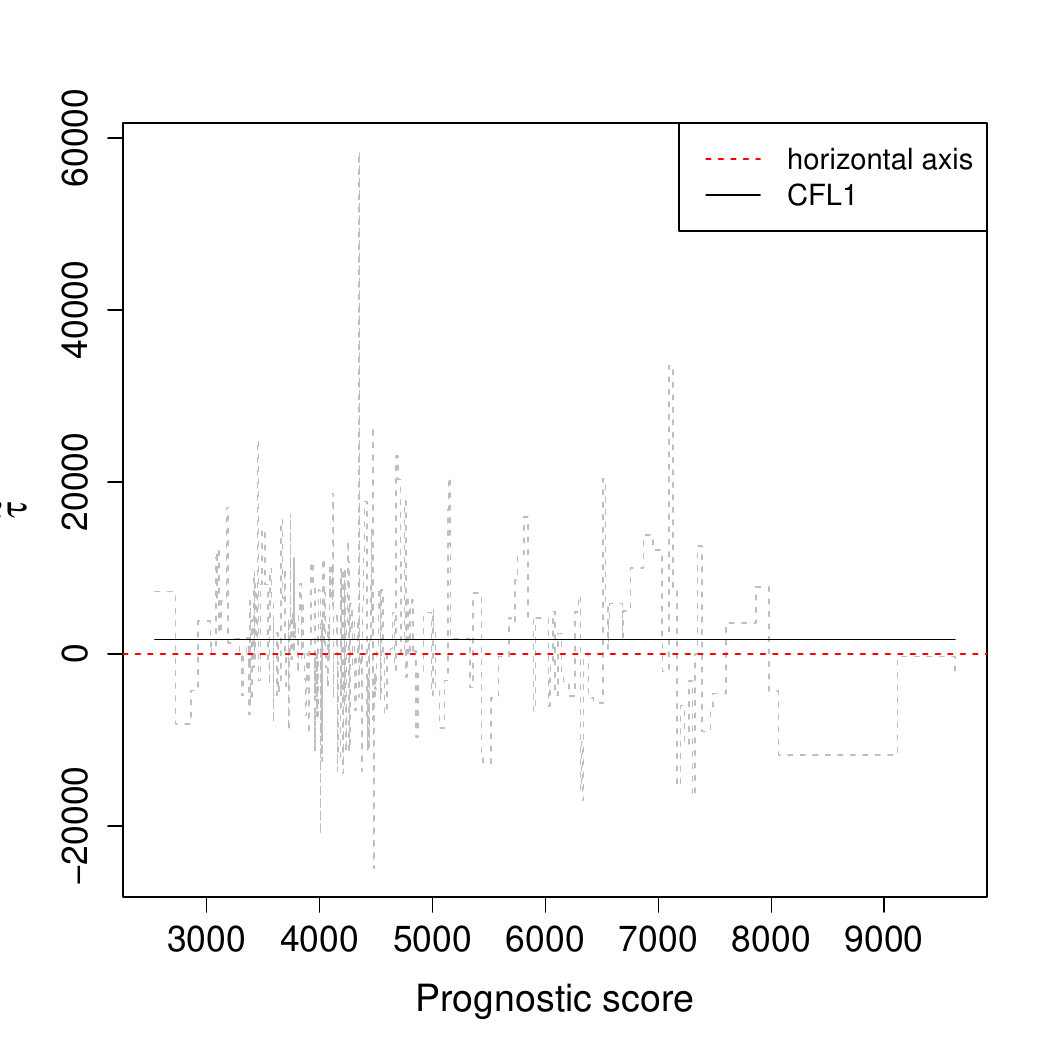} 
\includegraphics[width=2.58in,height=2.3in]{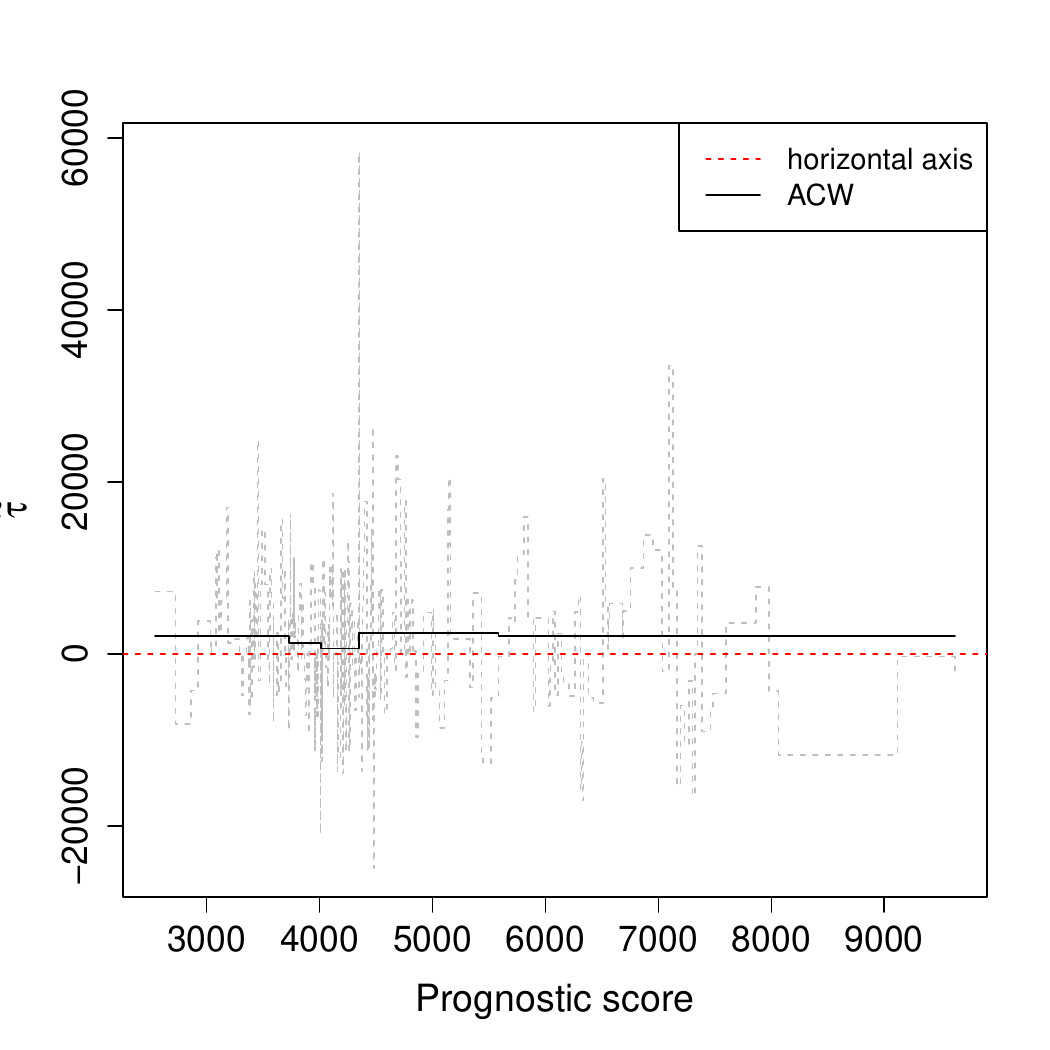} 	
\caption{\label{fig4}  For the NSW example from left to right the two panels  show the treatment effect estimates based on  causal fused lasso  with prognostic score (CFL1) and the ACW method from \cite{abadie2018endogenous}.}
\end{center}
\end{figure}

\begin{figure}
\begin{center}
\includegraphics[width=2.58in,height=2.3in]{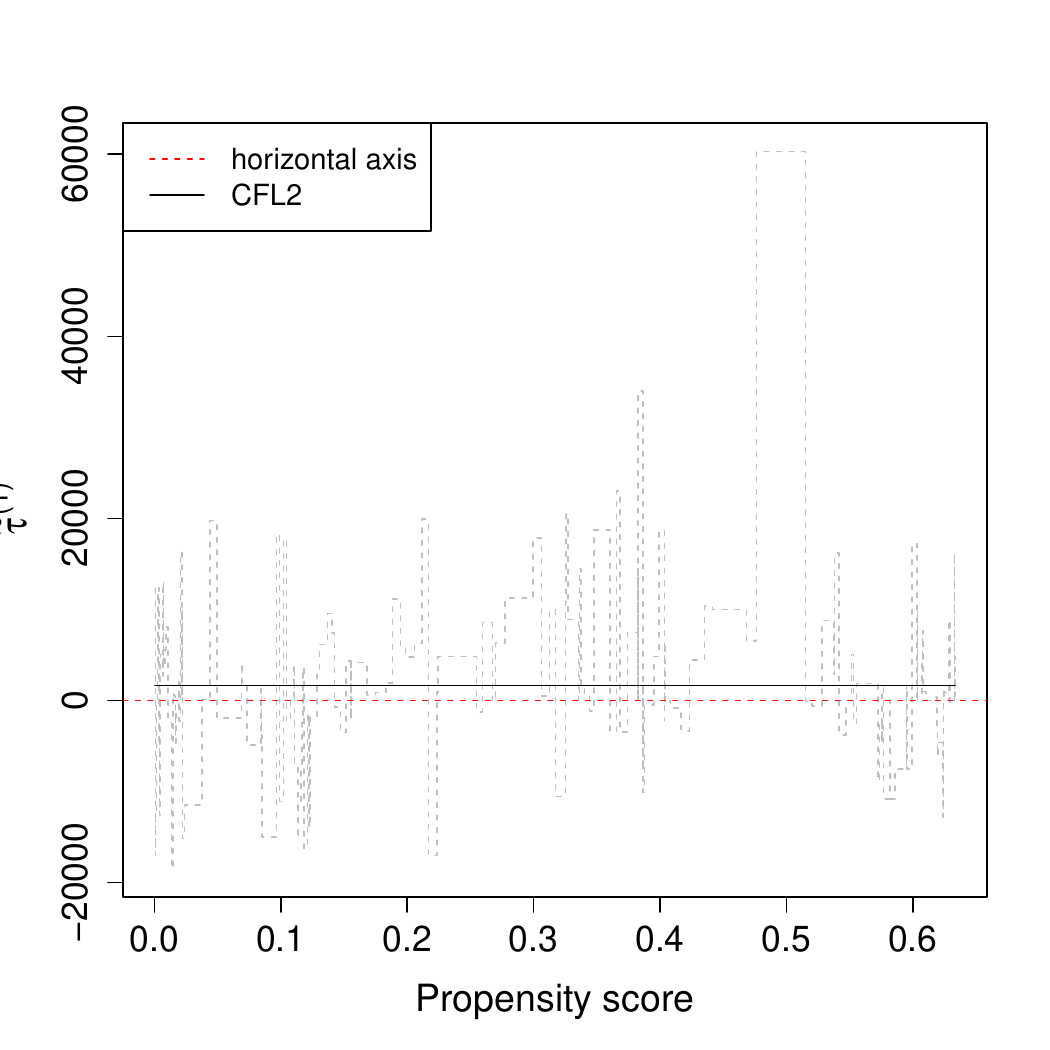} 
\includegraphics[width=2.58in,height=2.3in]{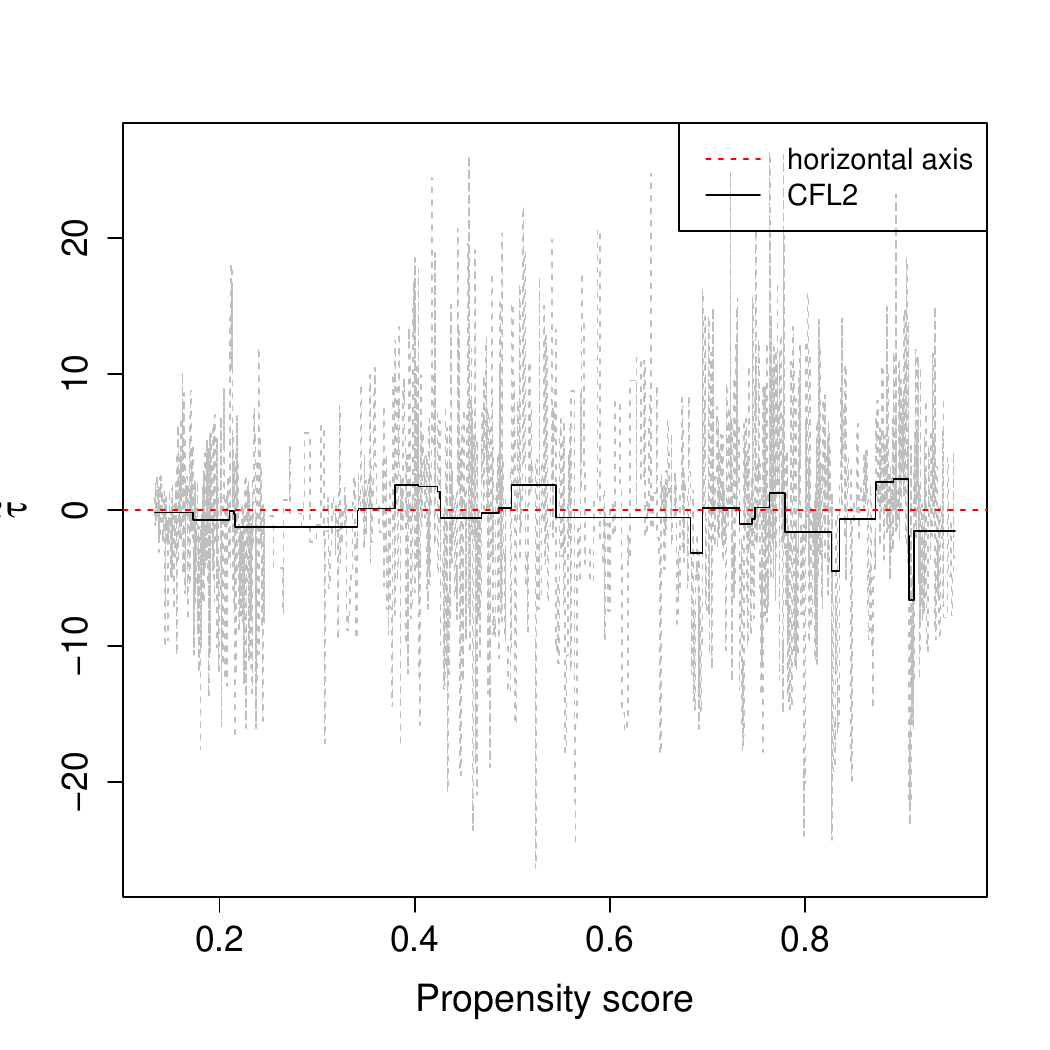} 
\caption{\label{fig5} The panel on the left shows the estimated treatment effects on the \emph{treated} based on the CFL2 estimator for the NSW observational data. The right panel then shows the corresponding treatment effect estimates of CFL2 for the NHANES data.}
\end{center}
\end{figure}

\subsection{National Supported Work  data}

\subsubsection{Randomized example}
\label{sec:real1}

To illustrate the behavior of our estimators, we use the  data  from \cite{lalonde1986evaluating,dehejia1999causal,dehejia2002propensity}. This dataset consists of a 445 sub-sample from the National Supported Work Demonstration (NSW). The NSW  was a program implemented in the mid-1970s  in which the treatment group consisted of randomly selected subjects to gain 12 to 18 months of work experience and  260  subjects in a control group.   The response  variable is  the post-study earnings in 1978. The predictor variables include  age, education, indicator of Black and Hispanic for race, marital status, high-school degree indicator, earnings in 1974, and earnings in 1975.

To construct our estimator, we first    estimate the prognostic score using the data from the control  group and running a linear regression model. 
With the prognostic scores, we then  compute an ordering and run the fused lasso  leading to our CFL1 estimator. This is  depicted in Figure \ref{fig4}. There, we also  see the estimates based on the method ACW from \cite{abadie2018endogenous}. Both CFL1 and ACW estimate small positive treatment effects, which is consistent with expectations given the nature of the NSW program—participants in the treatment group would be expected to benefit, in terms of future earnings, from the work experience they received. Interestingly, from an interpretability perspective, CFL1 estimates a single group effect, suggesting little to no meaningful heterogeneity in treatment effects. This finding is broadly in line with ACW, which reports slightly varying effects across four prespecified groups—groups that are not derived from the data but fixed in advance. This contrast might highlight a key strength of CFL1: its ability to adaptively detect underlying structure (or confirm its absence) directly from the data, rather than relying on predefined subgroup classifications.


\subsubsection{Observational  example}
\label{sec:real2}

For our second example based on the NSW data, we  combine the  185  observations in the treatment group  of the data from Section \ref{sec:real1} with the largest of the six observational control groups  constructed by Lalonde\footnote{Dataset is available here http://users.nber.org/~rdehejia/nswdata2.html}. This results in a total of  16177 samples. Due to the observational nature of the dataset,  we  run our propensity score based estimator  from Section \ref{sec:propensity} by only estimating treatment effects on the treated. Thus, our estimator is the one described in Corollary \ref{cor1} which we denote as CFL2.  
As shown in Figure \ref{fig5}, CFL2 estimates a constant positive treatment effect, consistent with our findings in Section \ref{sec:real1}. Specifically, CFL2 again estimates a single group effect, suggesting that there may be little to no heterogeneity in treatment effects—though a small positive effect is still detected.

\subsection{National Health and Nutrition Examination Survey}

In our final example, we use data from the   2007–2008 National Health and Nutrition Examination Survey (NHANES). The data consist of  2330 children and their participation in the National School Lunch or the School Breakfast programs in order to assess
the effectiveness of such meal programs in increasing  body mass index (BMI). In the study 1284 randomly selected children participated in the meal programs while 1046  did not. The predictor
variables are  age, gender, age of adult respondent, and categorical variables such as Black
race, Hispanic race,  whether the family of the child is  above 200\% of the federal poverty level, participation in Special Supplemental Nutrition program, Participation in food stamp program,
childhood food security, any type of insurance, and gender of the adult respondent.

Similarly,  as before, we run our propensity score based estimator   (CFL2) and  show the results in Figure \ref{fig5}. We can see that the sign of estimated treatment effects varies depending on the value of the propensity score. The latter was estimated by  logistic regression. Our findings for the treatment effects coincide with several authors who found positive and negative average treatment effects as discussed in \cite{chan2016globally}. In particular, we find that when the estimated propensity score is low (below 0.34), the treatment effect is predominantly negative. This suggests that individuals who are unlikely to participate in the program may be more prone to experiencing adverse effects. Further inspection of the data reveals that all individuals with propensity scores below 0.34 come from families with incomes above 200\% of the federal poverty level.

For propensity scores in the range of 0.34 to 0.68, the estimated treatment effects are generally small or slightly positive, indicating that the meal programs may offer modest benefits to individuals with a moderate likelihood of participation.

In contrast, for subjects with higher propensity scores, the estimated treatment effects display greater heterogeneity. While some subgroups exhibit small positive effects, others show negative ones. Notably, the largest positively affected subgroup falls within the 0.87 to 0.91 propensity score range. However, individuals with propensity scores above 0.91 consistently exhibit negative treatment effects. Examining the data, we find that all individuals with propensity scores above 0.87 fall below the 200\% poverty threshold. Among those in the 0.87–0.91 range, only 13\% experienced childhood food insecurity, whereas 66\% of those above 0.91 did.

Overall, this analysis highlights the heterogeneous nature of the data and the value of our method in identifying subgroups with differing treatment responses.

\subsection{Right Heart Catheterization (RHC) Study}

We now consider a clinical example drawn from an observational study examining the effects of right heart catheterization (RHC) on short-term survival outcomes in critically ill patients admitted to an intensive care unit (ICU). RHC is a diagnostic procedure used to assess cardiac function by directly measuring pulmonary pressures and cardiac output. While it can inform treatment decisions, it also poses non-negligible procedural risks, and its clinical benefit has been the subject of ongoing debate. Due to the ethical and logistical challenges of conducting randomized trials in this setting, observational data has been widely used to evaluate the causal effect of RHC on mortality outcomes.

The data consists of $2,707$ patients, with $1,103$ receiving RHC treatment within the first 24 hours of admission ($Z=1$), and $1,604$ not receiving the procedure ($Z=0$). The binary outcome $Y$ indicates whether the patient died within 180 days of hospital admission. Each patient is characterized by a set of $72$ covariates, which include continuous measurements, binary indicators, and dummy variables derived from categorical attributes. These variables capture demographic, clinical, and treatment-related characteristics relevant to patient prognosis.

We apply the propensity score-based fused lasso estimator CFL2 as defined in Section~\ref{sec:propensity} (Equation~(\ref{eqn:fused_lasso5})), focusing on the estimation of treatment effects for the treated population. Following standard preprocessing steps from prior clinical studies, we first exclude all patient records containing missing covariate values. Among the remaining 2,707 individuals, we drop any covariates with no variability and transform all categorical variables into dummy variables using one-hot encoding. This results in a covariate matrix with 72 features, including continuous, binary, and dummy-encoded categorical variables. We then fit a logistic regression model to estimate propensity scores using the full set of encoded covariates as predictors.

\begin{figure}[!htb]
\centering
\includegraphics[width=3.58in,height=4.in]{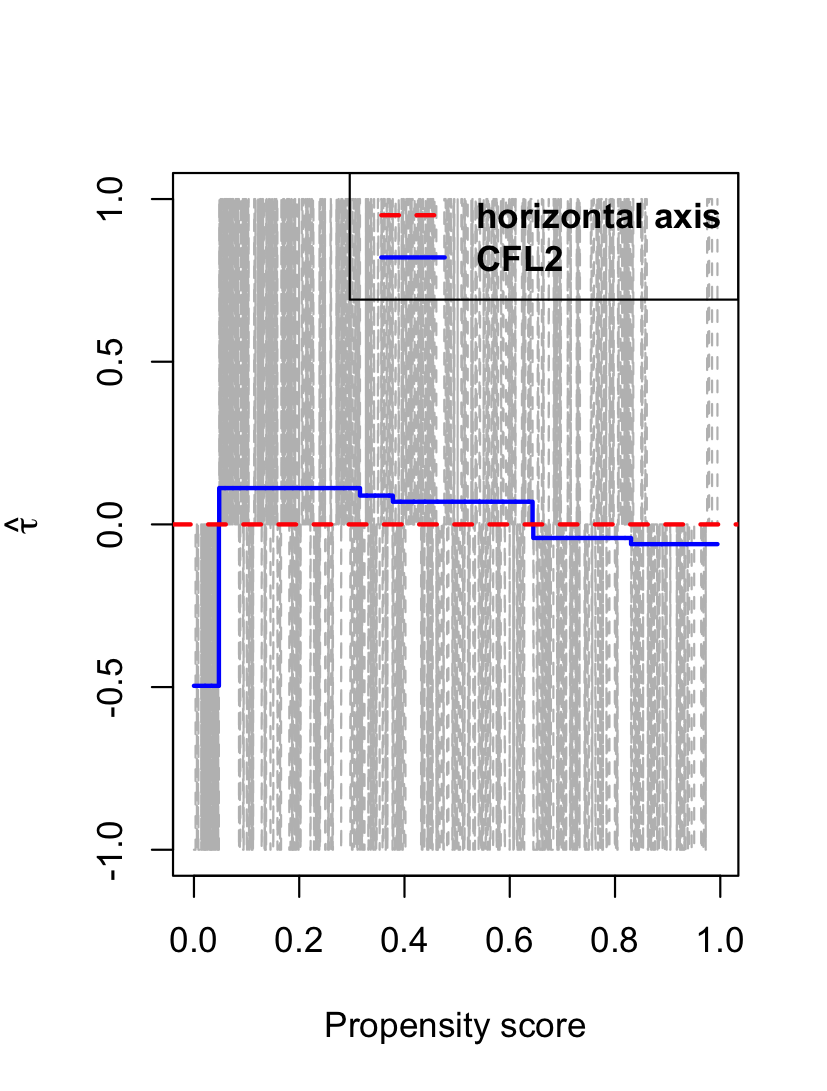}
\caption{Estimated treatment effects across the range of estimated propensity scores using CFL2. The stepwise structure of the estimated effects highlights heterogeneous responses to RHC treatment.}
\label{fig:rhc}
\end{figure}

Figure~\ref{fig:rhc} displays the estimated treatment effects across the range of estimated propensity scores using the CFL2 estimator. We observe that for most strata of the population, identified by similar estimated propensity scores, the estimated effect of RHC is either negative or close to zero. This suggests that RHC may offer limited or no survival benefit for the majority of patients, and may even be associated with worse short-term outcomes in some subgroups. These findings support the view that the use of RHC should be carefully evaluated on a patient-specific basis. The estimated step function highlights the ability of CFL2 to adaptively identify subgroups with differing treatment responses without imposing pre-specified strata. This interpretation is in line with prior work such as \cite{abadie2018endogenous}, which considers population partitions into a small number of subgroups with distinct average treatment effects. Our estimator adaptively detects such subgroup heterogeneity from data without requiring these partitions to be specified in advance. The presence of clear jumps in the estimated treatment effect across strata (as shown in Figure~\ref{fig:rhc}) aligns with this interpretation and illustrates the ability of CFL2 to flexibly model heterogeneity in observational settings.

\section{ Conclusion}

In this paper, we studied two methods for estimating heterogeneous treatment effects. The first approach, based on the prognostic score, is designed for randomized experiments. It involves constructing a prognostic score and then applying the fused lasso, using as input a noisy estimate of the treatment effect derived from matching, with the prognostic score serving as the covariate.

The second approach, which relies on the propensity score, is suitable for observational studies. It follows the same general structure as the first method, but substitutes the propensity score in place of the prognostic score.

A key strength of both methods lies in their simplicity and their usefulness as exploratory tools. For each, we provide theoretical guarantees in the form of finite-sample bounds on the mean squared error for estimating heterogeneous treatment effects. However, as with much of the existing literature on the fused lasso, our methods do not currently offer confidence bands with theoretical guarantees. One promising direction for addressing this limitation is to explore residual bootstrap techniques \citep{efron1992bootstrap}, potentially incorporating ideas from the nonparametric approach in \cite{padilla2024confidence}. We leave this extension for future research. 


On a related note, since both of our methods accommodate treatment effects that may exhibit discontinuities as a function of the score, it is natural to consider applications where such discontinuities arise organically. As one reviewer suggested, this setting might be particularly relevant in the social sciences \citep{wong2010addressing}, where individuals may become eligible for various forms of assistance or educational programs upon exceeding specific thresholds—such as standardized test scores or income cutoffs. 

Finally, another promising direction is to extend the results of Theorems~\ref{thm:prognostic_upper_bound} and~\ref{thm:tv} to out-of-sample settings. Specifically, for Theorem~\ref{thm:prognostic_upper_bound}, suppose $X_{n+1}$ is drawn independently from the same distribution as $X_1, \ldots, X_n$. We conjecture that it is possible to derive an upper bound on
\begin{equation}
\label{eqn:conjecture}
\mathbb{E}(  (\rho_{n+1}^* - \hat{\tau}_{n+1} )^2 )
\end{equation}
where  $\rho_{n+1}^* =  E\{Y (1)|g(X) =  g(X_{n+1}),  Z = 1\} -  E\{ Y (0)|g(X) =  g(X_{n+1}),  Z = 0\}$, and $\hat{\tau}_{n+1}$ is constructed via interpolation6y using the trained estimators $\hat{\tau}_1, \ldots, \hat{\tau}_n$. We expect that the upper bound for the quantity in (\ref{eqn:conjecture}) will match, up to logarithmic factors, the upper bound in (\ref{eqn:up1}). However, a formal proof of this result is nontrivial and is left for future work. We also conjecture that an analogous result may hold for Theorem~\ref{thm:tv}, replacing the prognostic score with the propensity score.

\appendix

\newpage
\section{Additional numerical results}
\label{add-simu-results}

In this appendix, we present an additional synthetic experiment, labeled \textit{Scenario 9}. This data-generating process is designed to challenge estimators with both strong nonlinearity in the outcome functions and treatment effect heterogeneity that interacts with the covariates in a complex, nonlinear fashion. In particular, it combines nonlinear transformations, sinusoidal interactions, and a nontrivial treatment assignment mechanism driven by a noisy logistic function of covariates.

\vspace{-4mm}
\begin{table}[h!]
\centering
\caption{\label{tab3} Performance evaluations (median $\pm$ standard error) over 50 Monte Carlo simulations for the additional nonlinear synthetic Scenario 9. $(n,d) = (4000, 10)$. \textbf{Bold} indicates the best method, and \textit{italic} indicates the second-best.}
\medskip
{\fontsize{13}{18}\selectfont
\setlength{\tabcolsep}{5pt}
\begin{tabular}{|l|c|}
\hline
Method & Scenario 9 \\
\hline
CFL1 & 0.267 $\pm$ 0.046 \\
\hline
CFL2 & \textit{\textbf{0.189 $\pm$ 0.039}} \\
\hline
WA1 & 0.395 $\pm$ 0.063 \\
\hline
WA2 & 0.319 $\pm$ 0.052 \\
\hline
GRF & 0.248 $\pm$ 0.045 \\
\hline
ACW & 0.401 $\pm$ 0.066 \\
\hline
BART & \textbf{0.174 $\pm$ 0.036} \\
\hline
AIPW & 0.305 $\pm$ 0.050 \\
\hline
GM & 0.362 $\pm$ 0.058 \\
\hline
MALTS & 0.331 $\pm$ 0.054 \\
\hline
LCM & 0.352 $\pm$ 0.057 \\
\hline
ADD-MALTS & 0.201 $\pm$ 0.041 \\
\hline
AHB & 0.369 $\pm$ 0.059 \\
\hline
FLAME & 0.402 $\pm$ 0.063 \\
\hline
DAME & 0.384 $\pm$ 0.061 \\
\hline
\end{tabular}
}
\end{table}
\vspace{-3.5mm}

\textit{Scenario 9.} This synthetic is inspired by nonlinear regression examples commonly used in benchmarking flexible estimators. The data is generated as follows, for $(n,d) = (4000,10)$:
\[
\begin{aligned}
& x_{i,1}, \ldots, x_{i,10} \overset{\mathrm{iid}}{\sim} \mathcal{U}(0,1), \quad \epsilon_{i,(0)}, \epsilon_{i,(1)}, \epsilon_{i,(\mathrm{treat})} \overset{\mathrm{iid}}{\sim} \mathcal{N}(0,1), \\
& Y_i(0) = 10 \sin \left( \pi x_{i,1} x_{i,2} \right) + 20(x_{i,3} - 0.5)^2 + 10 x_{i,4} + 5 x_{i,5} + \epsilon_{i,(0)}, \\
& Y_i(1) = Y_i(0) + x_{i,3} \cos \left( \pi x_{i,1} x_{i,2} \right) + \epsilon_{i,(1)}, \\
& Z_i = \boldsymbol{1}_{\left\{ \mathrm{expit}(x_{i,1} + x_{i,2} - 0.5 + \epsilon_{i,(\mathrm{treat})}) > 0.5 \right\}}, \\
& Y_i = (1 - Z_i) Y_i(0) + Z_i Y_i(1),
\end{aligned}
\]
where $\mathrm{expit}(u) = 1 / (1 + e^{-u})$ denotes the logistic sigmoid function.

This setup introduces modeling difficulties not only through the nonlinearity of the outcome regression, but also by encoding treatment effects that depend on both the covariates and their interactions. As shown in Table~\ref{tab3}, even in this challenging setting, our proposed estimator CFL2 performs competitively and achieves the second-best performance overall, closely trailing BART. Notably, CFL2 outperforms all other benchmark methods by a clear margin. This result highlights the flexibility and generalization strength of our proposed method under complex nonlinear conditions.

For completeness, Table~\ref{tab11} in the appendix summarizes the performance of all estimators across Scenarios 1–4.

\newpage

\begin{table}[h!]
\centering
\caption{\label{tab11} Performance evaluation (median $\pm$ standard error)  for synthetic scenarios with varying $(n,d)$. \textbf{Bold} indicates the best, and \textit{italic} indicates the second-best.}
\medskip
{\fontsize{8.5}{7.5}\selectfont
\setlength{\tabcolsep}{6pt}
\begin{tabular}{|l|l|l|l|l|l|}
\hline
Method & $(n,d)$ & Scenario 1 & Scenario 2 & Scenario 3 & Scenario 4 \\
\hline
\multirow{4}{*}{CFL1} 
& (800, 2)  & \textbf{0.004 $\pm$ 0.0012} & 0.195 $\pm$ 0.055 & 0.181 $\pm$ 0.037 & \textbf{0.301 $\pm$ 0.065} \\
& (1600, 2) & \textbf{0.003 $\pm$ 0.0009} & 0.108 $\pm$ 0.032 & 0.136 $\pm$ 0.028 & \textbf{0.183 $\pm$ 0.044} \\
& (800, 10) & \textbf{0.005 $\pm$ 0.0013} & \textbf{0.503 $\pm$ 0.083} & 0.412 $\pm$ 0.082 & \textbf{0.450 $\pm$ 0.086} \\
& (1600, 10)& \textbf{0.003 $\pm$ 0.0010} & \textbf{0.319 $\pm$ 0.069} & 0.293 $\pm$ 0.063 & \textbf{0.277 $\pm$ 0.068} \\
\hline
\multirow{4}{*}{CFL2} 
& (800, 2)  & 0.011 $\pm$ 0.0023 & * & \textbf{0.074 $\pm$ 0.017} & * \\
& (1600, 2) & 0.004 $\pm$ 0.0011 & * & 0.051 $\pm$ 0.014 & * \\
& (800, 10) & 0.016 $\pm$ 0.0038 & * & \textbf{0.146 $\pm$ 0.033} & * \\
& (1600, 10)& 0.005 $\pm$ 0.0016 & * & \textbf{0.109 $\pm$ 0.027} & * \\
\hline
\multirow{4}{*}{WA1} 
& (800, 2)  & 0.045 $\pm$ 0.011 & 0.495 $\pm$ 0.104 & 0.226 $\pm$ 0.044 & 4.923 $\pm$ 0.819 \\
& (1600, 2) & 0.029 $\pm$ 0.009 & 0.199 $\pm$ 0.057 & 0.164 $\pm$ 0.031 & 3.228 $\pm$ 0.607 \\
& (800, 10) & 0.067 $\pm$ 0.013 & 0.773 $\pm$ 0.139 & 0.550 $\pm$ 0.089 & 6.528 $\pm$ 0.908 \\
& (1600, 10)& 0.068 $\pm$ 0.012 & 0.534 $\pm$ 0.127 & 0.499 $\pm$ 0.086 & 6.146 $\pm$ 0.875 \\
\hline
\multirow{4}{*}{WA2} 
& (800, 2)  & 0.012 $\pm$ 0.0040 & 0.264 $\pm$ 0.069 & 0.143 $\pm$ 0.029 & 2.922 $\pm$ 0.486 \\
& (1600, 2) & 0.010 $\pm$ 0.0026 & 0.164 $\pm$ 0.043 & 0.106 $\pm$ 0.023 & 2.049 $\pm$ 0.417 \\
& (800, 10) & 0.007 $\pm$ 0.0021 & 0.794 $\pm$ 0.147 & 0.362 $\pm$ 0.063 & 5.796 $\pm$ 0.768 \\
& (1600, 10)& 0.003 $\pm$ 0.0015 & 0.723 $\pm$ 0.116 & 0.316 $\pm$ 0.059 & 5.322 $\pm$ 0.736 \\
\hline
\multirow{4}{*}{GRF} 
& (800, 2)  & 0.013 $\pm$ 0.0031 & \textbf{0.152 $\pm$ 0.048} & 0.143 $\pm$ 0.031 & 1.788 $\pm$ 0.318 \\
& (1600, 2) & 0.011 $\pm$ 0.0024 & \textbf{0.063 $\pm$ 0.018} & 0.106 $\pm$ 0.024 & 0.771 $\pm$ 0.177 \\
& (800, 10) & 0.010 $\pm$ 0.0032 & 0.565 $\pm$ 0.092 & 0.408 $\pm$ 0.078 & 3.261 $\pm$ 0.557 \\
& (1600, 10)& 0.006 $\pm$ 0.0021 & 0.350 $\pm$ 0.077 & 0.359 $\pm$ 0.070 & 1.291 $\pm$ 0.308 \\
\hline
\multirow{4}{*}{ACW} 
& (800, 2)  & 0.017 $\pm$ 0.0034 & 0.402 $\pm$ 0.088 & 0.213 $\pm$ 0.043 & 2.203 $\pm$ 0.378 \\
& (1600, 2) & 0.014 $\pm$ 0.0025 & 0.297 $\pm$ 0.067 & 0.181 $\pm$ 0.037 & 1.654 $\pm$ 0.324 \\
& (800, 10) & 0.029 $\pm$ 0.0056 & 0.638 $\pm$ 0.096 & 0.489 $\pm$ 0.081 & 4.394 $\pm$ 0.672 \\
& (1600, 10)& 0.024 $\pm$ 0.0049 & 0.412 $\pm$ 0.084 & 0.443 $\pm$ 0.075 & 3.612 $\pm$ 0.528 \\
\hline
\multirow{4}{*}{BART} 
& (800, 2)  & 0.009 $\pm$ 0.0026 & 0.184 $\pm$ 0.052 & 0.102 $\pm$ 0.026 & 1.842 $\pm$ 0.334 \\
& (1600, 2) & 0.006 $\pm$ 0.0017 & 0.095 $\pm$ 0.027 & 0.075 $\pm$ 0.019 & 1.021 $\pm$ 0.266 \\
& (800, 10) & 0.021 $\pm$ 0.0048 & 0.592 $\pm$ 0.096 & 0.351 $\pm$ 0.068 & 3.482 $\pm$ 0.562 \\
& (1600, 10)& 0.015 $\pm$ 0.0035 & 0.412 $\pm$ 0.085 & 0.306 $\pm$ 0.060 & 2.693 $\pm$ 0.504 \\
\hline
\multirow{4}{*}{AIPW} 
& (800, 2)  & 0.015 $\pm$ 0.0040 & 0.235 $\pm$ 0.058 & 0.109 $\pm$ 0.027 & 1.918 $\pm$ 0.362 \\
& (1600, 2) & 0.010 $\pm$ 0.0025 & 0.132 $\pm$ 0.036 & \textbf{0.049 $\pm$ 0.015} & 1.089 $\pm$ 0.291 \\
& (800, 10) & 0.023 $\pm$ 0.0052 & 0.672 $\pm$ 0.112 & 0.342 $\pm$ 0.059 & 3.988 $\pm$ 0.611 \\
& (1600, 10)& 0.019 $\pm$ 0.0041 & 0.498 $\pm$ 0.096 & 0.286 $\pm$ 0.051 & 3.134 $\pm$ 0.552 \\
\hline
\multirow{4}{*}{GM} 
& (800, 2)  & 0.020 $\pm$ 0.0051 & 0.300 $\pm$ 0.064 & 0.202 $\pm$ 0.041 & 2.421 $\pm$ 0.394 \\
& (1600, 2) & 0.005 $\pm$ 0.0013 & 0.187 $\pm$ 0.048 & 0.072 $\pm$ 0.020 & 1.618 $\pm$ 0.351 \\
& (800, 10) & 0.035 $\pm$ 0.0067 & 0.689 $\pm$ 0.103 & 0.489 $\pm$ 0.075 & 5.088 $\pm$ 0.779 \\
& (1600, 10)& 0.0062 $\pm$ 0.0015 & 0.522 $\pm$ 0.091 & 0.088 $\pm$ 0.023 & 4.523 $\pm$ 0.721 \\
\hline
\multirow{4}{*}{MALTS} 
& (800, 2)  & 0.018 $\pm$ 0.0042 & 0.248 $\pm$ 0.057 & 0.098 $\pm$ 0.023 & 2.156 $\pm$ 0.391 \\
& (1600, 2) & 0.011 $\pm$ 0.0031 & 0.112 $\pm$ 0.028 & \textit{\textbf{0.050 $\pm$ 0.014}} & 1.141 $\pm$ 0.283 \\
& (800, 10) & 0.030 $\pm$ 0.0061 & 0.618 $\pm$ 0.102 & 0.238 $\pm$ 0.051 & 3.274 $\pm$ 0.553 \\
& (1600, 10)& 0.024 $\pm$ 0.0054 & 0.401 $\pm$ 0.087 & 0.212 $\pm$ 0.046 & 2.148 $\pm$ 0.472 \\
\hline
\multirow{4}{*}{LCM} 
& (800, 2)  & 0.014 $\pm$ 0.0038 & 0.248 $\pm$ 0.051 & 0.129 $\pm$ 0.028 & 1.697 $\pm$ 0.323 \\
& (1600, 2) & 0.0042 $\pm$ 0.0012 & 0.157 $\pm$ 0.035 & 0.069 $\pm$ 0.018 & 0.996 $\pm$ 0.262 \\
& (800, 10) & 0.019 $\pm$ 0.0043 & 0.562 $\pm$ 0.089 & 0.278 $\pm$ 0.054 & 2.643 $\pm$ 0.497 \\
& (1600, 10)& 0.0055 $\pm$ 0.0014 & 0.384 $\pm$ 0.077 & 0.081 $\pm$ 0.021 & 1.839 $\pm$ 0.421 \\
\hline
\multirow{4}{*}{ADD-MALTS} 
& (800, 2)  & 0.015 $\pm$ 0.0040 & 0.276 $\pm$ 0.059 & 0.122 $\pm$ 0.026 & 1.823 $\pm$ 0.332 \\
& (1600, 2) & 0.0045 $\pm$ 0.0011 & 0.169 $\pm$ 0.036 & 0.063 $\pm$ 0.017 & 1.184 $\pm$ 0.271 \\
& (800, 10) & 0.017 $\pm$ 0.0038 & 0.588 $\pm$ 0.095 & 0.256 $\pm$ 0.051 & 2.392 $\pm$ 0.474 \\
& (1600, 10)& 0.0049 $\pm$ 0.0013 & 0.392 $\pm$ 0.081 & 0.072 $\pm$ 0.019 & 1.601 $\pm$ 0.409 \\
\hline
\multirow{4}{*}{AHB} 
& (800, 2)  & 0.010 $\pm$ 0.0028 & 0.231 $\pm$ 0.051 & 0.117 $\pm$ 0.025 & 1.503 $\pm$ 0.297 \\
& (1600, 2) & \textbf{0.003 $\pm$ 0.0009} & 0.136 $\pm$ 0.031 & 0.081 $\pm$ 0.017 & 0.908 $\pm$ 0.239 \\
& (800, 10) & 0.013 $\pm$ 0.0034 & 0.492 $\pm$ 0.083 & 0.234 $\pm$ 0.048 & 2.182 $\pm$ 0.459 \\
& (1600, 10)& 0.005 $\pm$ 0.0016 & 0.331 $\pm$ 0.071 & 0.191 $\pm$ 0.041 & 1.523 $\pm$ 0.382 \\
\hline
\multirow{4}{*}{FLAME} 
& (800, 2)  & 0.041 $\pm$ 0.009 & 0.298 $\pm$ 0.069 & 0.319 $\pm$ 0.062 & 4.121 $\pm$ 0.754 \\
& (1600, 2) & 0.030 $\pm$ 0.008 & 0.219 $\pm$ 0.052 & 0.271 $\pm$ 0.056 & 3.589 $\pm$ 0.645 \\
& (800, 10) & 0.064 $\pm$ 0.012 & 0.743 $\pm$ 0.135 & 0.613 $\pm$ 0.089 & 6.621 $\pm$ 0.922 \\
& (1600, 10)& 0.060 $\pm$ 0.011 & 0.701 $\pm$ 0.127 & 0.538 $\pm$ 0.082 & 6.179 $\pm$ 0.832 \\
\hline
\multirow{4}{*}{DAME} 
& (800, 2)  & 0.043 $\pm$ 0.010 & 0.334 $\pm$ 0.074 & 0.307 $\pm$ 0.058 & 4.308 $\pm$ 0.737 \\
& (1600, 2) & 0.038 $\pm$ 0.009 & 0.293 $\pm$ 0.069 & 0.278 $\pm$ 0.054 & 3.945 $\pm$ 0.692 \\
& (800, 10) & 0.072 $\pm$ 0.014 & 0.803 $\pm$ 0.141 & 0.622 $\pm$ 0.092 & 6.734 $\pm$ 0.941 \\
& (1600, 10)& 0.075 $\pm$ 0.013 & 0.762 $\pm$ 0.134 & 0.577 $\pm$ 0.088 & 6.502 $\pm$ 0.910 \\
\hline

\end{tabular}
}
\end{table}

\newpage
\section{Possible extensions}

A natural  extension of the  estimators  described in  Sections \ref{sec:prognostic}  and  \ref{sec:propensity} is to consider the case  where  the  number of  covariates can be large, perhaps  $d >>n$, but only a small number of them plays a role in the prognostic score (propensity score).  In the case of the prognostic score based estimator, it is reasonable to estimate  $g$ with lasso    regression \citep{tibshirani1996regression}.  The resulting  procedure  would be the same as in Section \ref{sec:prognostic}, except that we would define $\hat{g}(x) =  x^{\top} \hat{\theta}$ where 
\[
\displaystyle  \hat{\theta}   \,=\,    \underset{ \theta \in \mathbb{R}^d }{\arg \min }\,  \,  \left\{\frac{1}{m} \sum_{i=1}^m  \left(Y_i^{\prime} -  X_i^{\prime \top} \theta  \right)^2  +   \nu \sum_{ j=1  }^{d} \vert \theta_j \vert \right\}, 
\]
for a tuning parameter  $\nu>0$. Similarly, we can modify the  propensity  score  estimator, replacing $\hat{e}$ by $\ell_1$-regularized logistic regression in the spirit of \cite{ravikumar2010high}.

A simpler modification of the estimators from Sections \ref{sec:prognostic}--\ref{sec:propensity}  can be obtained  by adding a sparsity penalty in the objective function. This is similar to the definition of the fused lasso  in \cite{tibshirani2005sparsity}. The resulting estimator would be reasonable if there is a belief that most of the treatment effects are zero. It would basically  amount to apply  soft-thresholding to the estimators from  (\ref{eqn:fused_lasso3})--(\ref{eqn:fused_lasso5}), see for instance  \cite{wang2016trend}.

While the definition of our estimators naturally extends to high-dimensional settings, the more challenging task lies in analyzing their statistical properties. Our current theoretical results are limited to scenarios where the fused lasso is combined with parametric estimation of the prognostic or propensity score, which are tailored for low-dimensional covariate spaces. Extending the theory to high-dimensional contexts remains an important direction for future work.  

\section{Main result for propensity score based estimator}
\label{sec:propensity2}

We now  study  the statistical  properties of the  estimator  defined in Section \ref{sec:propensity}.  As in Section \ref{sec:theory2} we start by stating  required assumptions. 

\begin{assumption}[Sub-Gaussian errors]
\label{as2}
Define   $V(z, x) = E\lbrace Y | Z=z, e(X)=e(x)\rbrace$ and $\epsilon_i = Y_i - V(Z_i, X_i)$ for  $i=1,\ldots,n$, with  $e(\cdot)$ as in Assumption \ref{as1}. Then  the vector  $ (\epsilon_1,\ldots,\epsilon_n)^{\top}$  has independent coordinates that are mean zero sub-Gaussian$(v)$ for some constant $v>0$.   
\end{assumption}

Assumption \ref{as2} parallels of Assumption \ref{as8} when we replace the prognostic  score with the propensity score. 

\begin{assumption}
\label{as3} The functions $f_1(s)   =  E\{Y |  Z=1,\, e(X)   =  s  \}$, and  $f_0(s)   =  E\{Y|  Z=0,\, e(X)   =  s  \}$ for  $s \in [e_{\min},e_{\max}]$ are bounded and have  bounded  variation. 
\end{assumption}

As for the distribution of the covariates, we allow  for more generality than in Section \ref{sec:theory2}. Specifically, we  allow for general  multivariate sub-Gaussian distributions. 

\begin{assumption}[Distribution of covariates]
\label{as4}
The random vector  $X \in  \mathbb{R}^d$ is centered ($E(X)=0$) sub-Gaussian(C). 

\end{assumption}

We refer the reader to  \cite{vershynin2010introduction}  which contains important concentrations  results regarding multivariate sub-Gaussian distributions. 

\begin{assumption}
\label{as5}
The propensity score  staisfies  $e(X) :=   F( X^{\top} \theta^*)$ for some $\theta^* \in   \mathbb{R}^d $, where $F(x)=\exp(x)/\{1+\exp(x)\}$  as before. Furthermore  $e(X)$ is a continuous random variable   with  pdf  $h(\cdot)$  bounded by above  ($\|h\|_{\infty }  =  h_{\max}$ for some positive constant $h_{\max}$),  and 	
there exist  constants $a_1$ and $a_2$ such that
\begin{equation*}
a_1  t  \,\leq\,    \mathbb{P}(   \vert   e(X)  -   b   \vert \leq t  )   \,\leq \,a_2 t
\end{equation*}
for all $b $ in the support of $   e(X) $ and $t \in (0,t_0)$, where  $t_0>0$ is a constant.
\end{assumption}



\begin{assumption}[Dependency condition]
\label{as6}
Let  $\Lambda_{\min}(\cdot)$ and  $\Lambda_{\max}(\cdot)$  be the minimum  and maximum eigenvalue  functions, respectively. We assume that there exist
positive  $C_{\min}$ and $D_{\max}$ such that
\[
\Lambda_{\min}\left(E\left[ \{ \eta( X^{\top}\theta^*) X X^{\top}   \}   \right ] \right) \,>\, C_{\min } >  c_1\|X\|_{ \psi_2 }^2 \left(\frac{ d \log m }{m}\right)^{1/2},
\]
and
\[
\Lambda_{\max}\{E(   X X^{\top}  )  \}  \,<\, D_{\max},  
\]
with  $\eta(t) =F(t)\{1-F(t)\}$ with $F$ as in  Assumption \ref{as5}, and where $c_1 >0$ is a constant.
We also  require  that 
\begin{equation}
\label{eqn:low_b}
\frac{C_{ \min }^2}{D_{\max}   }  >   c_2\frac{d\log m}{ \sqrt{m} },
\end{equation}	
for a large enough constant $c_2>0$.
\end{assumption}

Assumption \ref{as6} is basically  the Dependency condition in the analysis of high-dimensional logistic regression from  \cite{ravikumar2010high}. As the authors there assert,  this condition prevents the covariates from becoming overly dependent.





We are now in position to present  the main result regarding our propensity scored based estimator.

\begin{theorem}
\label{thm:tv2}
Under Assumptions  \ref{as1}, and \ref{as2}--\ref{as6}, $   dn^{1/2} \geq C_{\min}$, $n \asymp m$,   there exists  $t>0$  such that
\begin{equation}
\label{eqn:scaling}
t\,\asymp\, \max\left\{ \frac{dn}{ C_{\min} } \frac{  \log^{1/2} m\,\log^{1/2} (nd) }{m^{1/2}}   , \log n\right\} 
\end{equation}
and choice of  $\lambda$ satisfying 
\[
\lambda  \,\asymp\,  n^{1/3}   (\log n)^{2}   (\log \log n) t^{-1/3},
\]
such that the estimator  $\hat{\tau}$ defined in (\ref{eqn:fused_lasso5}) satisfies
\begin{equation}
\label{eqn:up2}
\displaystyle \frac{1}{n}  \sum_{i=1}^{n} ( \rho^*_i  - \hat{\tau}_i )^2 \,=\,  O_{ \mathbb{P} }\left\{    \frac{ d^{2/3}  (\log n)^3 (\log \log n)}{C_{\min}^{2/3  }  n^{1/3 }}  \right\},
\end{equation}
where  $\rho_i^* =  E\{ Y (1)\,|  \,e(X) =  e(X_i),  Z = 1\} - E\{Y (0)\,|\,e(X) =  e(X_i),  Z = 0\}$  for  $i=1,\ldots,n$. If in addition   $Y(0), Y(1) \,\ind\,   Z   \,|\,e(X)$, then  (\ref{eqn:up2})  holds replacing $\rho_i^*$ with 
$\tau_i^* =  E\{ Y(1) - Y(0) \,|\,e(X)=e(X_i)\}$  for  $i=1,\ldots,n$.  
\end{theorem}

Importantly, Theorem \ref{thm:tv} implies that the estimator $\hat{\tau}$ defined in (\ref{eqn:fused_lasso5})  can consistently estimate the subgroup treatment effects  $\tau^*$ under general conditions.  One of such conditions is that  $Y(0), Y(1) \,\ind\,   Z   \,|\,e(X)$, which in the language of  \cite{rosenbaum1983central} means that treatment is strongly ignorable given $e(\cdot)$. As Theorem 3 in \cite{rosenbaum1983central} showed,  $Y(0), Y(1) \,\ind\,   Z   \,|\,e(X)$ holds under overlapping (Assumption \ref{as1}) and unconfoundedness   which can be written as $Y(0), Y(1) \,\ind\,   Z   \,|\,X$.  When these conditions are violated, Theorem \ref{thm:tv} shows that $\hat{\tau}$ can still approximate $\rho^*$ under Assumptions  \ref{as1}, and \ref{as2}--\ref{as6}.



Furthermore,   as in Remark \ref{remark1},    Theorem \ref{thm:tv}   can be relaxed.  Specifically,  we  can replace Assumption \ref{as3}  with 
\[
t^* \,:=\,  \max\{   \mathrm{TV}(f_0,n),\mathrm{TV}(f_1,n) \} .
\]
Then the upper bound in Theorem \ref{thm:tv} needs to be inflated by  $(t^*)^{2/3}$.

We conclude this section with immediate consequence of the proof of Theorem \ref{thm:tv}  concerning  heterogenous treatment effects of the treated units.

\begin{corollary}[Treatment effects of the treated]
\label{cor1}
Suppose that the conditions of  Theorem \ref{thm:tv} for  (\ref{eqn:up2}) to hold are met.  Let  $\hat{\tau}$ be the propensity score  estimator  from  Section  \ref{sec:propensity}  with a slight modification. After the matching is done and the signal  $Y-\tilde{Y}$ is calculated, we only run the the fused lasso estimator,  with the  ordering based on the estimated propensity score, on the  treated units. Then
\begin{equation}
\label{eqn:up3}
\displaystyle \frac{1}{n}  \sum_{i\,:\,  Z_i=1} ( \rho^*_i  - \hat{\tau}_i )^2 \,=\,  O_{\mathbb{P} }\left\{    \frac{ d^{2/3}   (\log n)^3 (\log \log n)}{C_{\min}^{2/3  }  n^{1/3 }}  \right\},
\end{equation}
where  $\rho_i^* =  E\{Y (1)\,|  \,e(X) =  e(X_i),  Z = 1\} - E\{ Y (0)\,|\,e(X) =  e(X_i),  Z = 0\}$  for  $i=1,\ldots,n$. If in addition $Y(0)\,\ind\,   Z   \,|\,e(X)$, then  (\ref{eqn:up3})  holds replacing $\rho_i^*$ with 
$\tau_i^* =  E\{ Y(1) - Y(0) \,|\,e(X)=e(X_i) ,\,Z= 1    \},$  for  $i=1,\ldots,n$.  
\end{corollary}

\section{Proof of   Theorem \ref{thm:tv} }

Throughout this section  we write 
\[
\hat{L}(\theta) =     - \sum_{i=1}^{m} \frac{1}{m}   \left[ Z_i^{\prime} \log  F( X_i^{\prime \top} \theta )  +  (1-Z_i^{\prime})\log \{1-  F(  X_i^{\prime \top} \theta )\}\right],
\]
and 
\[
\delta  \,:=\, \frac{d}{ C_{\min} }\frac{  (\log^{1/2} m)\log^{1/2} (nd)}{m^{1/2}}.
\]
We also define the first order matrix $\Delta^{(1)} \in \mathbb{R}^{(n-1) \times n}$, such that for any $b \in \mathbb{R}^{n}$, the following holds:
\[
\|\Delta^{(1)}b\|_1\,=\,  \sum_{i=1}^{n-1} \vert(\Delta^{(1)}b)_i\vert\,=\, \sum_{i=1}^{n-1} \vert b_i -b_{i+1}\vert.  
\]
Hence, with this notation, the estimator defined in (\ref{eqn:fused_lasso5}) becomes 
\begin{equation}
\label{eqn:delta_version}
\hat{\tau}\,=\,   \underset{b \in \mathbb{R}^n }{\arg \min } \left\{ \frac{1}{2}\sum_{i=1}^{n} (Y_i -   \widetilde{Y}_i   + (-1)^{ Z_i }  b_i )^2  \,+\, \lambda \| \Delta^{(1)}\hat{P} b \|_1 \right\}.
\end{equation}

\subsection{Total variation auxiliary lemmas}

\begin{lemma}
\label{lem1}
The  estimator $\hat{\tau}$   defined  in  (\ref{eqn:fused_lasso5}) satisfies
\[
\|\hat{\tau}\|_{\infty}  \,= \, O_{\mathbb{P}}\left(  \max\{ \|  f_0\|_{\infty},\|  f_1\|_{\infty} \}  +   \lambda  +  \log^{1/2} n \right).
\]
\end{lemma}

\begin{proof}
We beging by introducing some notation.	For  a vector  $x \in  \mathbb{R}^s$,  a vector  $\mathrm{sign}(x)  \in  \mathbb{R}^s$ is defined  as
\[
	(\mathrm{sign}(x ))_i \,=\, \begin{cases}
1     & \text{if }\,\,   x_i >0\\
-1     & \text{if }\,\,   x_i <0\\
\in [0,1]     & \text{otherwise.}\,\,  \\
\end{cases}
\]
To proceed,  we first condition  on $X$ and $Z$. Then, using Equation (\ref{eqn:delta_version}) along with the KKT conditions, we obtain:
\begin{equation}
\label{eqn:kkt}
Y -   \widetilde{Y}   + (-1)^{ Z } \circ \hat{\tau}        +   \lambda \hat{P}^T ( \Delta^{(1)})^{\top}\mathrm{sign}(\Delta^{(1)} \hat{P} \hat{\tau} )  \,=\,0,
\end{equation}
where $\circ$ is the  Hadamard product. 
Next, notice  that (\ref{eqn:kkt})   implies that
\[
\begin{array}{lll}
\|  \hat{\tau}\|_{\infty}    & \leq & 2\|Y\|_{\infty}   +   2\lambda\\
& \leq&   2 \max\{ \|  f_0\|_{\infty},\|  f_1\|_{\infty} \}  +   2\lambda +    2 \|\epsilon\|_{\infty}.
\end{array}
\]
The claim then  follows since  $\|\epsilon\|_{\infty} = O_{\mathbb{P}}( \log^{1/2} n )$, by Sub-Gaussian  tail inequality, and integrating over $X$ and $Z$.
\end{proof}

\begin{lemma}
\label{lem9}
Assumptions \ref{as1} and \ref{as2} imply that
\[
\displaystyle	 \left\vert \frac{1}{n^{1/2}  }  \sum_{i=1}^{n} (-1)^{Z_i+1}\epsilon_{N(i)}  \right\vert^2 \,=\, O_{\mathbb{P}}( \log^3 n  ).
\]
\end{lemma}

\begin{proof}
Let $S_i = \{j  \in  \{1,\ldots,n\} \,:\, N(j) =i  \}$, we have
\[
\begin{array}{lll}
\displaystyle	 \left\vert \frac{1}{n^{1/2} }  \sum_{i=1}^{n} (-1)^{Z_i+1}\epsilon_{N(i)}  \right\vert^2 \,=\,  \left\vert \frac{1}{n^{1/2}  }  \sum_{i=1}^{n} \left\{ \sum_{j \in S_i } (-1)^{Z_j+1} \right\}\epsilon_{i}  \right\vert^2,
\end{array}
\]
and
\[
\frac{1}{n}\sum_{i=1}^{n} \left\{ \sum_{j \in S_i } (-1)^{Z_j+1} \right\}^2 \,\leq \,   \left(  \underset{  i =1,\ldots,n}{\max} \, \vert S_i \vert  \right)^2.
\]
On the other hand, since  $\hat{e}(X_1),\ldots, \hat{e}(X_n)$ are independent and identically distributed, we have that 
\[
(  \vert S_1\vert,\ldots, \vert S_n\vert  )\,\sim \,  \text{Multinomial}\left(n; \frac{1}{n},\ldots,\frac{1}{n}   \right).
\]
Therefore,  by  Chernoff's inequality and union bound,
\[
\left(  \underset{  i =1,\ldots,n}{\max} \, \vert S_i \vert  \right)^2  \,=\,  O_{\mathbb{P}}( \log^2 n  ),
\]
and so the claim follows by the sub-Gaussian  tail inequality.

\end{proof}

\begin{lemma}
\label{thm1}
Let   $\tau_i^* =  E\{Y_i |e(X_i),  Z_i  = 1\} - E\{Y_i | e(X_i), Z_i  = 0\}$   be  the  treatment effect for unit $i$.  Suppose  that  $\hat{e} $ is independent of $\{(Y_i,X_i,Z_i)\}_{i=1}^n$ and satisfies that  for $\hat{ \sigma}$ as defined in (\ref{eqn:permutation_hat2}) we write 
\begin{equation}
\label{eqn:tv3}
\displaystyle 	 \mathrm{tv}_1\,:= \,\max\left\{   \sum_{i=1}^{n-1}  \left\vert f_0\{e( X_{ \hat{\sigma}(i) } ) \}  - f_0\{e( X_{ \hat{\sigma}(i+1) } ) \}\right\vert ,\,\, \sum_{i=1}^{n-1}  \left\vert f_1\{ e( X_{ \hat{\sigma}(i) } ) \}  - f_1\{e( X_{ \hat{\sigma}(i+1) } ) \}\right\vert  \right\},
\end{equation}
and
\begin{equation}
\label{eqn:tv2}
\mathrm{tv}_2 \,:=\, \max\left\{ 	\sum_{i=1}^{n} \vert   f_0\{  e(X_{i}  ) \}  -  f_0\{ e(X_{N(i)} ) \} \vert,	\,\,\sum_{i=1}^{n} \vert   f_1\{ e(X_{i}  ) \}  -  f_1\{ e(X_{N(i)} )\} \vert \right\},
\end{equation}
and assume that these  random sequences satisfy $\max\{  \mathrm{tv}_1,\mathrm{tv}_2\} = O_{\mathbb{P}}(t)$  for a deterministic $t$  that can depend on $n$ and  can diverge.  Then, if  Assumptions  \ref{as1}, \ref{as2} and \ref{as3}  hold, and  
\begin{equation}
\label{eqn:choice_of_lambda}
\lambda  \,=\,  \Theta\left\{ n^{1/3}   (\log n)^{2}   (\log \log n) t^{-1/3}       \right\},
\end{equation}
we have that
\[
\displaystyle \frac{1}{n}  \sum_{i=1}^{n} ( \tau^*_i  - \hat{\tau}_i )^2 \,=\,  O_{ \mathbb{P} }\left\{ \frac{\log^3 n}{n} +   (\log n)^2 (\log \log n)\left( \frac{t}{n} \right)^{2/3}  \,+\, \frac{ t \log^{1/2} n  }{n}  \,+\, \frac{t^2}{n^2}    \right\}.
\]
\end{lemma}

\begin{proof}

Let  $R  =   \mathrm{row}(\Delta^{(1)}) $,  the row  space  of $\Delta^{(1)}$. Also  write  $P_{{R}}$ and  $P_{ {R}^{\perp} }$
for the orthogonal projection matrices onto $ {R}$ and  its orthogonal complement $ {R}^{\perp}$, respectively. Then let
\[
\tilde{\tau} \,  =\, \underset{ \tau \in {R}^n }{\arg \min}\, \left\{  \frac{1}{2}\|  P_{R}U -  \tau     \|^2  +  \lambda\| \Delta^{(1)}\hat{P}\tau  \|_1   \right\} ,
\]
where  $U_i \,=\,   (-1)^{Z_i+1}\left\{ Y_i -  Y_{N(i)}  \right\}$, for  $i= 1,\ldots,n$. Hence, $\hat{\tau} =  P_{ {R}^{\perp} } U  + \tilde{\tau}$, see Proof of Theorem  3 in \cite{wang2016trend}. 

Therefore,
\begin{equation*}
\| P_{ {R}^{\perp} }(\hat{\tau} - \tau^*)  \|^2  \, =\,  \displaystyle \left\vert \frac{1}{n^{1/2} }  \sum_{i=1}^{n} (U_i  -  \tau_i^*)  \right\vert^2.
\end{equation*}
We now study the quantity
\[
\left| \frac{1}{\sqrt{n}} \sum_{i=1}^n (U_i - \tau_i^*) \right|^2.
\]
We begin by noticing that for each $i = 1, \ldots, n$, the observed outcome can be written as
\[
Y_i = f_1(e(X_i)) Z_i + f_0(e(X_i))(1 - Z_i) + \epsilon_i,
\]
where $e(X_i)$ is the true propensity score, and the functions $f_0, f_1$ are defined as
\[
f_0(s) := \mathbb{E}(Y \mid Z = 0, e(X) = s), \quad
f_1(s) := \mathbb{E}(Y \mid Z = 1, e(X) = s).
\]
The error term $\epsilon_i = Y_i - \mathbb{E}(Y \mid Z_i, e(X_i))$ is mean-zero and sub-Gaussian by Assumption~\ref{as8}. In addition, we have
\[
\tau_i^* = f_1(e(X_i)) - f_0(e(X_i)),
\]
by the definition of $\tau_i^*$, $f_0$, and $f_1$.
Expanding the definition of $U_i$ and subtracting $\tau_i^*$ yields
\[
\begin{aligned}
U_i - \tau_i^*
&= Z_i \left[ -f_0(e(X_{N(i)})) + f_0(e(X_i)) \right] + (1 - Z_i)\left[ -f_1(e(X_i)) + f_1(e(X_{N(i)})) \right] \\
&\quad + (-1)^{Z_i + 1}(\epsilon_i - \epsilon_{N(i)}).
\end{aligned}
\]
Consequently, the quantity of interest becomes
\[
\begin{aligned}
\left| \frac{1}{\sqrt{n}} \sum_{i=1}^n (U_i - \tau_i^*) \right|^2 
&= \left| \frac{1}{\sqrt{n}} \sum_{i=1}^n \Big[ Z_i \left\{ -f_0(e(X_{N(i)})) + f_0(e(X_i)) \right\} \right. \\
&\quad + (1 - Z_i) \left\{ -f_1(e(X_i)) + f_1(e(X_{N(i)})) \right\} \\
&\quad \left. + (-1)^{Z_i + 1}(\epsilon_i - \epsilon_{N(i)}) \Big] \right|^2.
\end{aligned}
\]
To bound this expression, define
\[
\begin{aligned}
A &:= \frac{1}{\sqrt{n}} \sum_{i=1}^n \left[ Z_i \left\{ -f_0(e(X_{N(i)})) + f_0(e(X_i)) \right\} + (1 - Z_i) \left\{ -f_1(e(X_i)) + f_1(e(X_{N(i)})) \right\} \right], \\
B &:= \frac{1}{\sqrt{n}} \sum_{i=1}^n (-1)^{Z_i + 1} \epsilon_i, \quad 
C := \frac{1}{\sqrt{n}} \sum_{i=1}^n (-1)^{Z_i + 1} \epsilon_{N(i)}.
\end{aligned}
\]
Then, applying the inequality $(a + b + c)^2 \leq 2a^2 + 4b^2 + 4c^2$, we obtain
\[
\left| \frac{1}{\sqrt{n}} \sum_{i=1}^n (U_i - \tau_i^*) \right|^2 \leq 2|A|^2 + 4|B|^2 + 4|C|^2.
\]
This leads to the bound
\[
\begin{aligned}
\left| \frac{1}{\sqrt{n}} \sum_{i=1}^n (U_i - \tau_i^*) \right|^2 
&\leq 2 \left| \frac{1}{\sqrt{n}} \sum_{i=1}^{n} \left( Z_i[f_0(e(X_{N(i)})) - f_0(e(X_i))] + (1 - Z_i)[f_1(e(X_i)) - f_1(e(X_{N(i)}))] \right) \right|^2 \\
&\quad + 4 \left| \frac{1}{\sqrt{n}} \sum_{i=1}^n (-1)^{Z_i + 1} \epsilon_i \right|^2 + 4 \left| \frac{1}{\sqrt{n}} \sum_{i=1}^n (-1)^{Z_i + 1} \epsilon_{N(i)} \right|^2.
\end{aligned}
\]

Therefore,    
\begin{equation}
\label{eqn:small_part}
\begin{array}{lll}
\| P_{ {R}^{\perp} }(\hat{\tau} - \tau^*)  \|^2  & =&  \displaystyle \left\vert \frac{1}{n^{1/2} }  \sum_{i=1}^{n} (U_i  -  \tau_i^*)  \right\vert^2\\
& \leq & \displaystyle 2\left\vert \frac{1}{n^{1/2} }  \sum_{i=1}^{n} \left( Z_i[ f_0\{e(X_{i})\}  - f_0\{e(X_{N(i)}) \}]   +  (1-Z_i)[ -f_1\{e(X_{i})\}  + f_1\{e(X_{N(i)})\} ] \right)  \right\vert^2 \\
& & \displaystyle  + 4\left\vert \frac{1}{n^{1/2} }  \sum_{i=1}^{n} (-1)^{Z_i+1}\epsilon_i  \right\vert^2  +  4\left\vert \frac{1}{n^{1/2} }  \sum_{i=1}^{n} (-1)^{Z_i+1}\epsilon_{N(i)}  \right\vert^2\\
& \leq &\displaystyle  \frac{4}{n} \left[    \sum_{i=1}^{n} \vert f_0\{e(X_{i})\}  - f_0\{e(X_{N(i)})\}  \vert   \right]^2  +   \frac{4}{n} \left[    \sum_{i=1}^{n} \vert f_1\{e(X_{i})\}  - f_1\{e(X_{N(i)})\}  \vert   \right]^2 \\
& & \displaystyle  + 4\left\vert \frac{1}{n^{1/2} }  \sum_{i=1}^{n} (-1)^{Z_i+1}\epsilon_i  \right\vert^2  +  4\left\vert \frac{1}{\sqrt{n}}  \sum_{i=1}^{n} (-1)^{Z_i+1}\epsilon_{N(i)}  \right\vert^2.\\
\end{array}
\end{equation}
However, by the sub-Gaussian  tail inequality,
\begin{equation}
\label{eqn:gaussian_tail}
\left\vert \frac{1}{n^{1/2}}  \sum_{i=1}^{n} (-1)^{Z_i+1}\epsilon_i  \right\vert^2   \,=\,  o_{\mathbb{P}}\left(\log n\right).
\end{equation}
Therefore,  combining (\ref{eqn:small_part}), (\ref{eqn:gaussian_tail}), (\ref{eqn:tv2}) and Lemma \ref{lem9},
\begin{equation}
\label{eqn:new}
\| P_{ {R}^{\perp} }(\hat{\tau} - \tau^*)  \|^2  \,=\, O_{\mathbb{P}} \left( \log^3 n  +  \frac{t^2}{n}  \right).
\end{equation}
Next,  we proceed to bound $	\| P_{R }(\hat{\tau} - \tau^*)  \|^2$. Notice that by the optimality of  $\tilde{\tau}$
\begin{equation}
\label{eqn:basic}
\begin{array}{lll}
\displaystyle  \frac{1}{2}  \| P_{R}(\tau^* - \tilde{\tau})   \|^2 \, \leq \,   \displaystyle  \tilde{\epsilon}^{\top}  P_R( \tilde{\tau} -  \tau^* )  +  \lambda\left\{  \|\Delta^{(1)}\hat{P}\tau^*\|_1  -  \|\Delta^{(1)} \hat{P} \tilde{\tau} \|_1       \right\},
\end{array}
\end{equation}
where
\[
\begin{array}{lll}
\tilde{\epsilon}_i   &:= &   (-1)^{Z_i+1}( Y_i - Y_{N(i)})  - \tau_i^*\\
&  = &  Z_i [ f_0\{e(X_i)\}  -  f_0\{  e(X_{N(i)})   \}  ]  +  (1- Z_i)[ -f_1(e(X_i))  +  f_1\{e(X_{N(i)})    \}  ] +     (-1)^{Z_i+1} ( \epsilon_i  - \epsilon_{N(i)}),
\end{array}
\]
for  $i=1,\ldots,n$.   

Then  by H\"{o}lder's  inequality, Lemma  \ref{lem1}, and the inequality $\| P_R v\|_{\infty} \leq 2\|v\|_{\infty}$, there exists a constant $\tilde{C}>0$   such that, with probability   approaching one,
\begin{equation}
\label{eqn:basic_2}
\begin{array}{lll}
\tilde{\epsilon}^{\top}  P_R( \tilde{\tau} -  \tau^* )  & \leq  & \displaystyle \sum_{i=1}^{n} (-1)^{Z_i+1} ( \epsilon_i  - \epsilon_{N(i)})\{P_R(\tilde{\tau} -\tau^* )   \}_i +  \\
& &  \displaystyle  + 2 \| \tilde{\tau} - \tau^* \|_{\infty}\sum_{i=1}^{n} \vert  f_1\{e(X_i)\}  - f_1\{  e(X_{N(i)}) \} \vert  + 2\| \tilde{\tau} - \tau^* \|_{\infty}\sum_{i=1}^{n} \vert  f_0\{e(X_i)\}  - f_0\{  e(X_{N(i)} )\} \vert\\
& \leq&  \displaystyle \sum_{i=1}^{n} (-1)^{Z_i+1} ( \epsilon_i  - \epsilon_{N(i)})\{P_R(\tilde{\tau} -\tau^* )   \}_i +   A,\\
\end{array}
\end{equation}
where  
\[
A :=     \tilde{C}\,\left( \max\{ \|  f_0\|_{\infty},\|  f_1\|_{\infty} \}  +   \lambda  +  \log^{1/2} n \right)t,
\]
for some positive constant $\tilde{C}$.

Next, suppose  that   $\|  P_R(\tau^* - \hat{\tau})  \|^2/4   \leq   A  $. Then, due to our choice  of $\lambda$, (\ref{eqn:basic}), and (\ref{eqn:basic_2}), there exists  $\tilde{C}_2 >0$ such that
\begin{equation}
\label{eqn:first_bound}
\begin{array}{lll}
\displaystyle	  \frac{1}{n}\| P_R( \tau^* - \hat{\tau})  \|^2  &\leq&  \displaystyle \frac{ 4\tilde{C}\,\left( \max\{ \|  f_0\|_{\infty},\|  f_1\|_{\infty} \}  +   \lambda  + \log^{1/2} n \right)t}{n}\\
&\leq&  \displaystyle  \frac{  4 \tilde{C}\,\max\{ \|  f_0\|_{\infty},\|  f_1\|_{\infty} \}t}{n }  + \\
& &\displaystyle  \frac{4\,\tilde{C}_2\,\tilde{C}  (\log \log n)   (\log n)^{2} \,t^{2/3}}{n^{2/3}   }\,+\,  \\
&&	\displaystyle    4 \tilde{C}\,\frac{t \log^{1/2} n }{n},
\end{array}
\end{equation}
with probability approaching one.

On the contrary,  if  $\| P_R( \tau^* - \hat{\tau})  \|^2/4   >  A$, then    (\ref{eqn:basic}) and (\ref{eqn:basic_2})  imply
\begin{equation}
\label{eqn:basic3}
\begin{array}{lll}
\displaystyle  \frac{1}{4} \| P_{R}(\tau^* - \tilde{\tau})   \|^2 &\leq&   \displaystyle \sum_{i=1}^{n} (-1)^{Z_i+1} ( \epsilon_i  - \epsilon_{N(i)})\{P_R(\tilde{\tau} -\tau^* )   \}_i   + \\
& & \lambda\left\{  \|\Delta^{(1)}\hat{P}\tau^*\|_1  -  \|\Delta^{(1)} \hat{P} \tilde{\tau} \|_1       \right\},
\end{array}
\end{equation}
Next, we proceed to bound the  first term in the right hand size of the previous inequality. Towards that end, with the notation from the proof of Lemma \ref{lem9}, we exploit the argument in the proof of Lemma 9 from \cite{wang2016trend}. First,  by  Lemma  9, Theorem  10, and Corollary  12 from \cite{wang2016trend}, we have  that
\begin{equation}
\label{eqn:empirical_process_p1}
\begin{array}{lll}
\displaystyle     \underset{u \in \mathrm{row}(\Delta^{(1)}) \,:\,   \|  \Delta^{(1)} \hat{P} u\|_1 \leq 1,    }{\sup }\, \frac{ \sum_{i=1}^{n }   (-1)^{ Z_i +1 }  \epsilon_i\, u_i   }{ \|u \|^{1/2} }  \,=\, O_{\mathbb{P}} \left\{  n^{1/4} (\log \log  n)^{1/2} \right\},
\end{array}
\end{equation}
which follows due to the independence and sub-Gaussian assumption  of the errors  $\{\epsilon_i\}_{i=1}^n$. 

On the other hand, conditioning on $X$ and $\hat{e}$, we define the  random vectors $\epsilon^{(1)},\ldots, \epsilon^{(M)  }   \in \mathbb{R}^n$,  constructed as follows. First, let
\[
M \,=\,  \underset{  1 \leq i \leq n }{\max} \,\left\vert \left\{   j   \in \{1,\ldots,n\}\,:\,  N(j) =i \right\} \right\vert.
\]
Then
\[
\epsilon^{(1)}_i   \,=\,\ \begin{cases}
(-1)^{ Z_i +1 }  \epsilon_{N(1)}    &  \text{if}\,\,\,  i =1,\\
(-1)^{ Z_i +1 }  \epsilon_{N(i)} &  \text{if}\,\,\,   N(j) \neq N(i) \,\,\,\mathrm{for}\,\,\mathrm{ all}\,\,\,    j < i\\
0   & \text{otherwise,}
\end{cases}
\]
and
\[
S^{(1)}  \,=\, \left\{ i \,:\,   N(j) \neq  N(i)  \,\,\, \,\forall  j <i \,\,\,\ \right\} \cup \{1\}.
\]
And for  $l > 1$, we iteratively construct
\[
\epsilon^{(l)}_i   \,=\,\ \begin{cases}
(-1)^{ Z_i +1 } \epsilon_{N(i)} &  \text{if}\,\,\,   i   \notin   \cup_{m=1}^{l-1}  S^{(m)},   \,\,\,\mathrm{and} \,\,\,  N(j) \neq N(i) \,\,\,\mathrm{for}\,\,\mathrm{ all}\,\,\,    j < i\,\,\, \text{and}\,\,\,\,  j   \notin   \cup_{m=1}^{l-1}  S^{(m)}\\
0   & \text{otherwise.}
\end{cases}
\]
Notice that, by construction, the components  of  each $\epsilon^{(l)}$  are independent  and subGaussian($v$).
Hence, by triangle inequality,
\begin{equation*}
\begin{array}{lll}
\displaystyle     \underset{u \in \mathrm{row}(\Delta^{(1)}) \,:\,   \|  \Delta^{(1)} \hat{P} u\|_1 \leq 1,    }{\sup }\, \frac{ \sum_{i=1}^{n }   (-1)^{ Z_i +1 }  \epsilon_{ N(i)} \, u_i   }{ \|u \|^{1/2} } &\leq &    	\displaystyle 	\sum_{j=1}^M   \underset{u \in \mathrm{row}(\Delta^{(1)}) \,:\,   \|  \Delta^{(1)} \hat{P} u\|_1 \leq 1,    }{\sup }\, \frac{ u^{\top}\epsilon^{(j)  } }{ \|u \|^{1/2} }.     \\
\end{array}
\end{equation*}
Then  as in  the proof of Theorem 10 in \cite{wang2016trend}, which  exploits  Lemma  3.5 from   \cite{van1990estimating}, we have that  
\[
\mathbb{P}\left(    \underset{ u \in \mathrm{row}(\Delta^{(1)})\,:\,   \|   \Delta^{(1)} \hat{P}u \|_1  \leq 1 }{\sup} \,   \frac{   u^{\top}   \epsilon^{(j)} }{\|u\|^{1/2}  }   \geq   c_1 L   n^{1/4} (\log \log n)^{1/2}  \,\bigg|\, X, \hat{e}  \right) \,\leq \,    \exp\left(   -\frac{  c_0   L^2  \log \log n}{ v } \right)  ,
\]
for $j = 1,\ldots,M$  and for some  constant  $c_0, c_1 >0$, and for any constant  $L >L_0$, where $L_0$  is  a fixed  constant.

Therefore, by a union bound,
\begin{equation}
\label{eqn:empirical_p2}
\begin{array}{l}
\displaystyle 	\mathbb{P}\left(  \underset{u \in \mathrm{row}(\Delta^{(1)}) \,:\,   \|  \Delta^{(1)} \hat{P} u\|_1 \leq 1,    }{\sup }\, \frac{ \sum_{i=1}^{n }   (-1)^{ Z_i +1 }  \epsilon_{ N(i)} \, u_i   }{ \|u \|^{1/2} }  \geq  c_1 L   n^{1/4} (\log \log n)^{1/2} M   \,\bigg|\, X, \hat{e} \right)\\
\displaystyle 	 \,\leq \,  M  \exp\left(   -\frac{  c_0   L^2  \log  \log n}{ v } \right).
\end{array}
\end{equation}

Hence,  combining  (\ref{eqn:empirical_process_p1})   with (\ref{eqn:empirical_p2}),  integrating over  $X$ and  $\hat{e}$, and proceeding as in the proof of Lemma \ref{lem9},  we arrive at 
\begin{equation}
\label{eqn:empirical_p3}
\underset{  u \in \mathrm{row}(\Delta^{(1)})\,:\,   \|\Delta^{(1)}  \hat{P} u  \|_1  \leq 1}{\sup}\,  \frac{  \sum_{i=1}^{n} (-1)^{Z_i+1} ( \epsilon_i  - \epsilon_{N(i)})u_i   }{\|  u\|^{1/2}  }   \,=\,O_{\mathbb{P}}\left(K \right),
\end{equation}
where 
\[
K = n^{ 1/4 }  (\log \log n )^{1/2} \log n. 
\]
Next, we  notice that due to (\ref{eqn:basic3})   , (\ref{eqn:empirical_p3}), the proof of Lemma 9  in \cite{wang2016trend}, and since $\tilde{\tau} =P_R \hat{\tau}$,  we have that
\begin{equation}
\label{eqn:upperbound_1}
\begin{array}{lll}
\displaystyle   \|  P_R (\tau^* -   \hat{\tau} ) \|^2 &=&O_{\mathbb{P}}\left\{   \lambda \,t+  K^{4} \left( \frac{1}{\lambda} \right)^2   \right\},\\
\end{array}
\end{equation}
and
\[
\begin{array}{lll}
\lambda \,t  +  K^{4 } \left( \frac{1}{\lambda} \right)^2  &=& \displaystyle  O \big\{  n^{1/3} (\log n)^{2} ( \log \log n)\, t^{2/3}      \big\}.
\end{array}
\]
Hence, combining  (\ref{eqn:new}), (\ref{eqn:first_bound})  and  (\ref{eqn:upperbound_1}), we obtain that
\[
\begin{array}{lll}
\displaystyle \| \tau^* - \hat{\tau}\|^2  & = &\displaystyle  O_{\mathbb{P}} \Bigg\{   \log^3 n +  n^{1/3} (\log n)^{2} ( \log \log n)\, t^{2/3} + t\log^{1/2} n +  \frac{t^2}{n} \Bigg\},
\end{array}
\] 
and the claim follows.

\end{proof}

\subsection{Propensity score auxiliary lemmas}

\begin{lemma}
\label{lem2}
Under  Assumption \ref{as4},	for  all $\zeta >0$, 
\[
\mathbb{P}\left[  \left\vert \nabla \hat{L}(\theta^*)_j  \right\vert  \leq \sigma_0 \left\{ \frac{2 \log (3/\zeta ) }{m} \right\}^{1/2}  \right]  \geq   1 - \zeta,
\] 
for all $j= 1,\ldots,d$ and for a positive constant $\sigma_0$.
\end{lemma}

\begin{proof}
First, we observe that 
\begin{equation}
\label{eqn:gradient}
\displaystyle 	   \nabla \hat{L}(\theta^*) \,=\,  -\frac{1}{m}  \sum_{i=1}^{m} X_i^{\prime} \left\{ Z_i^{\prime}   -  F(  X_i^{ \prime  \top} \theta^{*}  )  \right\}.
\end{equation}
Furthermore, notice that
\[
\begin{array}{lll}
L(\theta)	 & =  & E\left( E\left[   Z \log F(X^{\top}\theta)  + (1- Z) \log \{1- F(X^{\top}\theta)\}    |X\right]\right)\\
& = & E\left[   e(X) \log F(X^{\top}\theta)  + (1- e(X)) \log \{1- F(X^{\top}\theta)\}   \right],
\end{array}
\]
and define 
\[
H(x,\theta) =  e(x) \log F(x^{\top}\theta)  + \{1- e(x)\} \log \{1- F(x^{\top}\theta)\}.
\]
Then
\[
\nabla_{\theta} H(x,\theta) \,=\,  e(x)   x -   F(x^{\top}\theta) x, 
\]
and $\vert(\nabla_{\theta} H(x,\theta) )_j \vert  \leq   2 \vert  x_j \vert  $ for all $j \in \{1,\ldots,d\}$. Hence, by the dominated convergence theorem,
\[
\begin{array}{lll}
\nabla L(\theta) &=&  E\left\{  \nabla_{\theta} H(X,\theta)   \right\}\\ 
&= & E\left( e(X)   X -   F(X^{\top}\theta) X \right)\\
& = &  E\left[E\left\{ Z  X -   F(X^{\top}\theta) X |X\right\}\right]\\
& = & E\left\{ Z   X -   F(X^{\top}\theta) X \right\}.\\
\end{array}   
\]
Therefore, 
\begin{equation}
\label{eqn:zero_gradient}
E\left\{ Z   X -   F(X^{\top}\theta^*) X \right\} = 0.
\end{equation}
Also,
\[
\left\vert  X_{i,j}^{\prime } \left\{  Z_i^{\prime }   -  F(  X_i^{\prime \top} \theta^{*}  )  \right\} \right\vert  \leq   \vert  X_{i,j}^{\prime}\vert. 
\]
Hence, by Assumption \ref{as4}, (\ref{eqn:gradient}), (\ref{eqn:zero_gradient}) and Corollary 2.6 from \cite{boucheron2013concentration}, we obtain the conclusion.
\end{proof}

\begin{lemma}
\label{lem3}
Suppose  that Assumption \ref{as6}  holds and  $d \leq  4  \|X\|_{ \psi_2 }^2 m^{1/2}/\log^{1/2} m$.
Then 
\[
\displaystyle      \mathbb{P}\left\{ \Lambda_{\min}(Q^m) \leq  C_{\min}/2 \right\}\leq   \,2\exp\left(   -\frac{c d}{16}\log^2 m    +  d \log 9 \right),
\]
where  
\[
\displaystyle   Q^m   =   \frac{1}{m} \sum_{i=1}^m  \eta( X_i^{\prime \top}\theta^* ) X_i^{\prime} X_i^{\prime \top},
\]
and $c$ is a positive constant. 
Similarly,
\[
\displaystyle	\mathbb{P}\left( \Lambda_{\max}\left\{ \frac{1}{m} \sum_{i=1}^{m}  X_i^{\prime} X_i^{\prime \top}   \right\}\geq  3 D_{\max}/2 \right)\leq \,2\exp\left(   -\frac{c d}{16}\log^2 m    +  d \log 9 \right).
\]
\end{lemma}

\begin{proof}
Proceeding as in the proof of Lemma 5 in \cite{ravikumar2010high}, we obtain that 
\[
\Lambda_{\min}\left( Q^{m} \right) \geq  C_{\min}  -  \|Q -  Q^m  \|_2,
\]
where  $\| \cdot \|_2$  denotes the spectral norm and  with
\[
Q \,=\, E\left\{\eta( X^{\top}\theta^*) X X^{\top} \right\}.
\]

To bound the quantity  $\|Q -  Q^m \|_2 $   we let  $v \in   \mathbb{R}^d$ with $\|v \|=1$,  and notice that
\[
\displaystyle   v^T  \left(  Q^m - Q   \right)v = \frac{1}{m}\sum_{i=1}^{m}\left\{ \left[ \{\eta(X_i^{\prime \top} \theta^* )\}^{1/2}v^T X_i^{\prime } \right]^2 \,-\, E\left(\left[ \{\eta(X_i^{\prime\top} \theta^* )\}^{1/2}v^T X_i^{\prime} \right]^2  \right)\right\} ,
\]
and by Proposition 5.16 in \cite{vershynin2010introduction}
\[
\displaystyle   \mathbb{P}\left\{ \left\vert  v^T  \left(  Q^m - Q   \right)v  \right\vert  \geq r \right\}\,\leq \,2\exp\left(   -c \min\left\{  \frac{r^2  m }{16 \|X\|_{\psi_2}^4   },\frac{r m}{4   \|X\|_{\psi_2}^2 }  \right\}  \right),
\]
for all $r>0$, and for an absolute  constant  $c>0$. Hence, taking  $r =   \|X\|_{ \psi_2 }^2  d^{1/2} \log^{1/2} m/m^{1/2}$,
and with the same entropy based  argument from  the proof of Lemma 5 in \cite{wang2017optimal}, we arrive at 
\[
\displaystyle   \mathbb{P}\left\{  \left\| Q^m - Q    \right\|_2  \geq \|X\|_{ \psi_2 }^2 \left(\frac{ d \log m }{m}\right)^{1/2}  \right\}\,\leq \,2\exp\left(   -\frac{c d}{16}\log^2 m    +  d \log 9 \right).
\]
Finally, 
\[
\begin{array}{l}
\displaystyle	\mathbb{P}\left\{ \Lambda_{\max}\left( \frac{1}{m} \sum_{i=1}^{m}  X_i^{\prime} X_i^{ \prime  \top}   \right)\geq  3 D_{\max}/2 \right\}\\
\displaystyle \,\leq \,  \mathbb{P}\left\{  \left\| \frac{1}{m} \sum_{i=1}^{m} X_i^{\prime} X_i^{\prime \top}   -     E\left(\frac{1}{m} \sum_{i=1}^{m}  X_i^{\prime } X_i^{\prime \top} \right)  \right\|_2  \geq \|X\|_{ \psi_2 }^2 \left(\frac{ d \log m }{m}\right)^{1/2} \right\}

\end{array}
\]
and the proof concludes with the same argument  from above.

\end{proof}

\begin{lemma}
\label{lem4}
Suppose that Assumption \ref{as4}--\ref{as6} hold. 
Then for a positive  constant  $C_1$,
the  estimator  $\hat{\theta}$ defined in  (\ref{eqn:logistic_theta})  satisfies 
\[
\| \hat{\theta}  -  \theta^*  \|  \,\leq \,  \frac{C_1}{  C_{\min} } \left(  \frac{d  \log m}{m} \right)^{1/2}, 
\]
with probability approaching one.
\end{lemma}

\begin{proof}
For  $u \in   \mathbb{R}^d$  let
\[
G(u)  \,=\,   \hat{L}(\theta^* + u) - \hat{L}(\theta^*).
\]
Clearly,  $G(0) =  0$,   and  $G(\hat{u})  \leq  0$  where  $\hat{u} = \hat{\theta} - \theta^*$.
Let 
\begin{equation}
\label{eqn:b}
B  := \frac{8 \sigma_0 }{C_{\min}}   \left( \frac{12 d\log m}{m} \right)^{1/2}.
\end{equation}
We proceed to  show that   $G(u) >0$  for  all $\|u\|  = B$, which implies,  by convexity, that $\|\hat{u}\|\leq B$. Towards  that end, notice that,  by Taylor's theorem,  we have 
\[
G(u) =  \nabla \hat{L}(\theta^*)^{\top } u  +  u^{\top} \nabla^2 \hat{L}(\theta^*  +  \alpha   u ) u,
\]
for some $\alpha \in [0,1]$. Also,
\[
\vert \nabla \hat{L}(\theta^*)^{\top} u   \vert  \leq   \|\nabla \hat{L}(\theta^*)\|_{\infty} \|u\|_1 \leq   d^{1/2}\|\nabla \hat{L}(\theta^*)\|_{\infty} \|u\|.
\]
Hence,  by Lemma \ref{lem2} and a union bound,
\begin{equation}
\label{eqn:gradient_part}
G(u)   \geq -  d^{1/2} \|u\|\sigma_0   \left( \frac{12 \log  m}{m} \right)^{1/2}    +  u^T \nabla^2 \hat{L}(\theta^*  +  \alpha   u ) u,  
\end{equation}
with probability  at least  $1-  1/m$. 

Furthermore,
\[
\begin{array}{lll}
\displaystyle 	\nabla^2 L(\theta^* +  \alpha u)  
& =& \displaystyle   \frac{1}{m} \sum_{i=1}^m F^{\prime}\{ X_i^{\prime \top} (\theta^* +  \alpha u)  \}X_i^{\prime} X_i^{\prime \top}\\
& =: &  \mathcal{A}_1.
\end{array}
\]
Also,
\begin{equation}
\label{eqn:eigen1}
\begin{array}{lll}
q^*    &  := & \displaystyle   \Lambda_{\min}(  \mathcal{A}_1   )\\
& \geq  &\displaystyle \underset{  \alpha \in [0,1] }{\min} \,\Lambda_{\min}\left[  \frac{1}{m }   \sum_{i=1}^{m}  \eta\{X_i^{\prime \top}( \theta^* + \alpha u  )  \}  \, X_i^{\prime} X_i^{\prime \top}     \right ]\\
& \geq & \displaystyle \underset{  \alpha \in [0,1] }{\min} \,\Lambda_{\min}\left\{  \frac{1}{m  }   \sum_{i=1}^{m}  \eta( X_i^{\prime \top} \theta^*    )  \, X_i^{\prime} X_i^{\prime \top}     \right \}      \,-\,\\
& &  \displaystyle\underset{  \alpha \in [0,1] }{\max}  \left\|      \frac{1}{m}   \sum_{i=1}^{m}     \eta^{\prime}\{ X_i^{\prime \top}(\theta^* +  \alpha u )    \}(u^T X_i^{\prime }) X_i^{\prime} X_i^{\prime \top}  \right\|_2\\
& = & \displaystyle          \Lambda_{\min}(Q^m)  \, -\, \underset{  \alpha \in [0,1] }{\max}  \|A(\alpha)\|_2,\\
\end{array}
\end{equation}
where
\[
A(\alpha)   =       \frac{1}{m}   \sum_{i=1}^{m}     \eta^{\prime}\{  X_i^{\prime \top}(\theta^* +  \alpha u )    \} (u^T X_i^{\prime}) X_i^{\prime } X_i^{\prime \top} . 
\]

Now   for  $\alpha \in [0,1]$,  $v \in  \mathbb{R}^d$  with  $\|v\|=1$, we have 
\begin{equation}
\label{eigen2}
\begin{array}{lll}
\displaystyle    v^T  A(\alpha) v   &  = &    \displaystyle\frac{1}{m} \sum_{i=1}^{m} \eta^{\prime}\{ X_i^{\prime} ( \theta^* +  \alpha u )    \}\left( u^{\top} X_i^{\prime}   \right)\left( X_i^{\prime \top} v  \right)^2\\
& \leq & \displaystyle      \|u\|   \left\{ \underset{i = 1,\ldots ,m }{\max} \,  \left\vert (u/ \| u\|)^{\top} X_i^{\prime} \right\vert  \right\}  \|\eta^{\prime}\|_{\infty}  \frac{1}{m} \sum_{i=1}^{m }  \left( X_i^{\prime \top} v  \right)^2 \\
& \leq&  \displaystyle  \,     \|u\|   \left\{ \underset{i = 1,\ldots ,m }{\max} \,  \left\vert (u/ \| u\|)^{\top} X_i^{\prime } \right\vert  \right\} \left\| \frac{1}{m}  \sum_{i=1}^{m}   X_i^{\prime}  X_i^{\prime \top} \right\|_2\\
& \leq &\displaystyle c_2  d^{1/2} \log^{1/2} m \|u\|\, D_{\max}.
\end{array}
\end{equation}
with probability approaching  one, and where $c_2 >0$ is a constant. Here, we have used the fact the random variables  $\{ (u/ \| u\|)^{\top} X_i^{\prime}  \}_{i=1}^m$ are sub-Gaussian, and the second claim in Lemma \ref{lem3}. 
Therefore, with  probability approaching  one,
\[  
\begin{array}{lll}
\displaystyle	G(u)   &\geq &	\displaystyle  -  d^{1/2} \|u\|\sigma_0   \left( \frac{12 \log  m}{m} \right)^{1/2}     +  \frac{C_{\min}  \|u\|^2 }{2} -  c_2  (d \log m)^{1/2} \|u\|^3\, D_{\max}\\
& \geq&  	\displaystyle  - d^{1/2} \|u\|\sigma_0   \left( \frac{12 \log  m}{m} \right)^{1/2}     +  \frac{C_{\min}  \|u\|^2 }{4}\\
& >&0,
\end{array}
\]
where  the  first  inequality  follows from (\ref{eqn:gradient_part})-- (\ref{eigen2}), the second from  (\ref{eqn:b}) and (\ref{eqn:low_b}), and the  third  since
\[
\|u\| =  B   >  \frac{4 \sigma_0 }{C_{\min}}   \left(  \frac{12 d \log m}{m} \right)^{1/2}. 
\]
\end{proof}

\begin{lemma}
\label{thm:parametric}
Under Assumptions  \ref{as1} and \ref{as4}--\ref{as6},  there  exists a constant  $C_0>0$   such that 
\[
\underset{i =1,\ldots,n }{\sup}\,\vert  \hat{e}(X_i) -  e(X_i) \vert  \leq  C_0 \frac{d}{ C_{\min} }\left[  \frac{  (\log m)\{\log (nd)\} }{m} \right]^{1/2},    
\]
with probability approaching one, provided that $d = O(m)$.
\end{lemma}

\begin{proof}
We have  that 
\begin{equation}
\label{eqn:calculation}
\begin{array}{lll}
\underset{ i = 1,\ldots, n}{\max}\,\left\vert  \hat{e}(X_i) -   e(X_i)\right\vert & = &  \underset{ i = 1,\ldots, n}{\max}\,\left\vert  F( X_i^{\top}\hat{\theta} ) -  F(X_i^{\top}  \theta^* )\right\vert \\ 
& \leq &  \underset{ i = 1,\ldots, n}{\max}\,  \| F^{\prime} \|_{\infty}  \left\vert     X_i^{\top}\hat{\theta}  -X_i^{\top}  \theta^*\right\vert\\
& \leq &\|     \hat{\theta}  -   \theta^*\|\, \underset{ i = 1,\ldots, n,\,\,}{\max}\,  \|X_{i}\|  \\
&\leq&d^{1/2}\, \|     \hat{\theta}  -   \theta^*\|\,   \underset{ i = 1,\ldots, n,\,\, j = 1,\ldots, d }{\max}\,  \vert X_{i,j}\vert  \\
&\leq&d^{1/2}\, \|     \hat{\theta}  -   \theta^*\|\,   \underset{ i = 1,\ldots, n,\,\, j = 1,\ldots, d }{\max}\,  \vert X_{i,j}  - (\mathbb{E}(X))_j \vert  \\
& &  +  d^{1/2}\, \|     \hat{\theta}  -   \theta^*\|\,   \|E(X)\|_{\infty},\\
\end{array}
\end{equation}
and the claim follows  by  Lemma \ref{lem4}.

\end{proof}

\subsection{Lemma combining both stages}

\begin{lemma}
\label{lem5}
There exists  a positive  constant $C_1$ such that the event
\[
\underset{i = 1,\ldots, n}{\max }\,  \vert  i -  \sigma^{-1} \{\hat{\sigma}(i)\}  \vert   \leq  C_1 \max\{  \log n,  n \delta \},     
\]	
holds with probability  approaching one.

\end{lemma}

\begin{proof}
First, by  Lemma \ref{thm:parametric}, we will assume that  the event 
\begin{equation}
\label{eqn:uniform_p}  
\underset{i =1,\ldots,n }{\sup}\,\vert  \hat{e}(X_i) -   e(X_i) \vert  \leq  C_0 \delta,
\end{equation}
holds.

Then for  each $i \in \{1,\ldots,n\}$ define  
\[
m_i  = \left\vert  \left\{ k  \in   \{1,\ldots,n \} \backslash \{i\}  \,:\, e(X_k) \in \left( e(X_i) -2C_0 \tilde{\delta}   , e(X_i) + 2C_0 \tilde{\delta} \right)    \right\}   \right\vert.
\]
where  $\tilde{\delta} =  \max\{ \delta\,,\, C_1 \,n^{-1} \log n \}$ for a  positive  constant  $C_1$  to be chosen later.

Then by Assumption \ref{as5},
\[
m_i \,\sim\, \text{Binomial}\left(n-1,  \int_{ 0 }^{ 1 }\int_{ \max\{ 0, t -2 C_0\tilde{\delta}    \} }^{ \min\{1,  t+ 2  C_0\tilde{\delta}   \} }     h(s)\,h(t) \, ds \,dt   \right).
\]
Hence,
\[
E(m_i) \leq  4  C_0 h_{\max}^2 n \tilde{\delta}.
\]
Therefore,  by a union bound, and Chernoff's inequality,
\begin{equation}
\label{eqn:binom_concent.}
\begin{array}{lll}
\mathbb{P}\left(  \underset{ i = 1, \ldots,n }{\max}\,  m_i \,\geq \,  12 C_0\,h_{\max}^2 n \tilde{\delta}	 \right)  & \leq &  \exp\left(- C_1\log n/4     +\log n \right)  \,\rightarrow   0,  \,\,\,\,\,\,\,\text{as}  \,\,  n\rightarrow \infty,
\end{array}
\end{equation}
provide that  $C_1$  is  chosen large  enough.

On the other hand,  by (\ref{eqn:uniform_p}), we have 
\begin{itemize}
\item If $e(X_{ i^{\prime} })   <   e(X_i) -  2 C_0 \tilde{\delta}$, then   $\hat{e}(X_{ i^{\prime}  })  <  \hat{e}(X_i)$.
\item   If $e(X_{ i^{\prime} })   >   e(X_i) +  2 C_0 \tilde{\delta}$, then   $\hat{e}(X_{ i^{\prime}  })  >  \hat{e}(X_i)$.
\end{itemize}

Therefore,	
\[
\underset{i = 1,\ldots, n}{\max }\,  \vert  i -  \sigma^{-1} \{\hat{\sigma}(i)\}  \vert \,\leq \,  \underset{ i = 1, \ldots,n }{\max}\,  m_i,
\]
and the claim  follows combining (\ref{eqn:uniform_p}) and (\ref{eqn:binom_concent.}).
\end{proof}

\begin{lemma}
\label{lem7}
With the notation from before,
\[
\displaystyle	  \sum_{i=1}^{n-1} \vert f_1\{ e(X_{\hat{\sigma}(i)}  )\}   - f_1\{ e(X_{\hat{\sigma}(i+1)}  ) \}    \vert  \,=\, O_{\mathbb{P}}\left(   \max\{\log n, n\delta  \}\, \right),
\]
and
\[
\displaystyle		  \sum_{i=1}^{n-1} \vert f_0\{e(X_{\hat{\sigma}(i)}  )  \}   - f_0\{ e(X_{\hat{\sigma}(i+1)}  )  \}    \vert  \,=\,  O_{\mathbb{P}}\left(  \max\{\log n, n\delta \}\,  \right).
\]
\end{lemma}

\begin{proof}
Suppose that  the event 
\begin{equation}
\label{eqn:event}
\underset{i = 1,\ldots, n}{\max }\,  \vert  i - \sigma^{-1} \{\hat{\sigma}(i)\}  \vert  \leq  C_1  \max\{ \log n ,  n\delta   \},
\end{equation}

holds  for some positive  constant  $C_1$, see Lemma \ref{lem5}.

Then
\[
\begin{array}{lll}
\displaystyle	  \sum_{i=1}^{n-1} \vert f_1\{ e(X_{\hat{\sigma}(i)}  )\}   - f_1\{ e(X_{\hat{\sigma}(i+1)}  ) \}    \vert     & \leq & \displaystyle	\sum_{i=1}^{n-1} \,\sum_{j = \max\{1, i - C_1  \max\{ \log n ,  n\delta   \}   \}   }^{   \min\{n, i +  C_1  \max\{ \log n ,  n\delta   \}\}   }     \vert f_1\{e(X_{\sigma(j)})  \}  - f_1\{ e(X_{\sigma(j+1)}  ) \}     \vert \\
& \leq& \displaystyle 2 C_1   \max\{ \log n ,  n\delta   \}   \sum_{i=1}^{n-1} \vert f_1\{ e(X_{\sigma(i)}  )\}   - f_1\{ e(X_{\sigma(i+1)}  ) \}     \vert\\
\end{array}
\]
and the claim follows.

\end{proof}

\begin{lemma}
\label{lem6}
There exists  a positive constant $c_1$  such that the event
\begin{equation}
\label{eqn:rate}
\max\left\{   \underset{ i = 1,\ldots, n  }{\max} \,\underset{  Z_j \neq  Z_i  }{\min}\,  \left\vert     e(X_i) - e(X_j)  \right\vert,\,\, \underset{ i = 1,\ldots, n  }{\max} \,\underset{  Z_j =  Z_i  }{\min}\,  \left\vert     e(X_i) - e(X_j)  \right\vert \right\} \,\leq\, \frac{ c_1  \log n }{n},
\end{equation}
happens with probability  approaching one.

\end{lemma}

\begin{proof}
Let  $c_1>0$  be a constant to be chosen later. 
Fix  $i \in \{1,\ldots,n\}$ and  define  
$$\Lambda_i =  \left\{   j \geq  1 \,:\,   Z_{i+k} \neq  Z_i  \,\,\text{for}\,\,\,k \in \{1,\ldots,j\}   \right\}, $$
and    $m_i  =  \max \{a  \,:\,  a\in  \Lambda_i  \}$.

Then we have that 
\[
\begin{array}{lll}
\displaystyle    \mathbb{P}\left(m_i \geq   c_1 \log n   \right) &= &\displaystyle    \mathbb{P}\left(m_i \geq   c_1 \log n  |Z_i =  1  \right) \mathbb{P}(Z_i=1)  + \mathbb{P}\left(m_i \geq   c_1 \log n  |Z_i =  1  \right) \mathbb{P}(Z_i=0)    \\
& \leq & \displaystyle    \mathbb{P}\left(m_i \geq   c_1 \log n  |Z_i =  1  \right) e_{\max}  +  \mathbb{P}\left(m_i \geq   c_1 \log n  |Z_i =  0  \right) (1-e_{\min})    \\
& \leq &\displaystyle     (1-e_{\min})^{c_1 \log n  }e_{\max}  + e_{\max}^{c_1 \log n  } (1-e_{\min})\\
& \leq&   \max\{e_{\max},1-e_{\min}  \}^{  c_1 \log n   }\\
& \leq& \exp\{ c_1 (\max\{\log(e_{\max}),\log(1-e_{\min})  \}) \log n    \}.
\end{array}
\]
Hence,  by union bound, 
\[
\begin{array}{lll}
\displaystyle    \mathbb{P}\left(  \underset{i=1,\ldots,n}{\max}m_i \geq   c_1 \log n   \right) 
& \leq& \exp\{ c_1 (\max\{\log(e_{\max}),\log(1-e_{\min})  \}) \log n     +  \log n\},
\end{array}
\]
and so we set  $c_1 =  -2/\max\{\log(e_{\max}),\log(1-e_{\min})  \}$.  The claim follows  by Assumption \ref{as5} and the argument in the proof of Proposition 30 in \cite{von2014hitting}  implying that, with high probability,  the distance of   of each $ e(X_i)$ to its  $r$  nearest  neighbor is of order    $O(\log n/n)$  for  $r \asymp \log n$.

\end{proof}

\begin{lemma}
\label{lem8}
For any $i \in \{1,\ldots, n\}$ let 
\[
\xi_i\,=\, \left\vert  \{  j  \in \{1,\ldots,n\} \,:\,\,\,\,\,\text{and }\,\,\,\,\,  e(X_i) \leq  e(X_j) \leq  e(X_{N(i)}) \,\,\,\text{or}\,\,\,  e(X_{N(i)}) \leq e(X_j) \leq  e(X_{i})    \} \right\vert.
\]
Then, for some constant $C_2 >0$,
\[
\underset{i =1 ,\ldots,n}{\max}\,\xi_i  \,\leq \, C_2 \max\{  \log n,  n \delta \},
\]
with probability approaching one.
\end{lemma}

\begin{proof}
First, assume that the event  (\ref{eqn:uniform_p})
holds. Next, by Lemma \ref{lem6}, with high probability,  for all $i \in \{1,\ldots,n\}$  there exits  $j(i) \in \{1,\ldots,n\}$  such  that 
$Z_{j(i)} \neq Z_i$ and $\vert  e(X_i)    - e(X_{ j(i) }) \vert \leq  (c_1 \log n)/n $.  Under such event, 
\[
\begin{array}{lll}
\vert  e(X_i) - e(X_{N(i)} ) \vert   & \leq &  \vert  e(X_i) - \hat{e}(X_i) \vert +   \vert  \hat{e}(X_i) - \hat{e}(X_{N(i)}) \vert +   \vert \hat{e}(X_{N(i)}) -  e(X_{N(i)}) \vert
\\
& \leq&  \vert  e(X_i) - \hat{e}(X_i) \vert +   \vert  \hat{e}(X_i) - \hat{e}(X_{j(i)}) \vert +   \vert \hat{e}(X_{N(i)}) -  e(X_{N(i)})\vert\\
& \leq& 2C_0 \delta  + \vert  \hat{e}(X_i) - \hat{e}(X_{j(i)}) \vert\\
&\leq & 4C_0 \delta+\vert  e(X_i) - e(X_{j(i)}) \vert\\
&\leq & 4C_0 \delta +   \frac{c_1 \log n}{n},\\
\end{array}
\]
where the second inequality follows from the definition of $N(i)$,  the third and fourth  from 
Lemma \ref{thm:parametric}, and the last from the construction of $j(i)$. Therefore, 
\[
\xi_i \,\leq \, \left\vert \left\{   j\,:\,   \vert   e(X_i) - e(X_j) \vert \leq    4C_0 \delta   + \frac{c_1 \log n}{n} \right \}\right \vert,
\]
and the claim  follows  in the same way that we bounded the counts  $\{m_i\}$ in the proof of Lemma \ref{lem5}.
\end{proof}

\begin{lemma}
\label{lem10}
With the notation from before, 
\[
\sum_{i=1}^{n} \vert   f_l\{  e(X_{i}  ) \}  -  f_l\{e(X_{N(i)} )\} \vert   =  O_{\mathbb{P}}\left(      \max\{\log n, n\delta   \}\,\right),
\]
for $l \in \{0,1\}$.
\end{lemma}

\begin{proof}
By Lemma \ref{lem8} and the triangle inequality, we have that, with probability  approaching one,
\[
\begin{array}{lll}
\displaystyle		  	\sum_{i=1}^{n} \vert   f_l\{ e(X_{i}  ) \}  -  f_l\{ e(X_{N(i)} )\} \vert    & \leq & \displaystyle	\sum_{i=1}^{n-1} \,\sum_{j = \max\{1, i -  C_2\max\{\log n ,  n \delta  \} \}  }^{   \min\{n, i + C_2\max\{\log n ,  n \delta  \}   \}  }     \vert f_l\{ e(X_{ \sigma(j) }  ) \}   - f_l\{ e(X_{ \sigma(j+1) }  ) \}   \vert \\
& \leq & \displaystyle    \left[  C_2\max\{\log n ,  n \delta  \}  \right] \left[ \sum_{j=1}^{n-1} \vert f_l\{ e(X_{ \sigma(j) }  ) \}   - f_l\{ e(X_{ \sigma(j+1) }  ) \}  \vert   \right]\\
\end{array}
\]
and the claim follows.

\end{proof}

\subsection{Putting the pieces together }

The claim in Theorem \ref{thm:tv} follows immediately from Lemmas \ref{lem7}, \ref{lem10} and   \ref{thm1}.

\section{Proof of   Theorem \ref{thm:prognostic_upper_bound}}

\subsection{Notation}
Throughout this section we define the function  $\hat{L}  \,:\,   \mathbb{R}^d  \rightarrow  \mathbb{R}$ as
\[
\hat{L}(\theta)   \,=\, \frac{1}{m}\sum_{i=1}^{m}\left(Y_i^{\prime }  -X_i^{\prime \top}  \theta \right)^2,\,\,\,\,\,\,\,\,\,\,\theta \in  \mathbb{R}^d,
\]
and set 
\[
\delta  =   \frac{d}{C_{\min}} \left\{  \frac{  \left(  \log^{1/2} m    +   d^{1/2}\|\theta^*\|\right) \log m  }{m}\right\}^{1/2}.
\]
Furthermore,  we  consider  the  orderings  $\sigma$,  $\tilde{\sigma}$, and $\hat{\sigma}$  satisfying
\begin{equation}
\label{eqn:sigma}
g\{X_{\sigma(1)}\}   <\,\ldots\,< g\{X_{\sigma(n)}\},
\end{equation}
\begin{equation}
\label{eqn:sigma_tilde}
\tilde{g}\{X_{\tilde{\sigma}(1)}\}   <\,\ldots\,< \tilde{g}\{X_{\tilde{\sigma}(n)}\},
\end{equation}
and 
\begin{equation}
\label{eqn:sigma_hat}
\hat{g}\{X_{\hat{\sigma}(1)}\}   <\,\ldots\,< \hat{g}\{X_{\hat{\sigma}(n)}\}.
\end{equation}

\subsection{Auxiliary  lemmas}

\begin{lemma}
\label{lem11}
Under Assumptions \ref{as7}-\ref{as10}, we have that for  some  $C_1>0$,
\[
\underset{j =1 ,\ldots,d}{\max}\,\left\vert \left\{\nabla \hat{L}(\theta^*)\right\}_j\right\vert   \,\leq \,    C_1\left\{ \frac{  (\log^{1/2} m + d^{1/2} \|\theta^*\|  )\log m  }{m} \right\}^{1/2},
\]
with probability approaching one.
\end{lemma}

\begin{proof}
First notice  that   by the  optmiality  of  $\theta^*$  we have that
\[
\nabla L(\theta^*)   \,=\,   E\left\{  X( X^{\top} \theta^*-Y  )|Z=0 \right\}.
\]
Hence,
\[
0    \,=\,  E\left\{   \nabla \hat{L}(\theta^*) \right\}\,=\, E\left\{ \frac{1}{m}\sum_{i=1}^{m}   X_i^{\prime}( X_i^{\prime \top} \theta^*- Y_i^{\prime }  )     \right\}.
\]
Furthermore,  defining  $\tilde{ \epsilon}_i =  Y_i^{\prime} - f_0(X_i^{\prime} ) $, we have 
\begin{equation}
\label{eqn:t1}
\begin{array}{lll}
\vert  X_{i,j}^{\prime}( X_i^{\prime \top} \theta^*-Y_i^{\prime }  )  \vert  & \leq  &  \vert  X_{i,j}^{\prime}\vert  \left(  \|f_0\|_{\infty}  + \|\tilde{ \epsilon}\|_{\infty}   +    \|X_i^{\prime}\|_{\infty}  \|\theta^*\|_{1}\right)\\
& \leq&  \|X_i^{\prime}\|_{\infty}   \left(   \|f_0\|_{\infty}  + \|\tilde{ \epsilon}\|_{\infty}   +    \|X_i^{\prime}\|_{\infty} d^{1/2} \|\theta^*\|\right)\\
\end{array}
\end{equation}
where  the first inequality follows  by H\"{o}lder's inequality, and the  second by the relation between  $\ell_1$ and  $\ell_2$ norms.

However,    by  Assumption \ref{as8},  it follows that 
\[
\Omega =\{ \|\tilde{ \epsilon} \|_{\infty}   \leq   3\sigma \log^{1/2} m\},
\]
holds   with probability at least $1- \frac{1}{m^2}$. Hence, 
by Assumption \ref{as10}, (\ref{eqn:t1}), and Hoeffding's  inequality
\[
\begin{array}{l}
\left \vert  \frac{1}{m}\sum_{i=1}^{m}  X_{i,j}^{\prime}(X_i^{\prime \top} \theta^*-Y_i^{\prime }  )  \right \vert \\
\displaystyle      		 \,\leq \,  \left\{ \frac{   4\max\{ \|a\|_{\infty},\|b\|_{\infty} \}  \left(   \|f_0\|_{\infty}  +      3\sigma \log^{1/2} m    +   \max\{ \|a\|_{\infty},\|b\|_{\infty} \} d^{1/2} \|\theta^*\|\right)\log\left(   m^2  \right) }{2m}  \right\}^{1/2} ,
\end{array}
\]
with proability $1- \frac{4}{m^2}$.

\end{proof}

\begin{lemma}
\label{lem13}

For any $i \in \{1,\ldots, n\}$ let 
\[
\xi_i\,=\, \left\vert  \{  j  \in \{1,\ldots,n\} \,:\,\,\,\,\,\text{and }\,\,\,\,\,  \tilde{g}(X_i) \leq \tilde{g}(X_j) \leq \tilde{g}(X_{N(i)}) \,\,\,\text{or}\,\,\,  \tilde{g}(X_{N(i)}) \leq \tilde{g}(X_j) \leq \tilde{g}(X_{i})    \} \right\vert.
\]
Then, for some constant $C_2 >0$,
\[
\underset{i =1 ,\ldots,n}{\max}\,\xi_i  \,\leq \, C_2 \max\{  \log n,  n \delta  \},
\]
with probability approaching one.
\end{lemma}

\begin{proof}
The claim follows as  the proof of Lemma \ref{lem8}, exploiting  Lemma \ref{lem12}.
\end{proof}

\begin{lemma}
\label{lem12}
Under Assumptions  \ref{as7}--\ref{as10},   we have  that for some constant $\tilde{C}>0$,
\[ 
\begin{array}{lll}
\displaystyle  \underset{ i = 1,\ldots, n}{\max}\,\left\vert  \hat{g}(X_i) -  \tilde{g}(X_i)\right\vert & \leq  & \displaystyle    \frac{  \tilde{C}  d      }{C_{\min}} \left\{   \frac{  \left(   \log^{1/2} m   +    d^{1/2} \|\theta^*\|\right) \log m  }{m}\right\}^{1/2},
\end{array}
\]
with probability approaching one.
\end{lemma}

\begin{proof}
We  define the  function \[
G(u) =  \hat{L}(\theta^* + u)    - \hat{L}(\theta^*),
\]
and observe that  $G(0) =0$, and  $G(\hat{u}) <  0$ where  $\hat{u} =\hat{\theta} -\theta^*$.  Let 
\[
B =  \frac{c_1}{C_{\min}} \left\{   \frac{  d\left(   \log^{1/2} m    +   d^{1/2} \|\theta^*\|\right) \log m  }{m} \right\}^{1/2},
\]
for  some  $c_1>0$, and take  $u \in   \mathbb{R}^d$ such that  $\|u\|=B$. Then,   with probability  approaching one,   by Lemma \ref{lem11}  and  the proof  of Lemma \ref{lem3}, we have that  
\[
\begin{array}{lll}
G(u)  & = & \displaystyle  \frac{1}{m} \sum_{i=1}^{m} \left\{ X_i^{\prime \top}( \theta^*+u ) -Y_i^{\prime }  \right\}^2   \,- \, \frac{1}{m} \sum_{i=1}^{m} \left( X_i^{\prime \top} \theta^*-Y_i^{\prime }  \right)^2\\
& = &\displaystyle  u^{\top }\left(  \frac{1}{m} \sum_{i=1}^{m} X_i^{\prime} X_i^{\prime \top} \right) u     +   2  u^{\top } \nabla \hat{L}(\theta^*)\\
&\geq & \displaystyle  u^{\top }\left(  \frac{1}{m} \sum_{i=1}^{m} X_i^{\prime} X_i^{\prime \top} \right) u     -   2  \|u \|_1  \|\nabla \hat{L}(\theta^*)\|_{\infty}\\
& \geq & \frac{C_{\min} \|u\|^2  }{2} -  2 c_2\|u\| d^{1/2}   \left\{ \frac{  (\log^{1/2} m + d^{1/2} \|\theta^*\|  )\log m  }{m} \right\}^{1/2}\\ 
&> &
0  	\end{array}
\]
for some constant  $c_2>0$, and where the last inequality follows from the choice of $B$ with a large enough $c_1$.

Therefore,
\begin{equation}
\label{eqn:parameter}
\|\hat{\theta} - \theta^*\| \,\leq\,  \frac{4 c_2 }{C_{\min}} \left\{ \frac{  d\left(  \log^{1/2} m    +   d^{1/2} \|\theta^*\|\right) \log m  }{m} \right\}^{1/2},		
\end{equation}
with probability  approaching one.
Furthermore,
\begin{equation}
\label{eqn:calculation2}
\begin{array}{lll}
\underset{ i = 1,\ldots, n}{\max}\,\left\vert  \hat{g}(X_i) -  \tilde{g}(X_i)\right\vert & = &  \underset{ i = 1,\ldots, n}{\max}\,\left\vert X_i^{\top}\hat{\theta} - X_i^{\top}  \theta^* \right\vert \\ 
& \leq &\|     \hat{\theta}  -   \theta^*\|\, \underset{ i = 1,\ldots, n,\,\,}{\max}\,  \|X_{i}\|  \\
&\leq&d^{1/2}\, \|     \hat{\theta}  -   \theta^*\|\,   \underset{ i = 1,\ldots, n,\,\, j = 1,\ldots, d }{\max}\,  \vert X_{i,j}\vert,  \\
\end{array}
\end{equation}	
and the conclusion follows from (\ref{eqn:parameter}) and the fact the $X_i's$  have compact support.
\end{proof}

\subsection{Proof of Theorem  \ref{thm:prognostic_upper_bound} }

The theorem  follows as a Theorem \ref{thm1},  proceeding  as in the proof  of Theorem \ref{thm:tv}  by  using the lemmas below.

\begin{lemma}
\label{lem14}
Let  $\tilde{\sigma}$  and $\hat{\sigma}$ as defined in (\ref{eqn:sigma})--(\ref{eqn:sigma_hat}).	There exists  a positive  constant $C_2$ such that the event
\[
\underset{i = 1,\ldots, n}{\max }\,  \vert  i - \tilde{\sigma}^{-1} \{\hat{\sigma}(i)\}  \vert   \leq  C_2 \max\{  \log n,  n \delta \},     
\]	
holds with probability  approaching one.

\end{lemma}

\begin{proof}
The claim follows as  the proof of Lemma \ref{lem5}, exploiting  Lemma \ref{lem12}.
\end{proof}

\begin{lemma}
\label{lem15}
With the notation from  (\ref{eqn:sigma_hat}),
\[
\displaystyle	  \sum_{i=1}^{n-1} \vert f_1\{g(X_{\hat{\sigma}(i)}  )\}   - f_1\{g(X_{\hat{\sigma}(i+1)}  ) \}    \vert  \,=\, O_{\mathbb{P}}\left\{   \max\{\log n, n\delta\}   \left(  \overline{\kappa}_n+1  \right)  \, \right\},
\]
and
\[
\displaystyle		  \sum_{i=1}^{n-1} \vert f_0\{g(X_{\hat{\sigma}(i)}  )  \}   - f_0\{ g(X_{\hat{\sigma}(i+1)}  )  \}   \vert  \,=\,  O_{\mathbb{P}}\left\{  \max\{\log n, n\delta\}\, \left(  \overline{\kappa}_n+1  \right)    \right\}.
\]
\end{lemma}

\begin{proof}
By Lemma \ref{lem14} and the triangle inequality, we have that, with probability  approaching one,
\[
\begin{array}{lll}
\displaystyle		  	\sum_{i=1}^{n} \vert   f_l\{ g(X_{i}  ) \}  -  f_l\{ g(X_{N(i)} )\} \vert    & \leq & \displaystyle	\sum_{i=1}^{n-1} \,\sum_{j = \max\{1, i -  C_2\max\{\log n ,  n \delta \} \}  }^{   \min\{n, i + C_2\max\{\log n ,  n \delta\}   \}  }     \vert f_l\{ g(X_{ \tilde{\sigma}(j) }  ) \}   - f_l\{ g(X_{ \tilde{\sigma}(j+1) }  ) \}   \vert \\
& \leq & \displaystyle    \left[  C_2\max\{\log n ,  n \delta \}  \right] \left[ \sum_{j=1}^{n-1} \vert f_l\{ g(X_{ \tilde{\sigma}(j) }  ) \}   - f_l\{ g(X_{ \tilde{\sigma}(j+1) }  ) \}  \vert   \right].\\
\end{array}
\]
Furthermore,
\begin{equation}
\label{eqn:tv_b2}
\begin{array}{l}
\displaystyle	  \sum_{i=1}^{n-1} \vert f_1\{ g(X_{\tilde{\sigma}(i)}  )\}   - f_1\{ g(X_{\tilde{\sigma}(i+1)}  ) \}  \vert   \
\leq    \displaystyle \sum_{i=1}^{n-1} \sum_{j= \min\{ \sigma^{-1}( \tilde{\sigma}(i) ),  \sigma^{-1}( \tilde{\sigma}(i+1) )\}  }^{ \max\{ \sigma^{-1}( \tilde{\sigma}(i) ),  \sigma^{-1}( \tilde{\sigma}(i+1) )\}   } \vert f_1\{ g(X_{\sigma(j)}  )\}   - f_1\{ g(X_{\sigma(j+1)}  ) \}  \vert\\
= \,  \displaystyle \sum_{j =1}^{n}   \left\vert \left\{ i\,:\,   j \in  [\min\{ \sigma^{-1}( \tilde{\sigma}(i) ),  \sigma^{-1}( \tilde{\sigma}(i+1) )\}, \max\{ \sigma^{-1}( \tilde{\sigma}(i) ),  \sigma^{-1}( \tilde{\sigma}(i+1) )\}]   \right\}  \right \vert \cdot\\
\,\,\,\,\, \vert f_1\{ g(X_{\sigma(j)}  )\}   - f_1\{ g(X_{\sigma(j+1)}  ) \}  \vert.\\
\end{array}
\end{equation}
However, if  $  j \in  [\min\{ \sigma^{-1}( \tilde{\sigma}(i) ),  \sigma^{-1}( \tilde{\sigma}(i+1) )\}, \max\{ \sigma^{-1}( \tilde{\sigma}(i) ),  \sigma^{-1}( \tilde{\sigma}(i+1) )\}]   $, then 
$\sigma^{-1}(  \sigma(j) )$ is between  $\min\{ \sigma^{-1}( \tilde{\sigma}(i) ), \sigma^{-1}( \tilde{\sigma}(i+1) )\}$ and  $\max\{ \sigma^{-1}( \tilde{\sigma}(i) ),  \sigma^{-1}( \tilde{\sigma}(i+1) )\}$,   $\tilde{ \sigma}^{-1}( \tilde{\sigma}(i)   ) =i$,  and  $\tilde{ \sigma}^{-1}( \tilde{\sigma}(i)   ) =i+1$. Therefore,   either  $\tilde{ \sigma}(i)  \in   \mathcal{K}_{   \sigma(j) }$ or   $\tilde{ \sigma}(i+1)  \in   \mathcal{K}_{   \sigma(j) }$.   Hence,
\[
\begin{array}{lll}
\displaystyle	  \sum_{i=1}^{n-1} \vert f_1\{ g(X_{\tilde{\sigma}(i)}  )\}   - f_1\{ g(X_{\tilde{\sigma}(i+1)}  ) \}  \vert   \,
&\leq&      \displaystyle \sum_{j =1}^{n}  \left(  1+   \kappa_{ \sigma(j) }   \right)\vert f_1\{ g(X_{\sigma(j)}  )\}   - f_1\{ g(X_{\sigma(j+1)}  ) \}  \vert \\
& \leq & \displaystyle \left(1+ \kappa\right) \sum_{j =1}^{n}  \vert f_1( g(X_{\sigma(j)}  ))   - f_1( g(X_{\sigma(j+1)}  ) ), \\
&=& O_{\mathbb{P}}(\overline{\kappa}_n +1)
\end{array}
\]
where the last inequality follows from Assumption \ref{as6.3}.  The proof  for  $f_0$   proceeds  with the same argument.

\end{proof}

\begin{lemma}
\label{lem16}
For any $i \in \{1,\ldots, n\}$ let 
\[
\xi_i\,=\, \left\vert  \{  j  \in \{1,\ldots,n\} \,:\,\,\,\,\,\text{and }\,\,\,\,\,  \tilde{g}(X_i) \leq \tilde{g}(X_j) \leq \tilde{g}(X_{N(i)}) \,\,\,\text{or}\,\,\,  \tilde{g}(X_{N(i)}) \leq \tilde{g}(X_j) \leq \tilde{g}(X_{i})    \} \right\vert.
\]
Then, for some constant $C_2 >0$,
\[
\underset{i =1 ,\ldots,n}{\max}\,\xi_i  \,\leq \, C_2 \max\{  \log n,  n \delta  \},
\]
with probability approaching one.

\end{lemma}

\begin{proof}
This  follows as the proof of  Lemma  \ref{lem8}.
\end{proof}

\begin{lemma}
\label{lem17}
With the notation from  (\ref{eqn:sigma_hat}), 
\[
\sum_{i=1}^{n} \vert   f_l\{  g(X_{i}  ) \}  -  f_l\{ g(X_{N(i)} )\} \vert   =  O_{\mathbb{P}}\left\{      \max\{\log n, n\delta\}\,  \left(\overline{\kappa}_n+1\right)    \right\},
\]
for $l \in \{0,1\}$.
\end{lemma}

\begin{proof}
By Lemma \ref{lem14} and the triangle inequality, we have that, with probability  approaching one,
\[
\begin{array}{lll}
\displaystyle		  	\sum_{i=1}^{n} \vert   f_l\{  g(X_{i}  ) \}  -  f_l\{ g(X_{N(i)} )\} \vert    & \leq & \displaystyle	\sum_{i=1}^{n-1} \,\sum_{j = \max\{1, i -  C_2\max\{\log n ,  n\delta \} \}  }^{   \min\{n, i + C_2\max\{\log n ,  n \delta \}   \}  }     \vert f_l\{ g(X_{ \tilde{\sigma}(j) }  ) \}   - f_l\{ g(X_{ \tilde{\sigma}(j+1) }  ) \}   \vert \\
& \leq & \displaystyle    C_2\max\{\log n ,  n \delta\} \sum_{j=1}^{n-1} \vert f_l\{ g(X_{ \tilde{\sigma}(j) }  ) \}   - f_l\{g(X_{ \tilde{\sigma}(j+1) }  ) \}   \vert  \\
\end{array}
\]
and the claim follows as in Lemma  \ref{lem15}.
\end{proof}

\section{Details for comparisons  with   \cite{wager2018estimation}, and \cite{athey2019generalized}  }
\label{sec:ex1.2}

\textit{Scenario 1.}
This is  the first model considered in \cite{wager2018estimation} (see Equation 27 there). The data  satisfies
\[
\begin{array}{lll}
Y_i    &= &     (1-Z_i )Y_i(0) +     Z_i Y_i(1), \\
Z_i & \sim  & \mathrm{Binom}(1,    e(X_i)      ),\\
Y_i(0) & \sim &  \mathcal{N}(  2X_i^{\top} \boldsymbol{e}_1 -1, 1   ),\\
Y_i(1) & \sim &  \mathcal{N}(  2X_i^{\top} \boldsymbol{e}_1 -1, 1   ),\\
e(x) &   = &     \frac{1}{4}(  1+   \beta_{2,4}(x_1)  )   ,   \,\,\,\forall x  \in [0,1]^d, \\
X_i  &\overset{\mathrm{ind}}{\sim} &   U[0,1]^d,\,\,\,\forall i \{1,\ldots,n\},
\end{array}
\]
where $\boldsymbol{e}_1 =(1,0,\ldots,1)^{\top}$,  and  $\beta_{2,4}$ is $\beta$-density  with shape  parameters  $2$ and $4$.  Notice  that in this case  $\tau^*_i =0 $  for  all $i  \in \{1,\ldots,n\}$.

\textit{Scenario 2.} Our second scenario also comes  from   \cite{wager2018estimation} (see Equation 28 there).
\[
\begin{array}{lll}
Y_i &=&  m(X_i) + (Z_i- e(X_i) ) \tau(X_i) +  \epsilon_i,\\
Z_i& \sim  &\mathrm{Binom}(1,    e(X_i)      ),\\
m(x) &=& e(x)\tau(x),   \,\,\,\forall x  \in [0,1]^d, \\
e(x) &=& 0.5 ,   \,\,\,\forall x  \in [0,1]^d, \\
\tau(x)&=& \varsigma(x_1)\varsigma(x_2),   \,\,\,\forall x  \in [0,1]^d, \\
\varsigma(u)&=& 1+\frac{1}{1+\exp\{-20(u-\frac{1}{3})\}},\,\,\,\forall u  \in [0,1], \\
X_i  &\overset{\mathrm{ind}}{\sim} &   U[0,1]^d,\,\,\,\forall i \{1,\ldots,n\}.
\end{array}
\]
Hence, once again  $\tau^*_i =0 $  for  all $i  \in \{1,\ldots,n\}$.

\textit{Scenario 3.}  Here we  generate the measurements as 
\[
\begin{array}{lll}
Y_i    &= &     (1-Z_i) Y_i(0) + Z_i Y_i(1), \\
Z_i & \sim  & \mathrm{Binom}(1,    e(X_i) ),\\
Y_i(l) & \sim &  \mathcal{N}( f_l(e(X_i)) , 1 ), \quad \forall l \in \{0,1\}, \\
e(x) & = & \Phi( \beta^\top x ), \quad \forall x \in [0,1]^d, \\
f_0(s) & = & s^2, \quad \forall s \in [0,1], \\
f_1(s) & = & s^2 + \mathbf{1}_{\{s > 0.6\}}, \quad \forall s \in [0,1], \\
X_i & \overset{\mathrm{iid}}{\sim} & U[0,1]^d, \quad \forall i \in \{1, \ldots, n\},
\end{array}
\]
where \( \beta \in \mathbb{R}^p \) is defined by \( \beta_j = 1 \) for \( j \in \{1, \ldots, \lfloor p/2 \rfloor\} \), and \( \beta_j = -1 \) otherwise. Furthermore, \( \Phi \) denotes the cumulative distribution function of the standard normal distribution, Clearly, in this case \( \tau^{*}_i = \boldsymbol{1}_{ \{e(X_i) > 0.6\} } \) for all \( i \in \{1, \ldots, n\} \).

\textit{Scenario 4.} This is the model  described in (\ref{eqn:scenario1}).

\section{Details of comparisons  with  \cite{abadie2018endogenous} }

\subsection{National JTPA Study}
\label{sec:ex2.2}

\vspace{-.1in}

We  follow the  experimental setting in \cite{abadie2018endogenous}. Specifically,   let  the  JTPA measurements be   $\{(z_i^{ obs},x_i^{ obs} , y_i^{ obs})  \}_{i=1}^{2530}$, where  $y_i^{ obs}  \in  \mathbb{R}$ corresponds to the outcome (earnings), $x_i^{ obs} \in   \mathbb{R}^d$ to the covariates, and  $z_i^{ obs} \in \{0,1\}$ to the  treatment indicator. Then, to generate simulated outcomes,  construct a  parameter  \(\theta\)  as in \cite{abadie2018endogenous}:
\[
\theta \,= \,\underset{{ \beta \in   \mathbb{R}^d  }  }{\arg \min } \,\sum_{i=1}^{n_{obs}} \left\{\frac{\left(y_{i}^{obs}\right)^\lambda-1 }{\lambda}-x_{i}^{obs\top}\beta\right\}^2\mathbf{1}\left\{z_i^{obs}=0, y_i^{obs}>0\right\},
\] 
where
\(\lambda = 0.3667272\).

Furthermore,  the variance of  the errors is computed as: \[
\sigma^2\,=\,\frac{1}{n_{obs}-d-1}\sum_{i=1}^{n_{obs}} \left\{\frac{\left(y_{i}^{obs}\right)^\lambda-1 }{\lambda}-x_{i}^{obs\top}\theta\right\}^2\mathbf{1}\left\{z_i^{obs}=0,y_i^{obs}>0\right\}.
\]
A third parameter of interest is: \[
\gamma \,\overset{\Delta}{=} \,  \underset{{ \beta \in    \mathbb{R}^d  }  }{\arg \max }\, \sum_{i=1}^{n_{obs}}\mathbf{1}\left\{z_{i}^{obs}=0\right\}\log\left\{\left(\frac{e^{\beta^\top x_{i}^{obs}   }}{1+e^{\beta^\top x_{i}^{obs}       }}\right)^{\mathbf{1}\left\{y_{i}^{obs}>0\right\} }\left(1-\frac{e^{  \beta^\top x_{i}^{obs}        }}{1+e^{  \beta^\top x_{i}^{obs}        }}\right)^{1-\mathbf{1}\left\{y_{i}^{obs}>0\right\} }\right\},
\] the result of fitting a logistic regression model to predict whether
a unit in the experiment's control group will have positive earnings.

Next, simulation  data is generated as:
\[
\begin{array}{lll}
Y_i &= &\left(\text{max}\left\{0,1+\lambda\left(X_i^\top\theta+\epsilon_i\right)\right\}\times\mathbf{1}\left\{U_i>0\right\} \right)^{1/\lambda}, \\
U_i&\overset{i.i.d.}{\sim}\ & \text{Bernoulli}\left(\frac{e^{X_i^\top\gamma}}{1+e^{X_i^\top\gamma}}\right),\\
X_i &\overset{iid}{\sim}\ & \mathbb{P}\left(X=x_j^{obs};\{x_j^{obs}\}_{j=1}^{n_{obs}}\right)=\frac{1}{n_{obs}} ,\,\,\,\,j \in \{1,\ldots,n_{ obs }\},  \,\,\,\,\,\,\,\,\,\,\,\,\ (\text{Empirical Distribution})\\
\epsilon_i &\overset{i.i.d.}{\sim}\ & \mathcal{N}\left(0,\sigma^2\right).\\
\end{array}
\]
Here,  the treatment effect  is zero. Furthermore, the treatment indicators for the simulations are such that
\(\sum_i Z_i=1681\) for the training set, and
\(\mathbb{P}\left(Z_i=1\right)=\frac{1681}{2530}\) for the test set.

\subsection{Project STAR}
\label{sec:ex2.3}

With the observations $\{(z_i^{ obs},x_i^{ obs} , y_i^{ obs}  )\}_{i=1}^{3764}$ for this study, where $y_i^{ obs} \in   \mathbb{R}$ is the outcome variable,  $x_i^{ obs} \in   \mathbb{R}^d$ the vector of covariates, and $z_i^{ obs} \in \{0,1\}$ the treatment assignment,  we generate data following \cite{abadie2018endogenous}. Thus,   measurements arise from the model
\[
\begin{array}{lll}
Y_i &= & X_i^\top\beta_0+\epsilon_i \\
X_i &\overset{\mathrm{ind}}{\sim}\ & \mathbb{P}\left(X=x_j^{obs};\{x_j^{obs}\}_{j=1}^{n_{obs}}\right)=\frac{1}{n_{obs}} ,\,\,\,\,j \in \{1,\ldots,n_{ obs }\}  \,\,\,\,,\,\,\,\,\,\,\,\ (\text{Empirical Distribution})\\
\epsilon_i &\overset{\mathrm{ind}}{\sim}\ & \mathcal{N}\left(0,\sigma^2\right),\\
\end{array}
\]
where 
\[
\beta_0\, = \, \underset{ \beta  \in   \mathbb{R}^d }{\arg \min}\,  \sum_{i=1}^{n_{obs}} \left(y_{i}^{obs}-x_{i}^{obs\top}\beta\right)^2\mathbf{1}\left\{z_i^{obs}=0\right\},
\]
and the variance for the errors is computed as: 
\[
\sigma^2=\frac{1}{n_{obs}-p-1}\sum_{i=1}^{n_{obs}} \left(y_{i}^{obs}-x_{i}^{obs\top}\beta_0\right)^2\mathbf{1}\left\{z_i^{obs}=0\right\}.
\]
Finally, in this scenario, the treatment indicators for the simulations are such that
\(\sum_i W_i=  \ceil{n/2}  \).

\section{National Supported Work  data}
\label{sec:nsw}

\begin{figure}[ht!]
\begin{center}
\includegraphics[width=2.58in,height=2.3in]{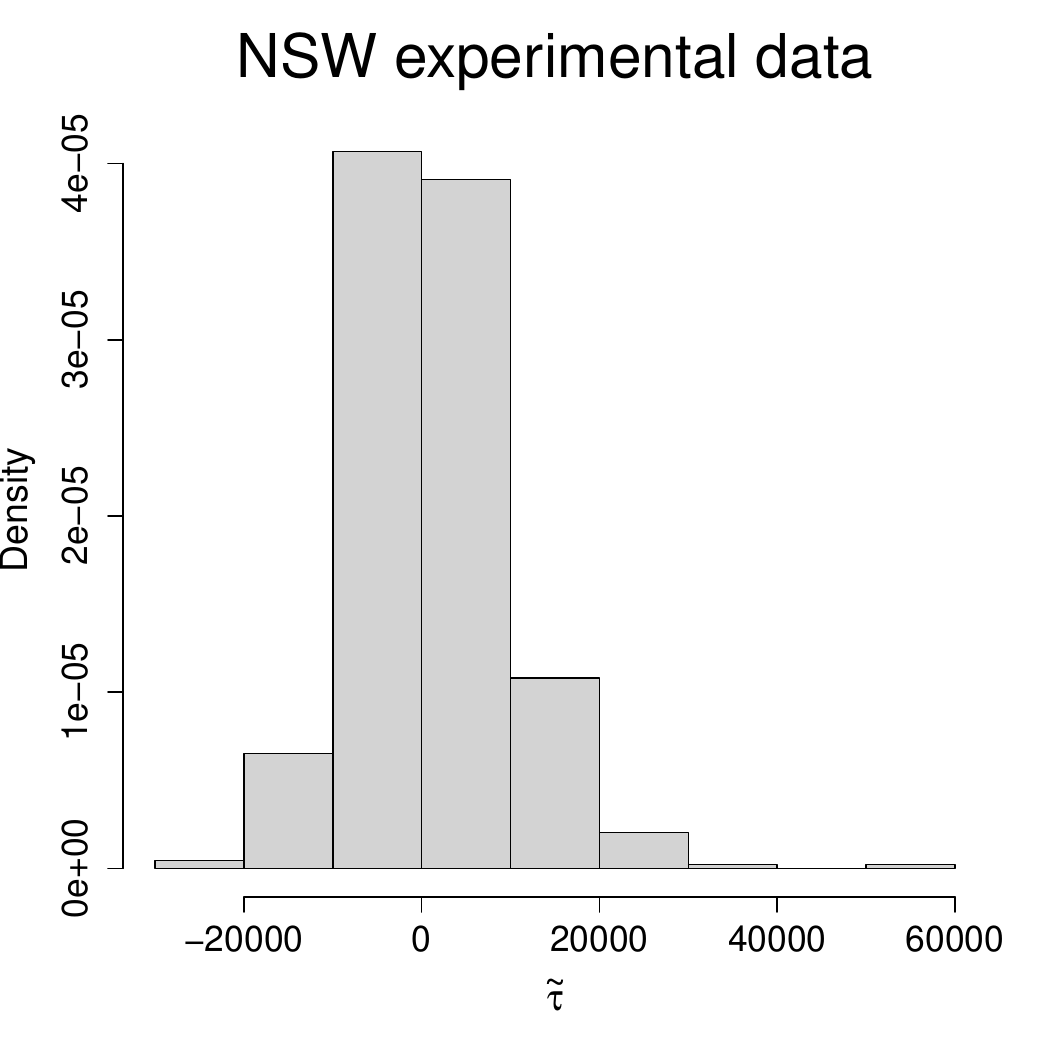} 
\includegraphics[width=2.58in,height=2.3in]{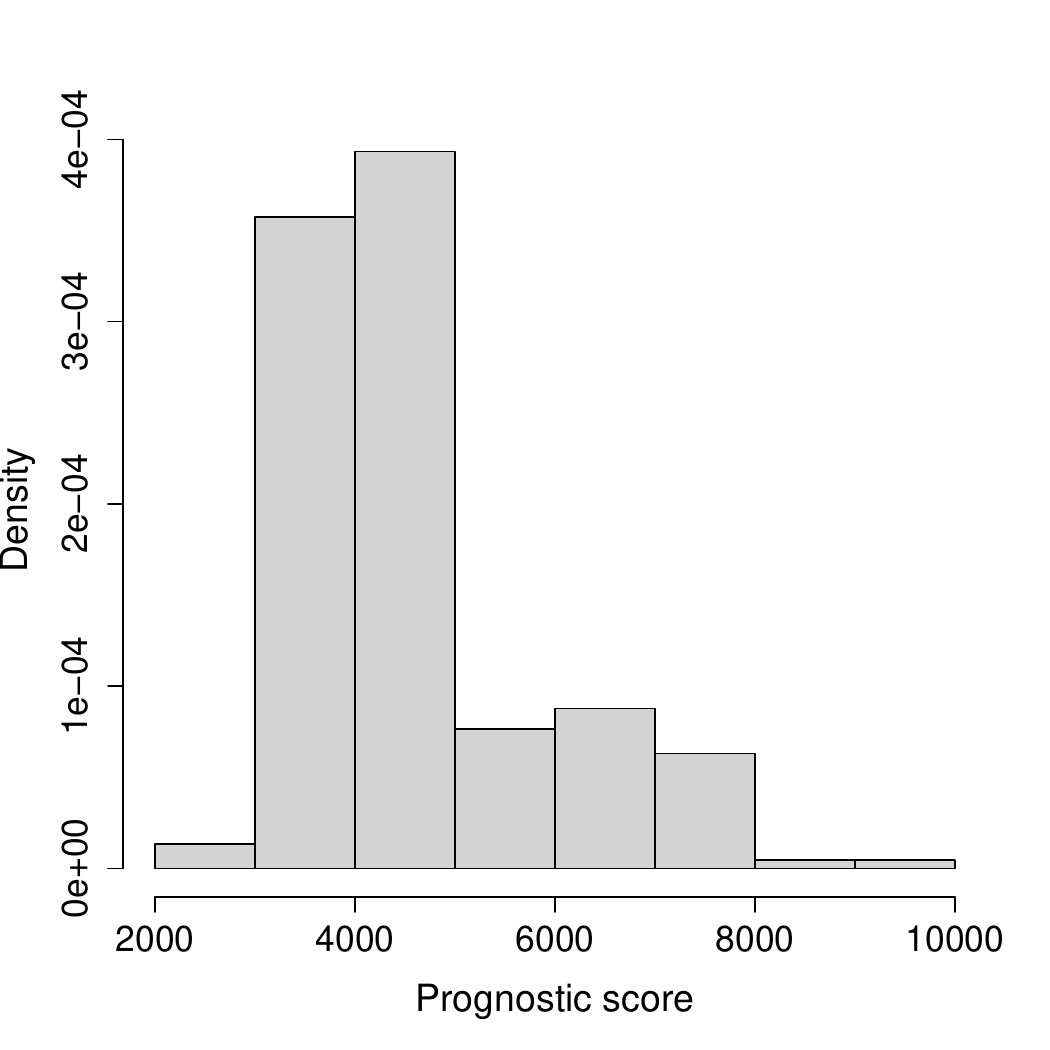} 
\caption{\label{fig6}    Data and estimates for the NSW example  described in Section \ref{sec:real1}. From left to right the two panels show  a histogram of $\tilde{\tau}^{(1)}$'s  (the scores obtained after matching)	and the  estimated prognostic scores. }
\end{center}
\end{figure}

\bibliographystyle{plainnat}
\bibliography{draft}

\end{document}